\documentclass[letterpaper, 11pt]{article}
 \usepackage{graphicx}
\usepackage{amsmath}
\usepackage{amsfonts}
\usepackage{hyperref}
\usepackage{amsthm}
 \usepackage{amscd}
\usepackage{mathrsfs}
\usepackage[top=1in,left=1in,right=1in,bottom=1in]{geometry}

\newtheorem{theorem}{Theorem}[section]
\newtheorem{proposition}[theorem]{Proposition}

\newtheorem{corollary}[theorem]{Corollary}

\newtheorem{example}{Example}[section]
\hypersetup{
	colorlinks = true,
	linkcolor  = red,
	citecolor  = green,
	filecolor  = cyan,
	urlcolor   = magenta,}

\newcommand{\ds}{\displaystyle}

\numberwithin{equation}{section}

\title{The generalized Marchenko method in the inverse scattering problem for a first-order linear system}
\author{T. Aktosun and R. Ercan\\
Department of Mathematics\\
University of Texas at Arlington\\
Arlington, TX 76019-0408, USA}

\date{}

\begin{document}

\maketitle

\begin{abstract}
The Marchenko method is developed in the inverse scattering problem for a linear system of first-order differential equations
containing potentials proportional to the spectral parameter.
The corresponding Marchenko system of integral equations is derived in such a way that
the method can be applied to some other linear systems for which a Marchenko method is
not yet available. It is shown how the potentials and the scattering solutions to the linear system
are constructed from the solution to the Marchenko system.
The bound-state information for the linear system with any number of
bound states and any multiplicities is described in terms of
a pair of constant matrix triplets.
When the potentials in the linear system are reflectionless, some explicit solution formulas
are presented in closed form for the potentials and for the scattering solutions
to the linear system. 
The theory is illustrated with some explicit examples.
\end{abstract}

{\bf {AMS Subject Classification (2020):}} 34A55, 34L25, 34L40, 47A40

{\bf Keywords:} Marchenko method, generalized Marchenko integral equation, inverse scattering, first-order linear system, 
energy-dependent potentials

\newpage

\section{Introduction}
\label{section1}

Our main goal in this paper is to develop the Marchenko method for the linear system 
\begin{equation}\label{1.1}
\ds\frac{d}{dx}\begin{bmatrix}
\alpha\\
\noalign{\medskip}
\beta
\end{bmatrix}=
\begin{bmatrix}
-i\zeta^2 & \zeta \,q(x)\\
\noalign{\medskip}
\zeta\, r(x) & i\zeta^2
\end{bmatrix}
\begin{bmatrix}
\alpha\\
\noalign{\medskip}
\beta
\end{bmatrix},\qquad -\infty<x<+\infty,
\end{equation}
where $x$ is the spacial coordinate,
$\zeta $ is the spectral parameter, the scalar quantities $q(x)$ and $r(x)$ are some complex-valued potentials,  and the column vector  
$\begin{bmatrix}
\alpha\\
\beta
\end{bmatrix}$ is the wavefunction depending on $x$ and $\zeta.$ 
We assume that the potentials $q$ and $r$ belong to the Schwartz class, i.e. the class of functions of 
$x$ on the real axis $\mathbb R$ for which the derivatives of all orders exist and all those derivatives decay faster than any 
negative power of $x$ as $x\to\pm\infty.$ Even though
our results hold for potentials satisfying weaker restrictions, in order to provide insight
into the development of the Marchenko method, for simplicity and clarity
we assume that the potentials belong to the Schwartz class.

The linear system \eqref{1.1} is 
associated with the first-order
system of nonlinear equations given by
\begin{equation}\label{1.2}
\begin{cases}
iq_t+q_{xx}-i(qrq)_x=0,\\
\noalign{\medskip}
ir_t-r_{xx}-i(rqr)_x=0,
\end{cases}
\qquad   x\in\mathbb R,\quad t>0,
\end{equation}
which is known \cite{AC1991,AS1981,KN1978,T2010} as the derivative NLS (nonlinear Schr{\"o}dinger) system or as the Kaup--Newell system. 
The derivative NLS equations have important physical applications in plasma physics, 
propagation of hydromagnetic waves 
traveling in a magnetic field, and
transmission of ultra short nonlinear pulses in optical fibers \cite{AC1991,KN1978}.
Hence, the study of \eqref{1.1} is physically relevant, and the development of the Marchenko method for
\eqref{1.1} is significant. 

We remark that our concentration in this paper is not on 
integrable nonlinear systems such as \eqref{1.2} but rather on the linear system \eqref{1.1}.
We 
present the Marchenko method for \eqref{1.1} in such a way that the
method can be applied on other linear systems and also on their discrete versions. We have already developed
\cite{AE2021}
the Marchenko method for the discrete analog of the linear system \eqref{1.1}, and hence our emphasis
in this paper is the development of the Marchenko method
for the linear system \eqref{1.1}.

A linear system of differential equations such as \eqref{1.1}, which contains the spectral parameter $\zeta$ and some potentials that are functions of
the spacial variable $x$ with sufficiently fast decay at infinity, yields a scattering scenario.
It may be possible to establish a one-to-one correspondence between the potentials 
in the linear system and an appropriate scattering data set, which usually
consists of some scattering coefficients that are functions of the spectral parameter $\zeta$
and the bound-state information related to the values of the spectral
parameter at which the linear system has square-integrable solutions.
The direct scattering problem consists of 
the determination of the scattering data set when the potentials are known.
On the other hand, the inverse scattering problem consists of
the determination of the potentials when the scattering data set is
known.

One of the most effective methods in the solution to an inverse scattering problem is the Marchenko method, originally developed by 
Vladimir Marchenko \cite{AM1963} for the half-line Schr{\"o}dinger equation 
\begin{equation*}
-\ds\frac{d^2\psi}{dx^2}+V(x)\,\psi=k^2\,\psi, \qquad 0<x<+\infty.
\end{equation*}
The Marchenko method was later extended by Faddeev \cite{F1967} to solve the inverse scattering problem for the
full-line 
Schr{\"o}dinger equation 
\begin{equation}\label{1.3}
-\ds\frac{d^2\psi}{dx^2}+V(x)\,\psi=k^2\,\psi, \qquad -\infty<x<+\infty.
\end{equation}
In the Marchenko method, the potential is recovered from the solution to a linear integral equation, usually called the
Marchenko equation, where the kernel and the nonhomogeneous term are constructed from the scattering data 
set with the help of a Fourier transformation. The Marchenko equation for \eqref{1.3} has the form 
\begin{equation}\label{1.4}
K(x,y)+\Omega(x+y)+\int_x^\infty dz\,K(x,z)\,\Omega(z+y)=0,\qquad x<y,
\end{equation}
if the scattering data set is related to the
measurements at $x=+\infty,$ and it has the form
\begin{equation}\label{1.5}
\tilde K(x,y)+\tilde\Omega(x+y)+\int_{-\infty}^{x}dz\,\tilde K(x,z)\,\tilde\Omega(z+y)=0,\qquad y<x,
\end{equation}
if the scattering data set is related to the
measurements at $x=-\infty.$ The integral kernels and the nonhomogeneous
terms in \eqref{1.4} and \eqref{1.5} are constructed from the corresponding scattering data sets, and 
the potential $V$ is obtained from the solution $K(x,y)$ to \eqref{1.4} as
\begin{equation}\label{1.6}
V(x)=-2\,\ds\frac{dK(x,x)}{dx},
\end{equation}
where $K(x,x)$ denotes the limit $K(x,x^+),$ or it is constructed from the solution $\tilde K(x,y)$ to \eqref{1.5} as
\begin{equation*}
V(x)=2\,\ds\frac{d\tilde K(x,x)}{dx},
\end{equation*}
where $\tilde K(x,x)$ denotes the limit $\tilde K(x,x^-).$

The Marchenko method is applicable to various other differential equations as well as systems
 of differential equations. For example, when applied to the AKNS system \cite{AC1991,AKNS1974}
\begin{equation}\label{1.7}
\ds\frac{d}{dx}\begin{bmatrix}
\xi\\
\noalign{\medskip}
\eta
\end{bmatrix}=
\begin{bmatrix}
-i\lambda & u(x)\\
\noalign{\medskip}
v(x) & i\lambda
\end{bmatrix}
\begin{bmatrix}
\xi\\
\noalign{\medskip}
\eta
\end{bmatrix},\qquad -\infty<x<+\infty,
\end{equation}
the corresponding Marchenko integral equation still has the form given in \eqref{1.4}, except that $K(x,y)$ and $\Omega(x+y)$ are 
now $2\times 2$ matrices. 
The nonhomogeneous term and the kernel are constructed from the scattering data in a similar manner as done for \eqref{1.3}, and 
the two potentials $u$ and $v$ 
in \eqref{1.7} are recovered from the solution to the relevant Marchenko equation by using a slight variation of \eqref{1.6}.

The Marchenko method is also applicable to various inverse scattering problems for linear difference equations such as the discrete Schr\"odinger equation 
on the half-line lattice given by
\begin{equation}\label{1.8}
-\psi_{n+1}+2\psi_n-\psi_{n-1}+V_n\,\psi_n=\lambda\,\psi_n, \qquad n\ge 1,
\end{equation}
where  $\lambda$ is the spectral parameter and the quantities
$\psi_n$ and $V_n$ denote the values of the wavefunction and the potential, respectively, at the lattice location $n.$ 
In this case, the Marchenko equation corresponding to \eqref{1.8} has the discrete form given by 
\begin{equation}\label{1.9}
K_{nm}+\Omega_{n+m}+\sum_{j=n+1}^\infty \,K_{nj}\,\Omega_{j+m}=0,\qquad n<m.
\end{equation}
The nonhomogeneous term and the kernel are still constructed from the corresponding scattering data set, and the potential 
value $V_n$ is recovered from the double-indexed solution $K_{nm}$ to \eqref{1.9} via \cite{AP2015}
\begin{equation*}
V_n=K_{(n-1)n}-K_{n(n+1)},\qquad n\ge 1,
\end{equation*}
with the understanding that $K_{01}=0.$

There are still many other inverse scattering problems described by various differential or difference equations, 
or system of  differential or difference equations, for which a Marchenko method is not yet available, and \eqref{1.1}
is one of them.
In this paper, we develop the Marchenko method for  
\eqref{1.1} and 
present the corresponding matrix-valued Marchenko integral equation in \eqref{4.40}. We note that
\eqref{4.40}
resembles \eqref{1.4}, but the integral kernel of \eqref{4.40} slightly differs from that of \eqref{1.4}. 
In \eqref{4.54} and \eqref{4.55}, we present the recovery of $q(x)$ and $r(x)$ from the solution to \eqref{4.40}.

The main result presented in this paper, i.e. the derivation of the Marchenko system for \eqref{1.1} and the recovery of the potentials $q$ and 
$r$ from the solution to that Marchenko system, is significant because not only it extends the powerful Marchenko method to \eqref{1.1} 
but it also provides a procedure that can be applied to various other inverse problems. 

In our extension of the Marchenko method to solve the inverse scattering problem for \eqref{1.1}, we use the following guidelines 
in order to refer to the extension still as the Marchenko method. First, the derived Marchenko system 
should resemble \eqref{1.4}, where the nonhomogeneous term and the kernel
should both be obtained from the scattering data for \eqref{1.1} with the help of a Fourier transform, but by allowing some minor modifications. 
Next, the potentials in \eqref{1.1} should be readily obtained from the solution to the derived Marchenko system, but by allowing some appropriate modifications. 
The same guidelines can also be used
to establish a Marchenko method for other differential and difference equations, or systems of  differential and difference equations.

Let us remark that, in the literature
related to the inverse scattering transform, some authors refer to
the Marchenko equation as the Gel'fand--Levitan--Marchenko equation, but this is a misnomer
\cite{N1980}.
The Gel'fand--Levitan integral equation \cite{AK2006,AW2006,CS1989,F1967,L1987,M1986,N1983} 
is different from the Marchenko integral equation. 
The standard Gel'fand--Levitan equation has the form 
\begin{equation}\label{1.10}
A(x,y)+G(x,y)+\ds\int_{0}^{x}dz\,A(x,z)\,G(z,y)=0,\qquad 0<y<x,
\end{equation}
where $G(x,y)$ appearing in the kernel and the nonhomogeneous term is
constructed from the spectral function of the corresponding linear system.
We note that that the integral limits in the Marchenko equation \eqref{1.4} are $x$ and $+\infty,$
whereas the integral limits in the Gel'fand--Levitan equation \eqref{1.10} are $0$ and $x.$

Our paper is organized as follows. In Section~\ref{section2} we provide the preliminaries by introducing the Jost solutions 
and the scattering coefficients
for the linear system
 \eqref{1.1}, and we present their relevant properties needed in the development of our Marchenko method. 
In Section~\ref{section3} we introduce the relevant information on the bound states for \eqref{1.1}, and
we show that the bound-state information can be presented
in a simple and elegant way for any number of bound states and any
multiplicities, and this is done by using a pair of constant matrix triplets.
In Section~\ref{section4} we present the matrix-valued
Marchenko system for \eqref{1.1}, where
the input to the Marchenko system consists of a pair of reflection coefficients
and the bound-state information.
We also show that the Marchenko system
can be written in an equivalent but uncoupled format, 
and we describe how the potentials and the Jost solutions are obtained from the solution to
the Marchenko system.
In Section~\ref{section5}, when the reflection coefficients are zero, with the most general
bound-state information expressed in terms of a pair of matrix triplets, we obtain
the closed-form solution to the Marchenko system.
This allows us to
present some explicit solution formulas for the potentials and the Jost solutions for \eqref{1.1}
expressed in closed form in terms of our matrix triplets.
In Section~\ref{section5}, we also prove a relevant restriction on the bound states for
\eqref{1.1} when the potentials $q$ and $r$ are reflectionless; namely, we prove that
the bound-state poles of the corresponding transmission coefficients must be equally 
distributed in the four quadrants of the complex $\zeta$-plane. We also prove that, for the AKNS system \eqref{1.7}, 
in the reflectionless case the bound-state poles of the corresponding transmission coefficients must be equally 
distributed in the upper and lower halves of  the complex $\lambda$-plane. 
Finally, in Section~\ref{section6}, we illustrate the theory developed in the earlier
sections, and in particular we provide
some examples of potentials and Jost solutions for \eqref{1.1}
in terms of elementary functions when the sizes of our matrix triplets are small.

\section{Preliminaries}
\label{section2}

In this section,
in order to prepare for the derivation of the Marchenko system for \eqref{1.1}, we introduce the Jost solutions and the
scattering coefficients for \eqref{1.1} and we present their relevant properties. We 
use the notation of \cite{AE2019} and rely some of the results
presented there. 

We let
$\psi(\zeta,x),$ $\bar\psi(\zeta,x),$  $\phi(\zeta,x),$ $\bar\phi(\zeta,x)$
denote the four Jost solutions to \eqref{1.1} satisfying the respective spacial asymptotics
\begin{equation}\label{2.1}
\psi(\zeta,x)=\begin{bmatrix}
o(1)\\
\noalign{\medskip}
 e^{i\zeta^2x}\left[1+o(1)\right]
\end{bmatrix} ,\qquad  x\to+\infty,
\end{equation}
\begin{equation}\label{2.2}
\bar\psi(\zeta,x)=\begin{bmatrix}
e^{-i\zeta^2x}\left[1+o(1)\right]\\
\noalign{\medskip}
o(1)
\end{bmatrix} ,\qquad  x\to+\infty,
\end{equation}
\begin{equation}\label{2.3}
\phi(\zeta,x)=\begin{bmatrix}
e^{-i\zeta^2x}\left[1+o(1)\right]\\
\noalign{\medskip}
o(1)
\end{bmatrix} ,\qquad   x\to-\infty,
\end{equation}
\begin{equation}\label{2.4}
\bar\phi(\zeta,x)=\begin{bmatrix}
o(1)\\
\noalign{\medskip}
e^{i\zeta^2x}\left[1+o(1)\right]
\end{bmatrix} ,\qquad  x\to-\infty.
\end{equation}
We remark that the overbar does not denote complex conjugation. 

There are six scattering coefficients associated with \eqref{1.1}, i.e. the transmission coefficients $T(\zeta)$ and $\bar T(\zeta),$ 
the right reflection coefficients  $R(\zeta)$ and $\bar R(\zeta),$ and the left reflection coefficients  $L(\zeta)$ and $\bar L(\zeta).$ 
Because the trace of the
coefficient matrix in \eqref{1.1} is zero, the transmission coefficients from the left and from the right are equal to each other, and 
hence we do not need to use separate notations for the left and right transmission coefficients. 
The six scattering coefficients can be defined in terms of the spacial asymptotics of the Jost solutions given by
\begin{equation}\label{2.5}
\psi(\zeta,x)=\begin{bmatrix}
\ds\frac{L(\zeta)}{T(\zeta)}\,e^{-i\zeta^2 x}\left[1+o(1)\right]\\
\noalign{\medskip}
\ds\frac{1}{T(\zeta)}\,e^{i\zeta^2 x}\left[1+o(1)\right]
\end{bmatrix}, \qquad   x\to-\infty,
\end{equation}
\begin{equation}\label{2.6}
\bar\psi(\zeta,x)=\begin{bmatrix}
\ds\frac{1}{\bar T(\zeta)}\,e^{-i\zeta^2 x}\left[1+o(1)\right]\\
\noalign{\medskip}
\ds\frac{\bar L(\zeta)}{\bar T(\zeta)}\,e^{i\zeta^2 x}\left[1+o(1)\right]
\end{bmatrix}, \qquad  x\to-\infty,
\end{equation}
\begin{equation}\label{2.7}
\phi(\zeta,x)=\begin{bmatrix}
\ds\frac{1}{T(\zeta)}\,e^{-i\zeta^2x}\left[1+o(1)\right]\\
\noalign{\medskip}
\ds\frac{R(\zeta)}{T(\zeta)}\,e^{i\zeta^2 x}\left[1+o(1)\right]
\end{bmatrix}, \qquad   x\to+\infty,
\end{equation}
\begin{equation}\label{2.8}
\bar\phi(\zeta,x)=\begin{bmatrix}
\ds\frac{\bar R(\zeta)}   {\bar T(\zeta)}\,e^{-i\zeta^2 x}\left[1+o(1)\right]\\
\noalign{\medskip}
\ds\frac{1}{\bar T(\zeta)}\,e^{i\zeta^2 x}\left[1+o(1)\right]
\end{bmatrix}, \qquad   x\to+\infty.
\end{equation}

In order to present the relevant properties of the Jost solutions, we use
the subscripts $1$ and $2$ to denote their first and second components, respectively, i.e. 
we let
\begin{equation} \label{2.9}
\begin{bmatrix}
\psi_1(\zeta,x)\\
\noalign{\medskip}\psi_2(\zeta,x)
\end{bmatrix}:=\psi(\zeta,x),\quad \begin{bmatrix}
\bar\psi_1(\zeta,x)\\ \noalign{\medskip}\bar\psi_2(\zeta,x)
\end{bmatrix}:=\bar\psi(\zeta,x),
\end{equation}
 \begin{equation} \label{2.10}
\begin{bmatrix}
\phi_1(\zeta,x)\\
\noalign{\medskip}\phi_2(\zeta,x)
\end{bmatrix}:=\phi(\zeta,x),\quad \begin{bmatrix}
\bar\phi_1(\zeta,x)\\ \noalign{\medskip}\bar\phi_2(\zeta,x)
\end{bmatrix}:=\bar\phi(\zeta,x).
\end{equation}
We relate the spectral parameter $\zeta$ appearing in
\eqref{1.1} to the parameter $\lambda$ in \eqref{1.7}
as
\begin{equation}\label{2.11}
\lambda=\zeta^{2},\quad \zeta=\sqrt{\lambda},
\end{equation}
with the square root denoting the principal branch of the complex-valued square-root function.
We use $\mathbb{C^+}$ and $\mathbb{C^-}$ to denote the upper-half and lower-half,
respectively, of the complex plane $\mathbb C,$  and we let  $\mathbb{\overline{C^+}}:=\mathbb{C^+}\cup\mathbb R$ and  
$\mathbb{\overline{C^-}}:=\mathbb{C^-}\cup\mathbb R.$

We recall that the Wronskian of any two column-vector solutions to \eqref{1.1} is defined as the determinant of the
$2\times 2$ matrix formed from those columns.
For example, the Wronskian
of $\psi(\zeta,x)$ and $\phi(\zeta,x)$ is given by
\begin{equation}\label{2.12}
[\psi;\phi]:=\left|\begin{matrix}\psi_1&\phi_1\\
\psi_2&\phi_2
\end{matrix}\right|.
\end{equation}
Due to the fact that the coefficient matrix in \eqref{1.1} has the zero trace, the value of the Wronskian
of any two solutions to \eqref{1.1} is independent of $x,$ and hence
the six scattering coefficients appearing in \eqref{2.5}--\eqref{2.8} can be 
expressed in terms of Wronskians of the Jost solutions \cite{AE2019} as
\begin{equation}\label{2.13}
T(\zeta)=\ds\frac{1}{[\phi(\zeta,x);\psi(\zeta,x)]},\quad
\bar T(\zeta)=\ds\frac{1}{[\bar\psi(\zeta,x);\bar\phi(\zeta,x)]},
\end{equation}
\begin{equation}\label{2.14}
R(\zeta)=\ds\frac{[\phi(\zeta,x);\bar\psi(\zeta,x)]}{[\psi(\zeta,x);\phi(\zeta,x)]},\quad
\bar R(\zeta)=\ds\frac{[\bar\phi(\zeta,x);\psi(\zeta,x)]}{[\bar\psi(\zeta,x);\bar\phi(\zeta,x)]},
\end{equation}
\begin{equation}\label{2.15}
L(\zeta)=\ds\frac{[\psi(\zeta,x);\bar\phi(\zeta,x)]}{[\phi(\zeta,x);\psi(\zeta,x)]},\quad
\bar L(\zeta)=\ds\frac{[\phi(\zeta,x);\bar\psi(\zeta,x)]}{[\bar\psi(\zeta,x);\bar\phi(\zeta,x)]}.
\end{equation}

It is possible to relate \eqref{1.1} to the AKNS system \eqref{1.7} by using \eqref{2.11} and 
by choosing 
the potentials $u$ and $v$ in terms of the potentials $q$ and $r$ as
\begin{equation}
\label{2.16}
u(x)=q(x)\,E(x)^{-2},
\end{equation}
\begin{equation}
\label{2.17}
v(x)=\left[-\ds\frac{i}{2}\,r'(x)+\ds\frac{1}{4}\,
q(x)\,r(x)^2\right]E(x)^2,
\end{equation}
where the prime denotes the derivative and the quantity $E(x)$ is defined as
\begin{equation}\label{2.18}
E(x):=\exp\left(\frac{i}{2}\int_{-\infty}^{x}dz\,q(z)\,r(z)\right).
\end{equation}
Since the potentials $q$ and $r$ are complex valued, we remark that in general
$E(x)$ does not have the unit modulus. From \eqref{2.18} it follows that
\begin{equation}
\label{2.19}
E(-\infty)=1, \quad E(+\infty)=e^{i\mu/2},
\end{equation}
where we have defined the complex constant $\mu$ as
\begin{equation}
\label{2.20}
\mu:=\int_{-\infty}^\infty dz\,q(z)\,r(z).
\end{equation} 

Besides \eqref{1.7}, it is also possible to relate \eqref{1.1} to another AKNS system given by
\begin{equation}\label{2.21}
\ds\frac{d}{dx}\begin{bmatrix}
\gamma\\
\noalign{\medskip}
\epsilon
\end{bmatrix}=
\begin{bmatrix}
-i\lambda &  p(x)\\
\noalign{\medskip}
s(x) & i\lambda
\end{bmatrix}
\begin{bmatrix}
\gamma\\
\noalign{\medskip}
\epsilon
\end{bmatrix}, \qquad x\in\mathbb R,
\end{equation}
by choosing the potentials $p$ and $s$ in terms of $q$ and $r$ as
\begin{equation}\label{2.22}
p(x)=\left[\frac{i}{2}\,q'(x)+\ds\frac{1}{4}\,q(x)^2\,r(x)\right]E(x)^{-2},
\end{equation}
\begin{equation}\label{2.23}
s(x)=r(x)\,E(x)^2.
\end{equation}

Let us remark that it is possible to analyze the direct and inverse scattering problems for \eqref{1.1}
without relating \eqref{1.1} to the AKNS systems \eqref{1.7} or \eqref{2.21}.
As done for \eqref{1.3} \cite{CS1989,F1967,L1987,M1986},
this can be accomplished for \eqref{1.1} by first determining the integral relations
satisfied by the four Jost solutions to \eqref{1.1}, where those integral
relations are obtained by combining \eqref{1.1} and
the asymptotic conditions \eqref{2.1}--\eqref{2.4}. Using those integral
relations, one can express the scattering coefficients for \eqref{1.1}
in terms of certain integrals involving the potentials $q$ and $r.$
The relevant properties of the scattering coefficients can be determined
from those integral relations. In a similar manner, the small and large
$\zeta$-asymptotics of the scattering coefficients, the bound states, and the inverse scattering problem for \eqref{1.1}
can all be analyzed without relating \eqref{1.1} to \eqref{1.7} or \eqref{2.21}.
On the other hand, the analysis of the direct and inverse scattering problems for \eqref{1.1},
by relating \eqref{1.1} to \eqref{1.7} or \eqref{2.21}, brings some physical insight
and intuition because the analysis of those two problems for an AKNS 
system is better understood. Note that \eqref{1.1} differs from
the AKNS systems \eqref{1.7} or \eqref{2.21} because the off-diagonal entries 
of the coefficient matrix in \eqref{1.1} contain the potentials as multiplied
by the spectral parameter $\zeta.$ This greatly complicates the analysis of
the direct and inverse scattering problems for \eqref{1.1}. On the other hand, 
the three linear systems \eqref{1.1}, \eqref{1.7}, and \eqref{2.21}
can all be viewed as different perturbations of the first-order unperturbed system
\begin{equation*}
\ds\frac{d}{dx}\begin{bmatrix}
\overset\circ\alpha\\
\noalign{\medskip}
\overset\circ\beta
\end{bmatrix}=
\begin{bmatrix}
-i\lambda & 0\\
\noalign{\medskip}
0 & i\lambda
\end{bmatrix}
\begin{bmatrix}
\overset\circ\alpha\\
\noalign{\medskip}
\overset\circ\beta
\end{bmatrix}, \qquad x\in\mathbb R,
\end{equation*}
and this helps us to understand the connections
among \eqref{1.1}, \eqref{1.7}, and \eqref{2.21}.

In  the next theorem we provide the relations among the Jost solutions to \eqref{1.1}, \eqref{1.7}, and  \eqref{2.21}, respectively,
when \eqref{2.11}, \eqref{2.16}, \eqref{2.17}, \eqref{2.22}, \eqref{2.23} hold. We omit the proof and refer the reader
to Theorems~3.1 and 3.2 of \cite{AE2019}.

\begin{theorem}
\label{theorem2.1}
Assume that the potentials $q$ and $r$ appearing in the first-order system \eqref{1.1}  belong to the Schwartz class. Let $E$ denote
the quantity $E(x)$ defined in \eqref{2.18}, and $\mu$ be the complex constant defined in \eqref{2.20}, 
and assume that the spectral
parameters $\zeta$ and $\lambda$ are related to each other as in  \eqref{2.11}. Then, we have the following:

\begin{enumerate}

\item[\text{\rm(a)}] The linear system \eqref{1.1} can be transformed into the AKNS system \eqref{1.7}, 
where the potential pair $(u,v)$ is related to $(q, r)$ as in \eqref{2.16} and \eqref{2.17}. 
It follows that the potentials $u$ and $v$ also belong to the Schwartz class. The four Jost solutions to
\eqref{1.1} appearing in \eqref{2.1}--\eqref{2.4}, respectively, and 
the four Jost solutions $\psi^{(u,v)},$ $\bar\psi^{(u,v)},$ $\phi^{(u,v)},$ $\bar\phi^{(u,v)}$ to \eqref{1.7}, satisfying
the corresponding asymptotics  in \eqref{2.1}--\eqref{2.4}, respectively, are related to each other as
\begin{equation}
\label{2.24}
\psi(\zeta,x)=e^{i\mu/2}\begin{bmatrix}
\sqrt{\lambda}\,E&0\\
\noalign{\medskip}\ds\frac{i}{2}\,r(x)\,E&E^{-1}
\end{bmatrix}\psi^{(u,v)}(\lambda,x),
\end{equation}
\begin{equation}\label{2.25}
\bar\psi(\zeta,x)=e^{-i\mu/2}
\begin{bmatrix}
E&0\\
\noalign{\medskip}
\ds\frac{i}{2\sqrt{\lambda}}\,r(x)\,E&\ds\frac{1}{\sqrt{\lambda}}\,
E^{-1}
\end{bmatrix}\bar\psi^{(u,v)}(\lambda,x),
\end{equation}
\begin{equation}
\label{2.26}
\phi(\zeta,x)=\begin{bmatrix}
E&0\\ \noalign{\medskip}
\ds\frac{i}{2\sqrt{\lambda}}\,r(x)\,E&\ds\frac{1}{\sqrt{\lambda}}\,E^{-1}
\end{bmatrix}\phi^{(u,v)}(\lambda,x),
\end{equation}
\begin{equation}\label{2.27}
\bar\phi(\zeta,x)=\begin{bmatrix}
\sqrt{\lambda}\,E&0\\
\noalign{\medskip}\ds\frac{i}{2}\,r(x)\,E&E^{-1}
\end{bmatrix}\bar\phi^{(u,v)}(\lambda,x).
\end{equation}

\item[\text{\rm(b)}]  The system \eqref{1.1} can be transformed into the system \eqref{2.21}, where the potential 
pair $(p,s)$ is related to $(q, r)$ as in \eqref{2.22} and \eqref{2.23}. It follows that the potentials $p$ and 
$s$ belong to the Schwartz class. The four Jost solutions to 
\eqref{1.1} 
and the four Jost solutions $\psi^{(p,s)},$ $\bar\psi^{(p,s)},$ $\phi^{(p,s)},$ $\bar\phi^{(p,s)}$ to \eqref{2.21}, satisfying
the corresponding asymptotics in \eqref{2.1}--\eqref{2.4}, respectively, are related to each other
as
\begin{equation}\label{2.28}
\psi(\zeta,x)=e^{i\mu/2}\begin{bmatrix}
\ds\frac{1}{\sqrt{\lambda}}\,E&-\ds\frac{i}{2\sqrt{\lambda}}\,q(x)\,E^{-1}
\\
\noalign{\medskip}0&E^{-1}
\end{bmatrix}\psi^{(p,s)}(\lambda,x),
\end{equation}
\begin{equation}
\label{2.29}
\bar\psi(\zeta,x)=e^{-i\mu/2}\begin{bmatrix}
E&-\ds\frac{i}{2}\,q(x)\,E^{-1}\\
\noalign{\medskip}
0&\sqrt{\lambda}\,E^{-1}
\end{bmatrix}\bar\psi^{(p,s)}(\lambda,x),
\end{equation}		
\begin{equation}
\label{2.30}
\phi(\zeta,x)=\begin{bmatrix}
E&-\ds\frac{i}{2}\,q(x)\,E^{-1}\\
\noalign{\medskip}0&\sqrt{\lambda}\,E^{-1}
\end{bmatrix}\phi^{(p,s)}(\lambda,x),
\end{equation}
\begin{equation}\label{2.31}
\bar\phi(\zeta,x)=\begin{bmatrix}
\ds\frac{1}{\sqrt{\lambda}}\,E&-\ds\frac{i}{2\sqrt{\lambda}}\,q(x)\,E^{-1}\\
\noalign{\medskip}0&E^{-1}
\end{bmatrix}\bar\phi^{(p,s)}(\lambda,x).
\end{equation}

\end{enumerate}
	
\end{theorem}

Next, we present the relevant analyticity and symmetry properties of the Jost solutions to \eqref{1.1}, which are needed in establishing
the Marchenko method for \eqref{1.1}.

\begin{theorem}
\label{theorem2.2}
Let the potentials $q$ and $r$ in \eqref{1.1} belong to the Schwartz class.
Assume that the spectral parameters $\zeta$ and $\lambda$ are related to each other as in
\eqref{2.11}. Then, we have the following:
	
\begin{enumerate}
		
\item[\text{\rm(a)}] For each fixed $x\in\mathbb R,$ the Jost solutions $\psi(\zeta,x)$ and $\phi(\zeta,x)$
to \eqref{1.1}  are analytic in the first and third quadrants in the complex $\zeta$-plane
and are continuous in the closures of those regions. Similarly, the Jost solutions $\bar\psi(\zeta,x)$ and
$\bar\phi(\zeta,x)$ are analytic in the second and fourth quadrants in the complex $\zeta$-plane
and are continuous in the closures of those regions. 

\item[\text{\rm(b)}]  The components of the Jost solutions appearing in \eqref{2.9} and \eqref{2.10} satisfy the following properties. 
The components $\psi_1(\zeta,x),$ $\bar\psi_2(\zeta,x),$ $\phi_2(\zeta,x),$ and  $\bar\phi_1(\zeta,x)$ 
are odd in $\zeta;$ and the components 
$\psi_2(\zeta,x),$ $\bar\psi_1(\zeta,x),$ $\phi_1(\zeta,x),$ and  $\bar\phi_2(\zeta,x)$ are even in $\zeta.$ 
Furthermore, for each fixed $x\in\mathbb R,$ the four scalar functions $\psi_1(\zeta,x)/\zeta,$ $\psi_2(\zeta,x),$ 
$\phi_1(\zeta,x),$ and $\phi_2(\zeta,x)/\zeta$ are even in 
$\zeta;$ and and hence they are analytic in $\lambda \in\mathbb C^+$ and continuous in 
$\lambda\in\overline{\mathbb C^+}.$ Similarly, for each fixed $x\in\mathbb R,$ the four scalar functions 
$\bar\psi_1(\zeta,x),$ $\bar\psi_2(\zeta,x)/\zeta,$ $\bar\phi_1(\zeta,x)/\zeta,$ and 
$\bar\phi_2(\zeta,x)$ are even in $\zeta;$ and hence they are analytic in 
$\lambda \in\mathbb C^-$ and continuous in $\lambda\in\overline{\mathbb C^-}.$

\end{enumerate}
\end{theorem}

\begin{proof}
The proof of (a) can be obtained by converting \eqref{1.1} and each of the asymptotics in \eqref{2.1}--\eqref{2.4} 
into an integral equation, then by solving the resulting four integral equations via iteration, and by expressing 
the Jost solutions as uniformly convergent infinite
series of terms that are analytic in the appropriate domains in the complex $\zeta$-plane
and are continuous in the closures of those domains. 
Alternatively, the proof of (a) can be obtained with the help of Theorem~\ref{theorem2.1} and
by using the corresponding analyticity and continuity properties \cite{AKNS1974,E2018} in $\lambda$ of the Jost solutions
to the AKNS systems \eqref{1.7} and \eqref{2.21}.
The proof of (b) is obtained by using the results in (a) and either the relations \eqref{2.24}--\eqref{2.27}
or \eqref{2.28}--\eqref{2.31}.
\end{proof}

In the following theorem, we present the small spectral asymptotics of the Jost solutions to \eqref{1.1}, which is
crucial for the establishment of the Marchenko method for \eqref{1.1}

\begin{theorem}
\label{theorem2.3}
Let the potentials $q$ and $r$ in \eqref{1.1} belong to the Schwartz 
class. 
Then, for each fixed $x\in\mathbb R,$ as $\zeta\to\ 0$ in their domains of continuity,
the Jost solutions to \eqref{1.1} appearing in \eqref{2.1}--\eqref{2.4} satisfy
\begin{equation}\label{2.32}
\psi(\zeta,x)=\begin{bmatrix}
-\zeta\ds\int_x^\infty dz\,q(z)+O\left(\zeta^3\right)
\\
\noalign{\medskip}
1+O\left(\zeta^2\right)
\end{bmatrix},
\end{equation}
\begin{equation}\label{2.33}
\bar\psi(\zeta,x)=\begin{bmatrix}
1+O\left(\zeta^2\right)
\\
\noalign{\medskip}
\zeta\ds\int_x^\infty dz\,r(z)+O\left(\zeta^3\right)
\end{bmatrix},
\end{equation}
\begin{equation}\label{2.34}
\phi(\zeta,x)=\begin{bmatrix}
1+O\left(\zeta^2\right)
\\
\noalign{\medskip}
\zeta\ds\int_{-\infty}^x dz\,r(z)+O\left(\zeta^3\right)
\end{bmatrix},
\end{equation}
\begin{equation}\label{2.35}
\bar\phi(\zeta,x)=\begin{bmatrix}
\zeta\ds\int_{-\infty}^x dz\,q(z)+O\left(\zeta^3\right)
\\
\noalign{\medskip}
1+O\left(\zeta^2\right)
\end{bmatrix}.
\end{equation}
	
\end{theorem}

\begin{proof} The domains of continuity for the Jost solutions are specified in Theorem~\ref{theorem2.2}. 
The proof for \eqref{2.32} and \eqref{2.35} can be obtained
by using  \eqref{2.24} and \eqref{2.27}, respectively, and the known small $\lambda$-asymptotics \cite{AE2019,E2018}
of the Jost solutions $\psi^{(u,v)}(\lambda,x)$ and $\bar\phi^{(u,v)}(\lambda,x)$ to \eqref{1.7},
and by taking into account the relationship between
$\zeta$ and $\lambda$ specified in \eqref{2.11}. 
Similarly, the proof for \eqref{2.33} and \eqref{2.34} can be obtained
by using \eqref{2.29} and \eqref{2.30} and the known small $\lambda$-asymptotics 
\cite{AE2019,E2018} of the Jost solutions $\bar\psi^{(p,s)}(\lambda,x)$ and $\phi^{(p,s)}(\lambda,x).$
\end{proof}

In relation to Theorem~\ref{theorem2.3}, let us remark that the 
small $\lambda$-asymptotics of the Jost solutions to \eqref{1.7} and \eqref{2.21}
expressed in terms of the quantities relevant to \eqref{1.1} can be found in
Proposition~6.1 of \cite{AE2019}.

In order to prepare for the derivation of
the Marchenko system for
\eqref{1.1}, we also need the large $\zeta$-asymptotics of the Jost solutions to \eqref{1.1}.
For convenience, in the following theorem those asymptotics are expressed in terms of $\lambda,$
which is related to $\zeta$ as in \eqref{2.11}. 

\begin{theorem}
\label{theorem2.4}
Let the potentials $q$ and $r$ in \eqref{1.1} belong to the Schwartz 
class, and let the parameter
$\lambda$ be related to the spectral parameter $\zeta$ as in \eqref{2.11}.
Then, for each fixed $x\in\mathbb R,$ as $\lambda\to\infty$ in $\overline{\mathbb C^+},$
the Jost solutions $\psi(\zeta,x)$ and $\phi(\zeta,x)$ to \eqref{1.1}
appearing in \eqref{2.1} and \eqref{2.3}, respectively, satisfy
\begin{equation}\label{2.36}
\psi(\zeta,x)=\begin{bmatrix}
\sqrt{\lambda}\,e^{i\mu/2+i\lambda x}E(x)\left[\ds\frac{q(x)\,E(x)^{-2}}{2i\lambda} +O\left(\ds\frac{1}{\lambda^2}\right)\right]
\\
\noalign{\medskip}\ds\frac{e^{i\mu/2+i\lambda x}}{E(x)}\left[1+\ds\frac{q(x)\,r(x)}{4\lambda}
-\ds\frac{1}{2i\lambda}\int_x^\infty dz\,\sigma(z)+O\left(\frac{1}{\lambda^2}\right)\right]
\end{bmatrix},
\end{equation}
\begin{equation*}
\phi(\zeta,x)=\begin{bmatrix}
\ds e^{-i\lambda x} E(x)\left[1-\ds\frac{1}{2i\lambda}\int_{-\infty}^x dz\,\sigma(z)
+O\left(\ds\frac{1}{\lambda^{2}}\right)\right]
\\
\noalign{\medskip}
\sqrt{\lambda}\,e^{-i\lambda x}
\left[\ds\frac{ir(x)\,E(x)}{2\lambda}+O\left(\ds\frac{1}{\lambda^2}\right)\right]
\end{bmatrix},
\end{equation*}
where $E(x)$ and $\mu$ are the quantities appearing in \eqref{2.18} and \eqref{2.20}, respectively, and
the complex-valued scalar quantity $\sigma(x)$ is defined as
\begin{equation}\label{2.37}
\sigma(x):=-\ds\frac{i}{2}\,q(x)\,r'(x)+\ds\frac{1}{4}\,q(x)^2\,r(x)^2.
\end{equation}
Similarly, for each fixed $x\in\mathbb R,$ 
as $\lambda\to\infty$ in $\overline{\mathbb C^-},$ the Jost solutions $\bar\psi(\zeta,x)$ and $\bar\phi(\zeta,x)$
to \eqref{1.1}
appearing in \eqref{2.2} and \eqref{2.4}, respectively, satisfy
\begin{equation}
\label{2.38}
\bar\psi(\zeta,x)=\begin{bmatrix}
\ds e^{-i\mu/2-i\lambda x} E(x)
\left[1+\ds\frac{1}{2i\lambda}\int_x^\infty dz\,\sigma(z)
+O\left(\ds\frac{1}{\lambda^{2}}\right)\right]\\
\noalign{\medskip}
\sqrt{\lambda}\,e^{-i\mu/2-i\lambda x}\left[\ds\frac{i\,r(x)\,E(x)}{2\lambda}+O\left(\ds\frac{1}{\lambda^2}\right)\right]
\end{bmatrix},
\end{equation}
\begin{equation*}
\bar\phi(\zeta,x)=\begin{bmatrix}
\sqrt{\lambda}\,e^{i\lambda x}\left[\ds\frac{q(x)\,E(x)^{-1}}{2i\lambda}+O\left(\ds\frac{1}{\lambda^2}\right)\right]\\
\noalign{\medskip}
\ds\frac{e^{i\lambda x}}{E(x)}\left[1+\ds\frac{q(x)\,r(x)}{4\lambda}
+\ds\frac{1}{2i\lambda}\int_{-\infty}^{x}dz\,\sigma(z)
+O\left(\frac{1}{\lambda^2}\right)\right]
\end{bmatrix}.
\end{equation*}
	
\end{theorem}

\begin{proof} The proof is obtained by using iteration on the integral representations of the Jost solutions aforementioned
in the proof of Theorem~\ref{theorem2.1} and by taking into consideration
of the fact that $\zeta$ is related to $\lambda$ as in \eqref{2.11}. Alternatively, the proof can be obtained
by using \eqref{2.24}--\eqref{2.27} and the known large $\lambda$-asymptotics \cite{AKNS1974,AE2019,E2018}
of the Jost solutions to \eqref{1.7}, and by taking into account the fact that the quantity
$\sigma(x)$ defined in \eqref{2.37} corresponds to the product $u(x)\,v(x)$ when
$u(x)$ and $v(x)$ are chosen as \eqref{2.16} and \eqref{2.17}, respectively.
Equivalently,  the proof can be obtained
by using \eqref{2.28}--\eqref{2.31} and the known large $\lambda$-asymptotics \cite{AKNS1974,AE2019,E2018}
of the Jost solutions to \eqref{2.21}, and by taking into consideration the fact that the quantity
$\sigma(x)$ defined in \eqref{2.37} corresponds to the product $p(x)\,s(x)$ when
$p(x)$ and $s(x)$ are chosen as \eqref{2.22} and \eqref{2.23}, respectively. 
\end{proof}

In the next theorem, in preparation for the establishment of the Marchenko method for
\eqref{1.1}, we present the relevant properties of the scattering coefficients for
\eqref{1.1}.

\begin{theorem}
\label{theorem2.5}
Assume that the potentials $q$ and $r$ in \eqref{1.1} belong to the Schwartz class.
Let $\lambda$ be related to the spectral parameter $\zeta$ as in \eqref{2.11}, and let
$\mu$ be the complex constant defined in \eqref{2.20}.
Then, the scattering coefficients $T(\zeta),$ 
$\bar T(\zeta),$ $R(\zeta),$ $\bar R(\zeta),$ $L(\zeta),$ and $\bar L(\zeta)$ appearing in \eqref{2.5}--\eqref{2.8} have the following properties:

\begin{enumerate}

\item[\text{\rm(a)}] The transmission coefficient $T(\zeta)$ is continuous
in $\zeta\in\mathbb R$ and
has a meromorphic extension  from $\zeta\in\mathbb R$ to the first and third quadrants in the complex $\zeta$-plane. 
Furthermore, $T(\zeta)$ is an even function of $\zeta,$ and hence it is a function of 
$\lambda$ in $\overline{\mathbb C^+}.$
Moreover, $T(\zeta)$ is meromorphic in $\lambda\in\mathbb C^+$ with a finite number of poles there, where the 
poles are not necessarily simple 
but have finite multiplicities. 
The large $\zeta$-asymptotics of $T(\zeta)$ expressed in $\lambda$ is given by
\begin{equation}\label{2.39}
T(\zeta)=\ds e^{-i\mu/2}\left[1+O\left(\frac{1}{\lambda}\right)\right],\qquad \lambda\to\infty  \text{\rm{ in }} 
\lambda\in\overline{\mathbb C^+}.
\end{equation}

\item[\text{\rm(b)}] The transmission coefficient $\bar T(\zeta)$ is continuous
in $\zeta\in\mathbb R$ and
has a meromorphic extension from $\zeta\in\mathbb R$ to the second and fourth quadrants in the complex 
$\zeta$-plane. Furthermore, $\bar T(\zeta)$ is an even function of $\zeta,$ and hence it is a function of
$\lambda$ in $\overline{\mathbb C^-}.$
Moreover, $\bar T(\zeta)$ is meromorphic in $\lambda\in\mathbb C^-$ with a finite number of poles, where 
the poles are not necessarily simple but have finite multiplicities. The large $\zeta$-asymptotics of $\bar T(\zeta)$ expressed in $\lambda$ is given by
\begin{equation}\label{2.40}
\bar T(\zeta)=\ds e^{i\mu/2}\left[1+O\left(\frac{1}{\lambda}\right)\right],\qquad \lambda\to\infty 
\text{\rm{ in }} \lambda\in\overline{\mathbb C^-}.
\end{equation}

\item[\text{\rm(c)}] Each of the four reflection coefficients $R(\zeta),$ $\bar R(\zeta),$ $L(\zeta),$ and $\bar L(\zeta)$ is continuous in 
$\zeta\in\mathbb R,$ is an odd function of $\zeta,$ and has the behavior $O(1/\zeta^{5/2})$ as $\zeta\to\pm\infty.$ 
Furthermore, the four function $R(\zeta)/\zeta,$ $\bar R(\zeta)/\zeta,$ $L(\zeta)/\zeta,$ $\bar L(\zeta)/\zeta$ are even 
in $\zeta;$ are continuous functions of
$\lambda\in\mathbb R;$ and expressed in $\lambda$ they behave as $O(1/\lambda^3)$ as $\lambda\to\pm\infty.$

\item[\text{\rm(d)}] The small
$\zeta$-asymptotics of the scattering coefficients $T(\zeta),$ $\bar T(\zeta),$ $R(\zeta),$ $\bar R(\zeta),$ $L(\zeta),$ 
and $\bar L(\zeta)$ are expressed in $\lambda$ as
\begin{equation}\label{2.41}
T(\zeta)=1+O(\lambda),\qquad \lambda\to 0
\text{\rm{ in }} \overline{\mathbb C^+},
\end{equation}
\begin{equation}\label{2.42}
\bar T(\zeta)=1+O(\lambda),\qquad \lambda\to 0
\text{\rm{ in }} \overline{\mathbb C^-},
\end{equation}
\begin{equation*}\
R(\zeta)=\sqrt{\lambda}\left[\ds\int_{-\infty}^{x}dz\,r(z)-\ds\int_x^\infty dz\,r(z)+O(\lambda)\right], \qquad \lambda\to 0
\text{\rm{ in }} \mathbb R,
\end{equation*}
\begin{equation*}
\bar R(\zeta)=\sqrt{\lambda}\left[ \ds\int_{-\infty}^\infty dz\,q(z)+O(\lambda)\right],\qquad \lambda\to 0
\text{\rm{ in }} \mathbb R,
\end{equation*}
\begin{equation*}
L(\zeta)=-\sqrt{\lambda}\left[\ds\int_{-\infty}^\infty dz\,q(z)+O(\lambda)\right],\qquad \lambda\to 0
\text{\rm{ in }} \mathbb R,
\end{equation*}
\begin{equation*}
\bar L(\zeta)=\sqrt{\lambda}\left[\ds\int_x^\infty dz\,r(z)-\ds\int_{-\infty}^{x}dz\,r(z)+O(\lambda)\right], \qquad \lambda\to 0
\text{\rm{ in }} \mathbb R.
\end{equation*}
		
\end{enumerate}
\end{theorem}

\begin{proof}
Since the scattering coefficients can be expressed in terms of the Wronskians of the Jost solutions
as in \eqref{2.13}--\eqref{2.15}, 
their stated properties can be established by using the properties of the Jost solutions provided
in Theorem~\ref{theorem2.1}. Alternatively, 
the proof can be obtained by using the relationships 
between the six scattering coefficients for \eqref{1.1} and the corresponding scattering coefficients for the 
two associated AKNS systems given in \eqref{1.7} and 
\eqref{2.21}, respectively, when the potential pairs $(u,v)$ and $(p,s)$ are chosen as in \eqref{2.16}, 
\eqref{2.17}, \eqref{2.22}, and \eqref{2.23}. In fact, we have
\cite{AE2019,E2018}
\begin{equation}\label{2.43}
T(\zeta)=e^{-i\mu/2}\, T^{(u,v)}(\lambda)=e^{-i\mu/2}\, T^{(p,s)}(\lambda),
\end{equation}
\begin{equation}\label{2.44}
\bar T(\zeta)=e^{i\mu/2}\, \bar T^{(u,v)}(\lambda)=e^{i\mu/2}\, \bar T^{(p,s)}(\lambda),
\end{equation}
\begin{equation}\label{2.45}
R(\zeta)=\ds\frac{e^{-i\mu}}{\sqrt{\lambda}}\, R^{(u,v)}(\lambda)=e^{-i\mu}\, \sqrt{\lambda}\,R^{(p,s)}(\lambda),
\end{equation}
\begin{equation}\label{2.46}
\bar R(\zeta)=e^{i\mu}\, \sqrt{\lambda}\ \bar R^{(u,v)}(\lambda)=\ds\frac{e^{i\mu}}{\sqrt{\lambda}}\, \bar R^{(p,s)}(\lambda),
\end{equation}
\begin{equation}\label{2.47}
L(\zeta)=\sqrt{\lambda}\, L^{(u,v)}(\lambda)=\ds\frac{1}{\sqrt{\lambda}}\, L^{(p,s)}(\lambda),
\end{equation}
\begin{equation}\label{2.48}
\bar L(\zeta)=\ds\frac{1}{\sqrt{\lambda}}\, \bar L^{(u,v)}(\lambda)=\sqrt{\lambda}\,\bar L^{(p,s)}(\lambda),
\end{equation}
where the superscripts $(u,v)$ and $(p,s)$ are used to refer to the scattering coefficients for
\eqref{1.7} and 
\eqref{2.21}, respectively.
Using \eqref{2.43}--\eqref{2.48} and the already known \cite{AKNS1974,AE2019,E2018}
properties of the scattering coefficients of the associated AKNS systems, the proof is established. 
\end{proof}

Let us now consider the question whether the scattering
coefficients for \eqref{1.1} can be determined from the knowledge of the scattering coefficients
for \eqref{1.7} or \eqref{2.21}, and vice versa. The presence of the factor $e^{i\mu/2}$ in \eqref{2.43}--\eqref{2.46} 
gives the impression that this is possible only if we know the value of $e^{i\mu/2}$ independently.
The next theorem shows that the value of $e^{i\mu/2}$ is indeed
determined by either one of the transmission coefficients for either \eqref{1.7} or \eqref{2.21}, and
hence the scattering coefficients for \eqref{1.7} and \eqref{2.21}
can be explicitly expressed in terms of the scattering coefficients for \eqref{1.1}.
Similarly, the value of $e^{i\mu/2}$ is indeed
determined by one of the transmission coefficients
for \eqref{1.1}, and hence the scattering coefficients for 
\eqref{1.7} and \eqref{2.21} can be determined from the knowledge of the scattering coefficients
for \eqref{1.1}.

\begin{theorem}
\label{theorem2.6}
Assume that the potentials $q$ and $r$ in \eqref{1.1} belong to the Schwartz class. 
Furthermore, suppose that the potential pairs $(u,v)$ and $(p,s)$ appearing in 
\eqref{1.7} and \eqref{2.21}, respectively, are related to the potential pair $(q,r)$
as in \eqref{2.16}, 
\eqref{2.17}, \eqref{2.22}, and \eqref{2.23}.
Let $\lambda$ be related to the spectral parameter $\zeta$ as in \eqref{2.11}, and let
$\mu$ be the complex constant defined in \eqref{2.20}. Then, we have the following:

\begin{enumerate}

\item[\text{\rm(a)}] The scalar constant $e^{i\mu/2}$ is uniquely determined by one of the transmission coefficients for either of
\eqref{1.7} or \eqref{2.21}. In fact, we have
\begin{equation}\label{2.49}
e^{i\mu/2}=T^{(u,v)}(0)=T^{(p,s)}(0),
\end{equation}
\begin{equation}\label{2.50}
e^{-i\mu/2}= \bar T^{(u,v)}(0)=\bar T^{(p,s)}(0),
\end{equation}
where we recall that the superscripts $(u,v)$ and $(p,s)$ are used to refer to the scattering coefficients for
\eqref{1.7} and 
\eqref{2.21}, respectively.

\item[\text{\rm(b)}] The scattering coefficients for \eqref{1.1} are uniquely determined by the scattering coefficients for either of the linear systems
\eqref{1.7} or \eqref{2.21}. In fact, we have
\begin{equation}\label{2.51}
T(\zeta)=\ds\frac{T^{(u,v)}(\lambda)}{T^{(u,v)}(0)}=\ds\frac{T^{(p,s)}(\lambda)}{T^{(p,s)}(0)},
\end{equation}
\begin{equation}\label{2.52}
\bar T(\zeta)=\ds\frac{\bar T^{(u,v)}(\lambda)}{\bar T^{(u,v)}(0)}=\ds\frac{\bar T^{(p,s)}(\lambda)}{\bar T^{(p,s)}(0)},
\end{equation}
\begin{equation}\label{2.53}
R(\zeta)=\ds\frac{R^{(u,v)}(\lambda)}{\sqrt{\lambda}\,T^{(u,v)}(0)^2} =\ds\frac{\sqrt{\lambda}\,R^{(p,s)}(\lambda)}
{T^{(p,s)}(0)^2},
\end{equation}
\begin{equation}\label{2.54}
\bar R(\zeta)=\ds\frac{\sqrt{\lambda}\,\bar R^{(u,v)}(\lambda)}{\bar T^{(u,v)}(0)^2}=\ds\frac{\bar R^{(p,s)}(\lambda)}
{\sqrt{\lambda}\, \bar T^{(p,s)}(0)^2},
\end{equation}
\begin{equation}\label{2.55}
L(\zeta)=\sqrt{\lambda}\, L^{(u,v)}(\lambda)=\ds\frac{1}{\sqrt{\lambda}}\, L^{(p,s)}(\lambda),
\end{equation}
\begin{equation}\label{2.56}
\bar L(\zeta)=\ds\frac{1}{\sqrt{\lambda}}\, \bar L^{(u,v)}(\lambda)=\sqrt{\lambda}\,\bar L^{(p,s)}(\lambda),
\end{equation}
where we remark that \eqref{2.55} and \eqref{2.56}
are the same as \eqref{2.47} and \eqref{2.48}, respectively, because the constant $e^{i\mu/2}$ does not appear
in \eqref{2.47} and \eqref{2.48} and hence the left reflection coefficients for
\eqref{1.1} are determined by
the left reflection coefficients for either of \eqref{1.7} or \eqref{2.21}
without using the value of $e^{i\mu/2}.$

\item[\text{\rm(c)}] The scalar constant $e^{i\mu/2}$ is uniquely determined by one of the transmission 
coefficients for \eqref{1.1}. Hence, the scattering coefficients for 
\eqref{1.7} and \eqref{2.21} can be determined from the knowledge of the scattering coefficients
for \eqref{1.1} by using \eqref{2.43}--\eqref{2.48}.

\end{enumerate}

\end{theorem}

\begin{proof}
From \eqref{2.41} we see that $T(0)=1,$ and hence by evaluating \eqref{2.43} at $\lambda=0$ we obtain \eqref{2.49}.
Similarly, from \eqref{2.42} we get $\bar T(0)=1,$ and hence by evaluating \eqref{2.44} at $\lambda=0$ we have \eqref{2.50}.
Thus, the proof of (a) is complete. By using the value
of $e^{i\mu/2}$ from \eqref{2.49} or \eqref{2.50} in \eqref{2.43}--\eqref{2.48}, we
obtain \eqref{2.51}--\eqref{2.56}, respectively. Thus, the proof of (b) is also complete.
Finally, from \eqref{2.39} or \eqref{2.40} we see that the value
of $e^{i\mu/2}$ is uniquely determined by one of the transmission 
coefficients for \eqref{1.1}, and hence \eqref{2.43}--\eqref{2.48}
can be used to express the scattering coefficients for 
\eqref{1.7} and \eqref{2.21} from the knowledge of the scattering
coefficients for \eqref{1.1}, which completes the proof of (c).
\end{proof}

\section{The bound states}
\label{section3}

The bound states for \eqref{1.1} correspond to square-integrable column vector solutions to \eqref{1.1}. The existence and nature of the
bound states are completely determined by the potentials $q$ and $r$ appearing in the coefficient matrix in \eqref{1.1}.
When the potentials $q$ and $r$ belong to the Schwartz class, the following are known \cite{AE2019}
about the bound states for \eqref{1.1}:

\begin{enumerate}

\item[\text{\rm(a)}]  The bound states cannot occur at any real $\zeta$ value in \eqref{1.1}. In particular, there is no bound
state at $\zeta=0.$ The bound states can only occur
at a complex value of $\zeta$ at which the transmission coefficient $T(\zeta)$ has a pole in the first or third quadrants in the complex
$\zeta$-plane or at which the transmission coefficient $\bar T(\zeta)$ has a pole in the second or the fourth quadrants. In fact, as indicated 
in Theorem~\ref{theorem2.5} the parameter
$\zeta$ appears as $\zeta^2$ in the transmission coefficients $T(\zeta)$ and $\bar T(\zeta),$ and hence the $\zeta$-values corresponding
to the bound states must be symmetrically located with respect to the origin in the complex $\zeta$-plane.

\item[\text{\rm(b)}] As seen from \eqref{2.43} and \eqref{2.44}, for the potential pairs $(u,v)$ and $(p,s)$ appearing in 
\eqref{2.16}, \eqref{2.17}, \eqref{2.22}, \eqref{2.23}, the poles of the corresponding transmission coefficients
for the linear systems \eqref{1.1}, \eqref{1.7}, and \eqref{2.21} coincide. Hence, 
the $\lambda$-values at which the bound states occurring
for \eqref{1.1}, \eqref{1.7}, and \eqref{2.21} must coincide. We recall that $\lambda$ and $\zeta$ are related to each other as in \eqref{2.11}.

\item[\text{\rm(c)}] The number of poles of $T(\zeta)$ in the upper-half  complex $\lambda$-plane is finite and we use
$\lambda_j$ to denote those poles and we use $N$ to denote their number without taking into account their multiplicities. 
Similarly, the number of poles of $\bar T(\zeta)$ 
in the lower-half  complex $\lambda$-plane is finite and we use
$\bar\lambda_j$ to denote those poles and we use $\bar N$ to denote their number without taking into account their multiplicities. 
The multiplicity of each of those poles is finite, and we use $m_j$ to denote
the multiplicity of the pole at $\lambda_j$ and use $\bar m_j$ to denote the multiplicity
of the pole at $\bar\lambda_j.$ We remark that the bound-state poles are not necessarily simple. In the literature \cite{KN1978,T2010}, 
it is often unnecessarily
assumed that the bound states are simple because the multiple poles may be difficult to deal with. However, we have an elegant 
method of handling bound states of any number and any multiplicities, and hence there is no reason to artificially assume the 
simplicity of bound states. 

\item[\text{\rm(d)}] As indicated in the previous steps, the bound-state information for \eqref{1.1} contains the sets $\{\lambda_j,m_j\}_{j=1}^N$ and
$\{\bar\lambda_j,\bar m_j\}_{j=1}^{\bar N}.$ Furthermore, for each bound state and multiplicity we must specify a norming constant. As the 
bound-state norming constants, we use
the double-indexed quantities $c_{jk}$ for $1\le j\le N$ and $0\le k\le (m_j-1)$ and
the double-indexed quantities $\bar c_{jk}$ for $1\le j\le \bar N$ and $0\le k\le (\bar m_j-1).$ 
The construction of the bound-state norming constants $c_{jk}$ 
from 
the transmission coefficient $T(\zeta)$ and the Jost solutions 
$\phi(\zeta,x)$ and $\psi(\zeta,x)$ and the
construction of the bound-state norming constants $\bar c_{jk}$ 
from 
the transmission coefficient $\bar T(\zeta)$ and the Jost solutions 
$\bar\phi(\zeta,x)$ and $\bar\psi(\zeta,x)$
are analogous to the constructions presented
for the discrete version of \eqref{1.1}, and we refer the reader to
\cite{AE2021} for the details. Such a construction involves the determination of
the double-indexed ``residues" $t_{jk}$ with
$1\le j\le N$ and $1\le k\le m_j$ and the the double-indexed ``residues" 
$\bar t_{jk}$  with
$1\le j\le \bar N$ and $1\le k\le \bar m_j,$ respectively,
by using the expansions
of the  transmission coefficients at the bound-state poles, which are given by
\begin{equation}
\label{3.1}
T(\zeta)= \ds\frac{t_{jm_j}}{(\lambda-\lambda_j)^{m_j}}+\ds\frac{t_{j(m_j-1)}}{(\lambda-\lambda_j)^{m_j-1}}+
\cdots+\ds\frac{t_{j1}}{(\lambda-\lambda_j)}+O\left(1\right),\qquad \lambda\to\lambda_j,
\end{equation}
\begin{equation}
\label{3.2}
\bar T(\zeta)=\ds\frac{\bar t_{j\bar m_j}}{(\lambda-\bar\lambda_j)^{\bar m_j}}+\ds\frac{\bar t_{j(\bar m_j-1)}}
{(\lambda-\bar\lambda_j)^{\bar m_j-1}}+
\cdots+\ds\frac{\bar t_{j1}}{(\lambda-\bar\lambda_j)}+O\left(1\right),\qquad \lambda\to\bar\lambda_j.
\end{equation}
Next, we construct the the double-indexed dependency constants $\gamma_{jk}$ with
$1\le j\le N$ and $0\le k\le (m_j-1).$
The dependency constants $\gamma_{jk}$ appear in the coefficients
when we express
at $\lambda=\lambda_j$ the value of each $d^k\phi(\zeta,x)/d\lambda^k$ for $0\le k\le (m_j-1)$ in terms
of the set of values $\{d^k\psi(\zeta,x)/d\lambda^k\}_{k=0}^{m_j-1}.$ 
We get
\begin{equation}\label{3.3}
\ds\frac{d^k \phi(\zeta_j,x)}{d\lambda^k}=\ds\sum_{l=0}^{k}\ds\binom{k}{l}\,\gamma_{j(k-l)}\,\ds\frac{d^l
\psi(\zeta_j,x)}{d\lambda^l},\qquad 0\le k \le m_j-1,
\end{equation}
where $\binom{k}{l}$ denotes the binomial coefficient.  
Note that \eqref{3.3} is obtained as follows. From the first equality of \eqref{2.13}, we have
\begin{equation}\label{3.4}
\ds\frac{1}{T(\zeta)}=[\phi(\zeta,x);\psi(\zeta,x)],
\end{equation}
where we recall that the Wronskian is defined as in \eqref{2.12}. Using \eqref{3.1} and the fact that $\zeta$ appears as $\zeta^2$ in $T(\zeta),$
from \eqref{3.4} it follows that the $\lambda$-derivatives of order $k$ for
$0\le k\le (m_j-1)$ vanish when $\lambda=\lambda_j$ or equivalently when $\zeta=\zeta_j.$
We then recursively obtain \eqref{3.3}. For the details of the procedure, we refer the reader to \cite{AE2021}.
Similarly,
the double-indexed dependency constants $\bar\gamma_{jk}$ with $1\le j\le \bar N$ and 
$0\le k\le (\bar m_j-1)$ appear  in the coefficients
when we express
at $\lambda=\bar\lambda_j$ the value of each $d^k\bar\phi(\zeta,x)/d\lambda^k$ for $0\le k\le (\bar m_j-1)$ in terms
of the set of values $\{d^k\bar\psi(\zeta,x)/d\lambda^k\}_{k=0}^{\bar m_j-1}.$ 
We have
\begin{equation}\label{3.5}
\ds\frac{d^k \bar\phi(\bar\zeta_j,x)}{d\lambda^k}=\ds\sum_{l=0}^{k}\ds\binom{k}{l}\,\bar\gamma_{j(k-l)}\,\ds\frac{d^l
\bar\psi(\bar\zeta_j,x)}{d\lambda^l},\qquad 0\le k \le \bar m_j-1.
\end{equation}
We remark that \eqref{3.5} is derived with the help of the Wronskian relation
\begin{equation}\label{3.6}
\ds\frac{1}{\bar T(\zeta)}=[\bar\psi(\zeta,x);\bar\phi(\zeta,x)],
\end{equation}
which is obtained from the second equality of \eqref{2.13}.
Using \eqref{3.2} and the fact that $\zeta$ appears as $\zeta^2$ in $\bar T(\zeta),$
from \eqref{3.6} it follows that the $\lambda$-derivatives of order $k$ for
$0\le k\le (\bar m_j-1)$ vanish when $\lambda=\bar\lambda_j$ or equivalently when $\zeta=\bar\zeta_j.$
We then recursively obtain \eqref{3.5}. 
The norming constants $c_{jk}$ are formed in an explicit manner by using the set of residues $\{t_{jk}\}_{k=1}^{m_j}$ and the set of dependency
constants $\{\gamma_{jk}\}_{k=0}^{m_j-1},$
and this procedure is explained in the proof of Theorem~\ref{theorem4.2} and it is similar to the procedure described in Theorem~15 of \cite{AE2021}.
In a similar manner, the norming constants $\bar c_{jk}$ are formed by using the set of residues $\{\bar t_{jk}\}_{k=1}^{\bar m_j}$ 
and the set of dependency
constants $\{\bar\gamma_{jk}\}_{k=0}^{\bar m_j-1}.$
Thus, we obtain the bound-state information for \eqref{1.1} consisting of the sets
\begin{equation}
\label{3.7}
\left\{\lambda_j,m_j,\{c_{jk}\}_{k=0}^{m_j-1}\right\}_{j=1}^N,\quad
\left\{\bar\lambda_j,\bar m_j,\{\bar c_{jk}\}_{k=0}^{\bar m_j-1}\right\}_{j=1}^{\bar N}.
\end{equation}
In the first two examples in Section~\ref{section6} we illustrate
the relationships connecting the norming constants to the residues and the dependency constants.

\item[\text{\rm(e)}] Let us remark that it is extremely cumbersome to use the bound-state information in the format
specified in \eqref{3.7}
unless that information is organized in an efficient format. In fact, this is the primary reason why it is artificially assumed in the literature
that the bound states are simple. The bound-state information given in \eqref{3.7} can be organized in an efficient and elegant
manner by introducing a pair of matrix triplets $(A,B,C)$ and $(\bar A,\bar B,\bar C)$
in such a way that the specification of the matrix triplet pair is equivalent to the
specification of the bound-state information in \eqref{3.7}.
Furthermore, in the Marchenko method, the bound-state information is easily and in an elegant manner 
incorporated in the nonhomogeneous term and in the integral kernel in the corresponding Marchenko system
when it is incorporated in the form of matrix triplets.
The use of the matrix triplets enables us to deal with any number of bound states and any number of multiplicities in a simple 
and elegant manner,
as if we only have one bound state of multiplicity one. Let us remark that
the use of the matrix triplets is not confined to any particular
linear system, but it can be used on any linear system for which a Marchenko method is available. In fact, this is
one of the reasons why we are interested in establishing the Marchenko method for the linear system given in \eqref{1.1}.

\item[\text{\rm(f)}] Without loss of any generality, the matrix triplets $(A,B,C)$ and $(\bar A,\bar B,\bar C)$
can be chosen as the minimal special triplets described later in this section.  
We refer the reader to \cite{ADV2007,BGK1979} for the description of
the minimality. The minimality amounts to choosing each of the square matrices $A$ and $\bar A$
with the smallest sizes by removing any zero columns or zero rows. By the special triplets, we mean 
choosing the matrices $A$ and $\bar A$ in their Jordan canonical forms and choosing the 
column vectors $B$ and $\bar B$ in the special forms consisting of zeros and ones, as described in \eqref{3.9}, \eqref{3.11}, 
\eqref{3.14}, and \eqref{3.17}.
The choice of the special forms for the matrix triplets is unique up to the permutations of the corresponding Jordan blocks.
We refer the reader to Theorem~3.1
of \cite{ADV2007} for the details and for the proof why there is no loss of generality in using the matrix triplets
in their minimal special forms.

\end{enumerate}

Next, we show how to convert the bound-state information given in \eqref{3.7} into the matrix triplet pair
$(A,B,C)$ and $(\bar A,\bar B,\bar C).$ Since there is no loss of
generality in choosing the matrix triplets in their special forms,
we only deal with those special forms.
For simplicity and clarity, we outline the main steps of the procedure by omitting the
details. We refer the reader to \cite{AE2021} where the details of the procedure are presented for the discrete version of
\eqref{1.1}. The steps presented in \cite{AE2021} are general enough to apply to \eqref{1.1} and other linear systems.
Let us also remark that for linear systems for which the potentials appear in diagonal blocks in the corresponding
coefficient matrix, only one matrix triplet $(A,B,C)$ is needed. On the other hand, for linear systems
for which the potentials appear in off-diagonal blocks in the corresponding
coefficient matrix, a pair of matrix triplets $(A,B,C)$ and $(\bar A,\bar B,\bar C)$ is used.
The potentials $q$ and $r$ appear in the
off-diagonal entries in the coefficient matrix in \eqref{1.1}, and hence we convert the bound-state information 
into the format consisting of the triplets $(A,B,C)$ and $(\bar A,\bar B,\bar C).$ 
For the use of matrix triplets for some other linear systems, we refer the reader to
\cite{ABDV2010,ADV2007,ADV2010,AV2006,B2008,B2017}.

The conversion of the bound-state information from \eqref{3.7} to the matrix triplet pair $(A,B,C)$ and $(\bar A,\bar B,\bar C)$ 
involves the following steps:

\begin{enumerate}

\item[\text{\rm(a)}]  For each bound state at $\lambda=\lambda_j$ with $1\le j\le N,$ we form the matrix subtriplet $(A_j,B_j,C_j)$
as
\begin{equation}\label{3.8}
A_j:=\begin{bmatrix}
\lambda_j&1&0&\cdots&0&0\\
0&\lambda_j&1&\cdots&0&0\\
0&0&\lambda_j&\cdots&0&0\\
\vdots&\vdots&\vdots&\ddots&\vdots&\vdots\\
0&0&0&\cdots&\lambda_j&1\\
0&0&0&\dots&0&\lambda_j
\end{bmatrix},
\end{equation}
\begin{equation}\label{3.9}
B_j:=\begin{bmatrix}
0\\ \vdots \\
0\\
1
\end{bmatrix},\quad C_j:=\begin{bmatrix}
c_{j(m_j-1)}&c_{j(m_j-2)}&\cdots&c_{j1}&c_{j0}
\end{bmatrix},
\end{equation}
where $A_j$ is the $m_j\times m_j$ square matrix in the Jordan canonical form with
$\lambda_j$ appearing in the diagonal entries, $B_j$ is the column vector with $m_j$ components
that are all zero except for the last entry which is $1,$ and $C_j$ is the row vector with $m_j$ components
containing all the norming constants in the indicated order. Note that if the bound state at $\lambda=\lambda_j$ is 
simple, then we have 
\begin{equation*}
A_j=\begin{bmatrix}\lambda_j\end{bmatrix},\quad B_j=\begin{bmatrix}1\end{bmatrix},\quad
C_j=\begin{bmatrix}c_{j0}\end{bmatrix}.
\end{equation*}
Similarly, for each bound state at $\lambda=\bar\lambda_j$ with $1\le j\le \bar N$ we form the matrix subtriplet $(\bar A_j,\bar B_j,\bar C_j)$
as
\begin{equation}\label{3.10}
\bar A_j:=\begin{bmatrix}
\bar\lambda_j&1&0&\cdots&0&0\\
0&\bar\lambda_j&1&\cdots&0&0\\
0&0&\bar\lambda_j&\cdots&0&0\\
\vdots&\vdots&\vdots&\ddots&\vdots&\vdots\\
0&0&0&\cdots&\bar\lambda_j&1\\
0&0&0&\dots&0&\bar\lambda_j
\end{bmatrix},
\end{equation}
\begin{equation}\label{3.11}
\bar B_j:=\begin{bmatrix}
0\\ \vdots \\
0\\
1
\end{bmatrix},\quad \bar C_j:=\begin{bmatrix}
\bar c_{j(\bar m_j-1)}&\bar c_{j(\bar m_j-2)}&\cdots&\bar c_{j1}&\bar c_{j0}
\end{bmatrix},
\end{equation}
where $\bar A_j$ is the $\bar m_j\times\bar m_j$ square matrix in the Jordan canonical form with
$\bar\lambda_j$ appearing in the diagonal entries, $\bar B_j$ is the column vector with $\bar m_j$ components
that are all zero except for the last entry which is $1,$ and $\bar C_j$ is the row vector with $\bar m_j$ components
containing all the norming constants in the indicated order.

\item[\text{\rm(b)}]  Using $A_j$ with $1\le j\le N,$
 we form the $\mathcal N\times\mathcal N$ block-diagonal matrix $A$ as
\begin{equation}\label{3.12}
A:=\begin{bmatrix}
A_1&0&\cdots&0&0\\
0&A_2&\cdots&0&0\\
\vdots&\vdots&\ddots&\vdots&\vdots\\
0&0&\cdots&A_{N-1}&0\\
0&0&\cdots&0&A_N
\end{bmatrix},
\end{equation}
where $\mathcal N$ is defined as
\begin{equation}\label{3.13}
\mathcal N:=\ds\sum_{j=1}^N m_j,
\end{equation}
and it represents the number of bound-state poles in the upper-half
complex $\lambda$-plane by including the multiplicities. We also form the column vector $B$ with $\mathcal N$ components
and the row vector $C$ with $\mathcal N$ components as
\begin{equation}\label{3.14}
B=\begin{bmatrix}
B_1\\
B_2\\
\vdots\\
B_N
\end{bmatrix},\quad 
C:=\begin{bmatrix}
C_1&C_2&\cdots&C_N
\end{bmatrix}.
\end{equation}
Similarly, we define $\bar{\mathcal N}$ as
\begin{equation}\label{3.15}
\bar{\mathcal N}:=\ds\sum_{j=1}^{\bar N} \bar m_j,
\end{equation}
which represents the number of bound-state poles in the lower-half
complex $\lambda$-plane by including the multiplicities. We then
use $\bar A_j$ with $1\le j\le \bar N$
in order to form the $\bar{\mathcal N}\times\bar{\mathcal N}$ block-diagonal matrix $\bar A$ as
\begin{equation}\label{3.16}
\bar A:=\begin{bmatrix}
\bar A_1&0&\cdots&0&0\\
0&\bar A_2&\cdots&0&0\\
\vdots&\vdots&\ddots&\vdots&\vdots\\
0&0&\cdots&\bar A_{N-1}&0\\
0&0&\cdots&0&\bar A_N
\end{bmatrix}.
\end{equation}
We also form 
the column vector $\bar B$ with $\bar{\mathcal N}$ components
and the row vector $\bar C$ with $\bar{\mathcal N}$ components as
\begin{equation}\label{3.17}
\bar B=\begin{bmatrix}
\bar B_1\\
\bar B_2\\
\vdots\\
\bar B_N
\end{bmatrix},\quad 
\bar C:=\begin{bmatrix}
\bar C_1&\bar C_2&\cdots&\bar C_N
\end{bmatrix}.
\end{equation}

\end{enumerate}

\section{The Marchenko method}
\label{section4}

In this section we develop the Marchenko method for \eqref{1.1} by deriving the corresponding Marchenko system 
of linear integral equations
and also by showing how the Jost solutions and the potentials are recovered from the solution to that Marchenko system. 
We present the derivation of the Marchenko system in such a way that the method can be applied to other 
linear systems and to their discrete analogs. For the simplicity of the presentation, we first provide the derivation in the absence of bound states, 
and then we indicate the main modification needed to include the bound-state information in the Marchenko system.

In the following we outline the basic steps in the development of our Marchenko method for \eqref{1.1} in order to show the
similarities and differences with the development of the standard Marchenko method:

\begin{enumerate}

\item[\text{\rm(a)}] We start with the Riemann--Hilbert problem for \eqref{1.1}
by expressing the two Jost solutions $\phi(\zeta,x)$ and $\bar\phi(\zeta,x)$ as a linear combination of the Jost solutions 
$\psi(\zeta,x)$ and $\bar\psi(\zeta,x)$. This eventually yields the Marchenko system for
\eqref{1.1} with $x<y<+\infty$ as an analog of \eqref{1.4}. Note that this is also the step used in the 
derivation of the standard Marchenko method.
In order to derive the Marchenko system for \eqref{1.1} with $-\infty<y<x$ as an analog of
\eqref{1.5}, we need to express the Jost solutions $\psi(\zeta,x)$ and $\bar\psi(\zeta,x)$ as a linear combination of the Jost 
solutions $\phi(\zeta,x)$ and $\bar\phi(\zeta,x).$ However, we will only present the derivation of the former Marchenko system
and hence only deal with the Riemann--Hilbert problem for the former case.
We remark that the coefficients in the Riemann--Hilbert problem associated with the Marchenko system with $x<y<+\infty$ 
are directly related to the scattering coefficients  $T(\zeta),$ $\bar T(\zeta),$ $R(\zeta),$ and $\bar R(\zeta),$ 
and  the coefficients in the Riemann--Hilbert problem associated with the Marchenko system with $-\infty<y<x$
are directly related to the scattering coefficients  $T(\zeta),$ $\bar T(\zeta),$ $L(\zeta),$ and $\bar L(\zeta).$

\item[\text{\rm(b)}] Next, we combine the two column-vector equations arising in 
the formulation of the Riemann--Hilbert problem into a $2\times 2$ matrix-valued system. This step is also 
used in the 
development of the standard Marchenko method. 
	
\item[\text{\rm(c)}] We slightly modify our $2\times 2$ matrix-valued system obtained in the previous step. This modification
is not needed in the development of the standard Marchenko method.
The modification involving the diagonal entries is carried out in order to take into account the large $\zeta$-asymptotics 
of the Jost solutions. The modification involving the off-diagonal entries is carried out in order to 
formulate the $2\times 2$ matrix-valued Riemann--Hilbert problem in the spectral parameter $\lambda$ 
rather than in $\zeta,$ where $\lambda$ and $\zeta$ are related to each other as in \eqref{2.11}.
	
\item[\text{\rm(d)}] With the modification described in the previous step, 
we are able to take the Fourier transform from the $\lambda$-space to the $y$-space. This yields the $2\times 2$ coupled 
Marchenko system. This step is also used in the 
development of the standard Marchenko method. 
	
\item[\text{\rm(e)}] We uncouple the $2\times 2$ matrix-valued Marchenko system and obtain the 
associated uncoupled scalar Marchenko integral equations. This is also
the step used in the 
development of the standard Marchenko method.
	
\item[\text{\rm(f)}] With the help of the inverse Fourier transform,
we show how the Jost solutions to \eqref{1.1} are constructed from the solution to the Marchenko system.
This is also the step used in the development of the standard Marchenko method.
	
\item[\text{\rm(g)}] Finally, we describe how the potentials $q$ and $r$ 
appearing in \eqref{1.1} are recovered from the solution to our Marchenko system.
This step is slightly more involved than the step used in the development of the standard Marchenko method.
However, the formulas for the potentials are explicit in terms of the 
solution to our Marchenko system.

\end{enumerate}

In the next theorem we introduce the $2\times 2$ matrix-valued 
Marchenko integral system for \eqref{1.1} in the absence of bound states.

\begin{theorem}
\label{theorem4.1}
Let the potentials $q$ and $r$ in \eqref{1.1} belong to the Schwartz class,
and assume that there are no bound states. Then,
the corresponding Marchenko system for \eqref{1.1} is given by
\begin{equation}\label{4.1}
\begin{split}
\begin{bmatrix}
0&0\\ 
\noalign{\medskip}
0&0
\end{bmatrix}=&\begin{bmatrix}
\bar K_1(x,y)&K_1(x,y)\\ \noalign{\medskip}\bar K_2(x,y)&K_2(x,y)
\end{bmatrix}+ \begin{bmatrix}
0&\hat{\bar R}(x+y)\\ 
\noalign{\medskip}
\hat R(x+y)&0
\end{bmatrix}\\
\noalign{\medskip}
&+\ds\int_x^\infty dz \begin{bmatrix}
-i K_1(x,z)\,\hat R'(z+y)&\bar K_1(x,z)\,\hat{\bar R}(z+y)\\ 
\noalign{\medskip}
K_2(x,z)\,\hat R(z+y)&i\bar K_2(x,z)\,\hat{\bar R}'(z+y)
\end{bmatrix},\qquad x<y,
\end{split}
\end{equation}
where $\hat R(y)$ and $\hat{\bar R}(y)$ are related to the reflection coefficients $R(\zeta)$ and $\bar R(\zeta)$ for \eqref{1.1} 
via the Fourier transforms given by
\begin{equation}\label{4.2}
\hat R(y):=\ds\frac{1}{2\pi}\ds\int_{-\infty}^\infty  
d\lambda\,\ds\frac{R(\zeta)}{\zeta}\,e^{i\lambda y},\quad \hat{\bar R}(y):=\ds\frac{1}{2\pi}
\ds\int_{-\infty}^\infty  d\lambda\,\ds\frac{\bar R(\zeta)}{\zeta}\,e^{-i\lambda y},
\end{equation}
with $\hat R'(y)$ and $\hat{\bar R}'(y)$ denoting the derivatives of $\hat R(y)$ and $\hat{\bar R}(y),$ respectively, and
$\lambda$ being related to $\zeta$ as in \eqref{2.11}. We also have 
\begin{equation}\label{4.3}
K_1(x,y):= 
\ds\frac{1}{2\pi }\int_{-\infty}^\infty d\lambda \left[\ds\frac{e^{i\mu/2}\,\,\psi_1(\zeta,x)}{\zeta\,E(x)}\right] e^{-i\lambda y},
\end{equation}
\begin{equation}\label{4.4}
K_2(x,y):= 
\ds\frac{1}{2\pi }\int_{-\infty}^\infty d\lambda \left[e^{-i\mu/2}\,E(x)\,\psi_2(\zeta,x)-e^{i\lambda x}\right] e^{-i\lambda y},
\end{equation}
\begin{equation}\label{4.5}
\bar K_1(x,y):= 
\ds\frac{1}{2\pi }\int_{-\infty}^\infty 
d\lambda \left[\ds\frac{e^{i\mu/2}\,\bar{\psi}_1(\zeta,x)}{E(x)}-e^{-i\lambda x}\right] e^{i\lambda y},
\end{equation}
\begin{equation}\label{4.6}
\bar K_2(x,y):= 
\ds\frac{1}{2\pi }\int_{-\infty}^\infty d\lambda
\left[\ds\frac{e^{-i\mu/2}\,E(x)\,\bar{\psi}_2(\zeta,x)}{\zeta}\right] e^{i\lambda y},
\end{equation}
with $E(x)$ and $\mu$ being the quantities defined in \eqref{2.18} and \eqref{2.20}, 
respectively, and $\psi_1(\zeta,x),$ $\psi_2(\zeta,x),$ $\bar{\psi}_1(\zeta,x),$ and $\bar{\psi}_2(\zeta,x)$ are the 
components of the Jost solutions given in \eqref{2.9}.
\end{theorem}

\begin{proof}
For notational simplicity, we suppress the arguments and write
$\psi$ for $\psi(\zeta,x),$ $\bar\psi$ for $\bar\psi(\zeta,x),$ $\phi$ for $\phi(\zeta,x),$ 
$\bar\phi$ for $\bar\phi(\zeta,x),$ $T$ for $T(\zeta),$
$\bar T$ for $\bar T(\zeta),$
$R$ for $R(\zeta),$ $\bar R$ for $\bar R(\zeta),$ and $E$ for $E(x).$ From \eqref{2.1} and \eqref{2.2} we see that
the columns of the  Jost solutions $\psi$ and $\bar\psi$ to \eqref{1.1} are linearly independent, and hence those four columns form
a fundamental set of column-vector solutions to \eqref{1.1}.
Thus, each of the other
two Jost solutions $\phi$ and $\bar\phi$ can be expressed as linear combinations of
$\psi$ and $\bar\psi.$ With the help of \eqref{2.1}, \eqref{2.2}, \eqref{2.7}, and \eqref{2.8}, for $\zeta\in\mathbb R$ we obtain
\begin{equation}\label{4.7}
\begin{cases}
\phi=\ds\frac{1}{T}\,\bar{ \psi}+\ds\frac{R}{T}\,\psi,
\\ \noalign{\medskip}
\bar\phi=\ds\frac{\bar R}{\bar T}\,\bar{\psi}+\ds\frac{1}{\bar T}\,\psi,
\end{cases}
\end{equation}
or equivalently
\begin{equation}\label{4.8}
\begin{cases}
T\,\phi=\bar{ \psi}+R\,\psi,
\\ \noalign{\medskip}
\bar T\,\bar\phi=\bar R\,\bar{ \psi}+\psi,
\end{cases}
\end{equation}
which forms our Riemann--Hilbert problem consisting of the construction of the Jost solutions from 
the knowledge of $T,$ $\bar T,$ $R,$ and $\bar R.$
Let us now derive our Marchenko system starting from \eqref{4.8}.
We first combine the two column-vector equations in \eqref{4.8} 
and obtain the $2\times 2$ matrix-valued system
\begin{equation}\label{4.9}
\begin{bmatrix}
T\,\phi&\bar T\,\bar\phi
\end{bmatrix}=\begin{bmatrix}
\bar\psi&\psi
\end{bmatrix}+\begin{bmatrix}
R\,\psi&\bar R\,\bar\psi
\end{bmatrix}.
\end{equation}
Using \eqref{2.9} and \eqref{2.10}, we write \eqref{4.9} as
\begin{equation}\label{4.10}
\begin{bmatrix}
T\,\phi_1&\bar T\,\bar\phi_1\\
\noalign{\medskip}
T\,\phi_2&\bar T\,\bar\phi_2
\end{bmatrix}=\begin{bmatrix}
\bar\psi_1&\psi_1
\\ \noalign{\medskip}
\bar\psi_2&\psi_2
\end{bmatrix}+\begin{bmatrix}
R\,\psi_1&\bar R\,\bar\psi_1
\\ \noalign{\medskip}
R\,\psi_2&\bar R\,\bar\psi_2
\end{bmatrix}.
\end{equation}
We first postmultiply \eqref{4.10} with the diagonal
matrix $\text{\rm{diag}}\{
e^{i\mu/2} E^{-1},
e^{-i\mu/2}E\}$
and then divide  by $\zeta$ the off-diagonal entries in the resulting matrix-valued system.
From the resulting
$2\times 2$ matrix-valued equation, we subtract the diagonal matrix
$\text{\rm{diag}}\{
e^{-i\lambda x},
e^{i\lambda x}\}$
from both sides, and we obtain
\begin{equation}\label{4.11}
\begin{split}
&\begin{bmatrix}
e^{i\mu/2}\,E^{-1}\,T\,\phi_1-e^{-i\lambda x}&\ds\frac{1}{\zeta}e^{i\mu/2}\,E^{-1}\bar T\,\bar{ \phi}_1\\
\noalign{\medskip}
\ds\frac{1}{\zeta}e^{-i\mu/2}\,E\,T\,\phi_2&e^{-i\mu/2}\,E\,\bar T\,\bar{ \phi}_2-e^{i\lambda x}
\end{bmatrix}
\\
&\phantom{xx}
=
\begin{bmatrix}
e^{i\mu/2}\,E^{-1}\,\bar{\psi}_1-e^{-i\lambda x}&\ds\frac{1}{\zeta}e^{i\mu/2}\,E^{-1}\psi_1
\\ \noalign{\medskip}
\ds\frac{1}{\zeta}e^{-i\mu/2}\,E\,\bar{\psi}_2&e^{-i\mu/2}\,E\,\psi_2-e^{i\lambda x}
\end{bmatrix}
+\begin{bmatrix}
e^{i\mu/2}\,E^{-1}\,R\,\psi_1&\ds\frac{1}{\zeta}e^{i\mu/2}\,E^{-1}\,\bar R\,\bar{\psi}_1
\\ \noalign{\medskip}
\ds\frac{1}{\zeta}e^{-i\mu/2}\,E\,R\,\psi_2&e^{-i\mu/2}\,E\,\bar R\,\bar{\psi}_2
\end{bmatrix}.
\end{split}
\end{equation}
We now take the Fourier transform of \eqref{4.11} 
with $\int_{-\infty}^\infty d\lambda\,e^{i\lambda y}/2\pi$  in the first columns and 
with $\int_{-\infty}^\infty d\lambda\,e^{-i\lambda y}/2\pi$ in the second columns. This yields
the $2\times 2$ matrix-valued equation
\begin{equation}
\label{4.12}
\text{\rm{LHS}}=K(x,y)+\text{\rm{RHS}},
\end{equation}
where we have defined
\begin{equation}
\label{4.13}
K(x,y):=\begin{bmatrix}
\bar K_1(x,y)&K_1(x,y)
\\
\noalign{\medskip}
\bar K_2(x,y)&K_2(x,y)
\end{bmatrix},
\end{equation}
with the entries $K_1(x,y),$ $K_2(x,y),$ $K_1(x,y),$ and $K_2(x,y)$ are as in
\eqref{4.3}--\eqref{4.6}, respectively, and
\begin{equation}
\label{4.14}
\text{\rm{LHS}}:=\begin{bmatrix}
\text{\rm{LHS}}_{11}&\text{\rm{LHS}}_{12}
\\
\noalign{\medskip}
\text{\rm{LHS}}_{21}&\text{\rm{LHS}}_{22}
\end{bmatrix},
\end{equation}
\begin{equation}
\label{4.15}
\text{\rm{RHS}}:=\begin{bmatrix}
\text{\rm{RHS}}_{11}&\text{\rm{RHS}}_{12}
\\
\noalign{\medskip}
\text{\rm{RHS}}_{21}&\text{\rm{RHS}}_{22}
\end{bmatrix},
\end{equation}
with the matrix entries defined as
\begin{equation}
\label{4.16}
\text{\rm{LHS}}_{11}:=\ds\int_{-\infty}^\infty \frac{d\lambda}{2\pi}\left[e^{i\mu/2}E^{-1}T\phi_1-e^{-i\lambda x}\right]e^{i\lambda y},
\end{equation}
\begin{equation}
\label{4.17}
\text{\rm{LHS}}_{12}:=
\ds\int_{-\infty}^\infty \frac{d\lambda}{2\pi}\,e^{i\mu/2}\,E^{-1}\bar T\,\frac{\bar\phi_1}{\zeta}\,e^{-i\lambda y},
\end{equation}
\begin{equation}
\label{4.18}
\text{\rm{LHS}}_{21}:=\ds\int_{-\infty}^\infty \ds\frac{d\lambda}{2\pi}\,e^{-i\mu/2}\,E\,T\,\ds\frac{\phi_2}{\zeta}\,e^{i\lambda y},
\end{equation}
\begin{equation}
\label{4.19}
\text{\rm{LHS}}_{22}:=
\ds\int_{-\infty}^\infty \frac{d\lambda}{2\pi}\left[e^{-i\mu/2}\,E\,\bar T\,\bar{ \phi}_2-e^{i\lambda x}\right]e^{-i\lambda y},
\end{equation}
\begin{equation}
\label{4.20}
\text{\rm{RHS}}_{11}:=
\ds\int_{-\infty}^\infty \frac{d\lambda}{2\pi}\,e^{i\mu/2}\,E^{-1}\,R\,\psi_1\,
e^{i\lambda y},
\end{equation}
\begin{equation}
\label{4.21}
\text{\rm{RHS}}_{12}:=
\ds\int_{-\infty}^\infty \ds\frac{d\lambda}{2\pi}\,e^{i\mu/2}\,E^{-1}\,\ds\frac{\bar R}{\zeta}\,\bar\psi_1\,e^{-i\lambda y},
\end{equation}
\begin{equation}
\label{4.22}
\text{\rm{RHS}}_{21}:=
\ds\int_{-\infty}^\infty \ds\frac{d\lambda}{2\pi}\,e^{-i\mu/2}\,E\,\ds\frac{R}{\zeta}\,\psi_2\,e^{i\lambda y},
\end{equation}
\begin{equation}
\label{4.23}
\text{\rm{RHS}}_{22}:=\ds\int_{-\infty}^\infty \frac{d\lambda}{2\pi}\,e^{-i\mu/2}\,E\,\bar R\,\bar\psi_2\,e^{-i\lambda y}.
\end{equation}
Using the continuity properties of the Jost solutions stated in Theorem~\ref{theorem2.2},
the continuity and asymptotic properties of the  scattering coefficients presented 
in Theorem~\ref{theorem2.5}, and the small and large $\zeta$-asymptotics of the
Jost solutions stated in Theorems~\ref{theorem2.3} and \ref{theorem2.4}, respectively, 
we see that each integrand in \eqref{4.3}--\eqref{4.6} and \eqref{4.16}--\eqref{4.23} is continuous in $\lambda\in\mathbb R$ and $O(1/\lambda)$ as 
$\lambda\to\pm\infty.$ Thus, the $L^2$-Fourier transforms in \eqref{4.3}--\eqref{4.6} and \eqref{4.16}--\eqref{4.23} are all
well defined. Furthermore, in the absence of bound states, for $y>x$ the integrands in \eqref{4.3} and \eqref{4.4} are
analytic in $\lambda\in\mathbb C^+$ and uniformly $o(1)$ as $\lambda\to\infty$ in $\overline{\mathbb C^+}.$ Similarly,
in the absence of bound states, for $y>x$ the integrands in \eqref{4.5} and \eqref{4.6} are
analytic in $\lambda\in\mathbb C^-$ and uniformly $o(1)$ as $\lambda\to\infty$ in $\overline{\mathbb C^-}.$
Thus, from Jordan's lemma it follows that the four entries of the $2\times 2$ matrix $K(x,y)$ defined in \eqref{4.13}
are each equal to zero when $x>y.$ Hence, using the inverse Fourier transform, from \eqref{4.3}--\eqref{4.6} we get
\begin{equation}\label{4.24}
e^{i\mu/2}\,\ds\frac{\psi_1(\zeta,x)}{\zeta\,E(x)}=\ds\int_x^\infty dz\,K_1(x,z)\,e^{i\lambda z},
\end{equation}
\begin{equation}\label{4.25}
e^{-i\mu/2}\,E(x)\,\psi_2(\zeta,x)=e^{i\lambda x}+\ds\int_x^\infty dz\,K_2(x,z)\,e^{i\lambda z},
\end{equation}
\begin{equation}\label{4.26}
e^{i\mu/2}\,\ds\frac{\bar\psi_1(\zeta,x)}{E(x)}=e^{-i\lambda x}+\ds\int_x^\infty dz\,\bar K_1(x,z)\,e^{-i\lambda z},
\end{equation}
\begin{equation}\label{4.27}
e^{-i\mu/2}\,E(x)\,\ds\frac{\bar\psi_2(\zeta,x)}{\zeta}=\ds\int_x^\infty dz\,\bar K_2(x,z)\,e^{-i\lambda z}.
\end{equation}
Let us now show that each of the four entries of $\text{\rm{RHS}}$ defined in \eqref{4.15} is
a convolution. By using the inverse Fourier transform, from \eqref{4.2}
we have
\begin{equation}\label{4.28}
\ds\frac{R(\zeta)}{\zeta}=\ds\int_{-\infty}^\infty ds\,\hat R(s)\,e^{-i\lambda s},
\quad 
\frac{\bar R(\zeta)}{\zeta}=\ds\int_{-\infty}^\infty ds\,\hat{\bar R}(s)\,e^{i\lambda s}.
\end{equation}
Also, 
by taking the derivatives, from \eqref{4.2} we obtain
\begin{equation}\label{4.29}
\hat R'(y)=\ds\frac{i}{2\pi }\ds\int_{-\infty}^\infty  
d\lambda\,\ds\frac{R(\zeta)}{\zeta}\lambda\,e^{i\lambda y},
\quad \hat{\bar R}'(y)=-\ds\frac{i}{2\pi }\ds\int_{-\infty}^\infty  
d\lambda\,\ds\frac{\bar R(\zeta)}{\zeta}\lambda\,e^{-i\lambda y}.
\end{equation}
Using the inverse Fourier transform, from \eqref{4.29} we have
\begin{equation}\label{4.30}
\ds\frac{R(\zeta)}{\zeta}\,\lambda=-i\ds\int_{-\infty}^\infty ds\,\hat R'(s)\,e^{-i\lambda s},
\quad 
\frac{\bar R(\zeta)}{\zeta}\,\lambda=i\ds\int_{-\infty}^\infty ds\,\hat{\bar R}'(s)\,e^{i\lambda s}.
\end{equation}
Note that \eqref{4.20} is equivalent to
\begin{equation}
\label{4.31}
\text{\rm{RHS}}_{11}=
\ds\int_{-\infty}^\infty \frac{d\lambda}{2\pi}\,e^{i\lambda y}\left(e^{i\mu/2}\,E^{-1}\ds\frac{\psi_1}
{\zeta}\right)\left(\frac{R}{\zeta}\,\lambda\right).
\end{equation}
Using \eqref{4.24} and the first equality of \eqref{4.30}
on the right-hand side of \eqref{4.31}, we get the convolution
\begin{equation}
\label{4.32}
\text{\rm{RHS}}_{11}=-i\ds\int_x^\infty  dz\,K_1(x,z)\,\hat R'(z+y).
\end{equation}
Proceeding in a similar manner, we write \eqref{4.23} as
\begin{equation}
\label{4.33}
\text{\rm{RHS}}_{22}=\ds\int_{-\infty}^\infty \frac{d\lambda}{2\pi}
\,e^{-i\lambda y}\left(e^{-i\mu/2}\,E\,\ds\frac{\bar{\psi}_2}{\zeta}\right)\left(\frac{\bar R}{\zeta}\,\lambda\right).
\end{equation}
Using \eqref{4.27} and the second equality of \eqref{4.30}
on the right-hand side of \eqref{4.33}, we obtain the convolution 
\begin{equation}\label{4.34}
\text{\rm{RHS}}_{22}=\ds\int_x^\infty  dz\,\bar K_2(x,z)\,\hat{\bar R}'(z+y).
\end{equation}
In a similar manner, by using \eqref{4.25}, \eqref{4.26}, and \eqref{4.28}, we write
\eqref{4.21} and \eqref{4.22}, respectively, as
\begin{equation}\label{4.35}
\text{\rm{RHS}}_{12}
=\hat{\bar R}(x+y)+\ds\int_x^\infty  dz\,\bar K_1(x,z)\,\hat{\bar R}(z+y).
\end{equation}
\begin{equation}\label{4.36}
\text{\rm{RHS}}_{21}=\hat R(x+y)+\ds\int_x^\infty  dz\,K_2(x,z)\,\hat R(z+y).
\end{equation}
Hence, using \eqref{4.32}, \eqref{4.35}, \eqref{4.36}, and \eqref{4.34}
in \eqref{4.12}, we see that $\text{\rm{RHS}}$ is
equal to the sum of the second and third terms on the right-hand side of \eqref{4.1}.
Thus, in order to complete the derivation of \eqref{4.1}, it is sufficient
to show that $\text{\rm{LHS}}$ is the $2\times 2$ zero matrix when $x<y$ in the absence of bound states.
This is proved as follows. When $x<y,$ with the help of
Theorems~\ref{theorem2.2}--\ref{theorem2.5}, we observe that the integrands in \eqref{4.16} and \eqref{4.18} are analytic
in $\lambda\in\mathbb C^+,$ continuous in $\lambda\in\overline{\mathbb C^+},$
and uniformly $O(1/\lambda)$ as $\lambda\to\infty$ in $\overline{\mathbb C^+}.$
Hence, when $x<y,$ using Jordan's lemma and the residue theorem we conclude that
$\text{\rm{LHS}}_{11}$ and $\text{\rm{LHS}}_{21}$ are both zero.
Similarly, when $x<y,$ with the help of
Theorems~\ref{theorem2.2}--\ref{theorem2.5}, we observe that the integrands in \eqref{4.17} and \eqref{4.19} are analytic
in $\lambda\in\mathbb C^-,$ continuous in $\lambda\in\overline{\mathbb C^-},$
and uniformly $O(1/\lambda)$ as $\lambda\to\infty$ in $\overline{\mathbb C^-}.$
Hence, when $x<y,$ using Jordan's lemma and the residue theorem we conclude that
$\text{\rm{LHS}}_{12}$ and $\text{\rm{LHS}}_{22}$ are both zero. Thus, the proof is complete.
\end{proof}

The Marchenko integral system we have established in \eqref{4.1} is valid provided \eqref{1.1} has no bound states.
When the bound states are present,
the only modification needed in the proof of Theorem~\ref{theorem4.1} is that
the quantity $\text{\rm{LHS}}$ appearing in \eqref{4.12} and \eqref{4.14} is no longer equal to the zero matrix
due to the fact that we must take into
account the bound-state poles of the transmission coefficients in evaluating
the integrals \eqref{4.16}--\eqref{4.19}. It turns out that, using the matrix triplet pair $(A,B,C)$  and $(\bar A,\bar B,\bar C)$ 
appearing in \eqref{3.12}, \eqref{3.14}, \eqref{3.16}, \eqref{3.17},
we can express the effect 
of the bound states in the Marchenko system in a simple and elegant manner.
This amounts to replacing $\hat R(y)$ and $\hat{\bar R}(y)$ appearing in the Marchenko system
\eqref{4.1} with $\Omega(y)$ and $\bar\Omega(y),$ respectively, where we have defined
\begin{equation}\label{4.37}
\Omega(y):=\hat R(y)+C\,e^{iAy}\,B,\quad \bar\Omega(y):=\hat{\bar R}(y)+\bar C\,e^{-i\bar A y}\,\bar B.
\end{equation}
By taking the derivatives, from \eqref{4.37} we get
\begin{equation}\label{4.38}
\Omega'(y)=\hat R'(y)+i \,C A\, e^{iAy} B,\quad \bar\Omega'(y)=\hat{\bar R}'(y)-i\, \bar C \bar A \,e^{-i\bar A y} \bar B,
\end{equation}
and hence in \eqref{4.1} we also replace $\hat R'(y)$ and $\hat{\bar R}'(y)$
with $\Omega'(y)$ and $\bar\Omega'(y),$ respectively.

In fact, in the Marchenko equations for any linear system, the substitution 
 \begin{equation}\label{4.39}
\hat R(y)\mapsto \hat R(y)+C\,e^{iAy}\,B,
\quad \hat{\bar R}(y)\mapsto \hat{\bar R}(y)+\bar C\,e^{-i\bar A y}\,\bar B,
\end{equation}
is all that is needed in order to take into consideration the effect of any 
number of bound states with any multiplicities. Certainly, for linear systems where the potentials appear in the
diagonal blocks in the coefficient matrix
rather than in the off-diagonal blocks, we only use one matrix triplet $(A,B,C),$
and in that case \eqref{4.39} still holds with the understanding that the second matrix triplet $(\bar A,\bar B,\bar C)$
is absent. We remark that \eqref{4.39} is elegant for several reasons. 
When there is only one simple bound state, the eigenvalue of the matrix $A$ becomes the same as the matrix itself. In that sense,
there is an apparent correspondence between the factor $e^{i\lambda y}$ in \eqref{4.2} and $e^{i Ay}$ in
\eqref{4.39} induced by $\lambda\leftrightarrow A.$ The same is also true for the correspondence
between the factor $e^{-i\lambda y}$ in \eqref{4.2} and $e^{-i \bar Ay}$ in
\eqref{4.39} induced by $\lambda\leftrightarrow\bar A.$ 
The information containing any number of bound states with any multiplicities
and with the corresponding bound-state norming constants is all imbedded in \eqref{4.39} through the structure of
the two matrix triplets there.

In the next theorem we present the Marchenko integral system for \eqref{1.1} in the presence of bound states.

\begin{theorem}
\label{theorem4.2}
Let the potentials $q$ and $r$ in \eqref{1.1} belong to the Schwartz class.  
In the presence of bound states, the
corresponding Marchenko system for \eqref{1.1} is obtained from \eqref{4.1}
by using the substitution \eqref{4.39}, where $(A,B,C)$  and $(\bar A,\bar B,\bar C)$
are the pair of matrix triplets  
appearing in \eqref{3.12}, \eqref{3.14}, \eqref{3.16}, \eqref{3.17}. Hence, the Marchenko system for
\eqref{1.1} is
given by
\begin{equation}\label{4.40}
\begin{split}
\begin{bmatrix}
0&0\\ \noalign{\medskip}0&0
\end{bmatrix}=&\begin{bmatrix}
\bar K_1(x,y)&K_1(x,y)\\ \noalign{\medskip}\bar K_2(x,y)&K_2(x,y)
\end{bmatrix}+ \begin{bmatrix}
0&\bar\Omega(x+y)\\ \noalign{\medskip}\Omega(x+y)&0
\end{bmatrix}\\
\noalign{\medskip}
&+\ds\int_x^\infty dz\begin{bmatrix}
-iK_1(x,z)\,\Omega'(z+y)&\bar K_1(x,z)\,\bar\Omega(z+y)\\ \noalign{\medskip}
K_2(x,z)\,\Omega(z+y)&i\bar K_2(x,z)\,\bar\Omega'(z+y)
\end{bmatrix},\qquad x<y,
\end{split}
\end{equation}
where $\Omega(y)$ and $\bar\Omega(y)$ are the quantities
defined in \eqref{4.37}; $\Omega'(y)$ and
$\bar\Omega'(y)$ are the derivatives appearing in
\eqref{4.38}; and $K_1(x,y),$ $K_2(x,y),$ $\bar K_1(x,y),$
and $\bar K_2(x,y)$ are the quantities defined in
\eqref{4.3}--\eqref{4.6}, respectively.
\end{theorem}

\begin{proof}
As indicated in the proof of Theorem~\ref{4.1}, the quantity $\text{\rm{LHS}}$ in \eqref{4.14} is no longer the
$2\times 2$ zero matrix when the bound states are present. When $x<y,$ the integrands in
\eqref{4.16} and \eqref{4.18} are continuous 
in $\lambda\in\mathbb R,$ are $O(1/\lambda)$ as $\lambda\to\infty$
in $\overline{\mathbb C^+},$ and 
are meromophic in $\lambda\in\mathbb C^+$ with the
poles at $\lambda=\lambda_j$ with multiplicity $m_j$ for $1\le j\le N,$ where those poles are
the bound-state poles of $T(\zeta).$ Hence, when $x<y$ 
those integrals can be evaluated by using the residue theorem. The resulting
expressions contain the residues $t_{jk}$ appearing in \eqref{3.1}
and $d^k\phi(\zeta_j,x)/d\lambda^k$ for $1\le j\le N$ and $0\le k\le (m_j-1).$
Using \eqref{3.3} in the resulting expressions, we express those integrals
in terms of the residues $t_{jk}$ and the dependency constants $\gamma_{jk}$
appearing in \eqref{3.3}. In a similar manner, when $x<y$ the integrands in
\eqref{4.17} and \eqref{4.19} are continuous 
in $\lambda\in\mathbb R,$ are $O(1/\lambda)$ as $\lambda\to\infty$
in $\overline{\mathbb C^-},$ and 
are meromophic in $\lambda\in\mathbb C^-$ with the
poles at $\lambda=\bar\lambda_j$ with multiplicity $\bar m_j$ for $1\le j\le \bar N,$ where those poles are
the bound-state poles of $\bar T(\zeta).$ Thus, when $x<y$ 
those integrals can be evaluated by using the residue theorem. The resulting
expressions contain the residues $\bar t_{jk}$ appearing in \eqref{3.2}
and $d^k\bar\phi(\bar\zeta_j,x)/d\lambda^k$ for $1\le j\le \bar N$ and $0\le k\le (\bar m_j-1).$
Using \eqref{3.5} in the resulting expressions, we express those integrals
in terms of the residues $\bar t_{jk}$ and the dependency constants $\bar\gamma_{jk}$
appearing in \eqref{3.5}. We omit the details because the procedure is similar to
that given in the proof of Theorem~15 of \cite{AE2021}.
The only effect of the contribution from $\text{\rm{LHS}}$ to \eqref{4.12}
amounts to the substitution specified in \eqref{4.39}. Hence, 
with the help of \eqref{4.1}, \eqref{4.37}, and \eqref{4.38} we obtain
\eqref{4.40}, where the norming constants $c_{jk}$ are explicitly expressed
in terms of $t_{jk},$ $\gamma_{jk},$ and $\zeta_j,$
and 
the norming constants 
$\bar c_{jk}$
are explicitly expressed in terms of $\bar t_{jk},$ $\bar\gamma_{jk},$ and $\bar\zeta_j.$
\end{proof}

Let us remark that the $2\times 2$ matrix-valued coupled Marchenko system presented in \eqref{4.40} can readily be uncoupled, and it is
equivalent to the respective uncoupled scalar Marchenko integral equations for $K_1(x,y)$ and $\bar K_2(x,y)$
given by
\begin{equation}\label{4.41}
\begin{cases}
K_1(x,y)+\bar\Omega(x+y)+i\ds\int_x^\infty dz\int_x^\infty 
ds\,K_1(x,z)\,\Omega'(z+s)\,\bar\Omega(s+y)=0,
\\
\noalign{\medskip}
\bar K_2(x,y)+\Omega(x+y)-i\ds\int_x^\infty dz\int_x^\infty 
ds\,\bar K_2(x,z)\,\bar\Omega'(z+s)\,\Omega(s+y)=0,
\end{cases}
\end{equation}
where $x<y,$ with the auxiliary equations given by 
\begin{equation}\label{4.42}
\begin{cases}
\bar K_1(x,y)=i\ds\int_x^\infty dz\,K_1(x,z)\,\Omega'(z+y),\qquad x<y,
\\ \noalign{\medskip}
K_2(x,y)=-i\ds\int_x^\infty dz\,\bar K_2(x,z)\,\bar\Omega'(z+y),\qquad x<y.
\end{cases}
\end{equation}

Having established the Marchenko system for \eqref{1.1}, our goal now
is to recover the potentials $q$ and $r$ in \eqref{1.1} from the solution
$K(x,y)$ to the Marchenko system \eqref{4.40} or from the 
equivalent system of uncoupled equations given in \eqref{4.41} and \eqref{4.42}. In preparation for this, 
in the next theorem we evaluate $K(x,x)$ and $\bar K(x,x)$ from $K(x,y)$ and $\bar K(x,y)$ by letting $y \to x^+.$

\begin{proposition}\label{proposition4.3}
Assume that the potentials $q$ and $r$  appearing in \eqref{1.1} belong to the Schwartz class.
Let $K(x,y)$ be the solution to the Marchenko system
\eqref{4.40}, with the components $K_1(x,y),$ $K_2(x,y),$ $\bar K_1(x,y),$ $\bar K_2(x,y)$ 
as in \eqref{4.13}. In the limit $y\to x^+$ we have
\begin{equation}\label{4.43}
K_1(x,x)=-\ds\frac{e^{i\mu}}{2}\ds\frac{q(x)}{E(x)^2},
\end{equation}
\begin{equation}\label{4.44}
K_2(x,x)=-\ds\frac{iq(x)\,r(x)}{4}+\ds\frac{1}{2}\int_x^\infty dy\,\sigma(y),
\end{equation}
\begin{equation}\label{4.45}
\bar K_1(x,x)=\ds\frac{1}{2}\int_x^\infty dy\,\sigma(y),
\end{equation}
\begin{equation}\label{4.46}
\bar K_2(x,x)=-\ds\frac{e^{-i\mu}}{2}\,r(x)\,E(x)^2,
\end{equation}
where $E(x)$, $\mu,$ and $\sigma(x)$ are the quantities defined in \eqref{2.18}, \eqref{2.20}, and \eqref{2.37}, respectively.
\end{proposition}

\begin{proof}
Let us recall that $\zeta$ and $\lambda$ are related to each other as in \eqref{2.11}.
We obtain the proof by establishing the large $\lambda$-asymptotics of the Jost solutions
$\psi(\zeta,x)$ and $\bar\psi(\zeta,x)$ 
expressed in terms of the Fourier transforms given in \eqref{4.24}--\eqref{4.27} and by
comparing the results with the corresponding asymptotics expressions given in Theorem~\ref{theorem2.4}.
For example, in order to establish \eqref{4.43}, we write \eqref{4.24} as 
\begin{equation}\label{4.47}
\ds\frac{e^{i\mu/2}\,\psi_1(\zeta,x)}{\zeta\,E(x)}= 
\ds\int_x^\infty dy\left[K_1(x,y)\,\frac{d}{dy}\,\frac{e^{i\lambda y}}{i\lambda}\right],
\end{equation}
and using integration by parts, from \eqref{4.47}
we obtain
\begin{equation}
\label{4.48}
\ds\frac{e^{i\mu/2}\,\psi_1(\zeta,x)}{\zeta\,E(x)}=
K_1(x,y)\,\frac{e^{i\lambda y}}{i\lambda}\Bigg|_{y=x}^{y=\infty}-\int_x^\infty dy\,\frac{e^{i\lambda y}}
{i\lambda}\frac{\partial\,K_1(x,y)}{\partial y}.
\end{equation}
Since the potentials in \eqref{1.1} belong to the Schwartz class, the 
corresponding Jost solutions and their Fourier transforms are sufficiently smooth.
By letting 
$\lambda\to\pm\infty$ in \eqref{4.48} and using the Riemann--Lebesgue lemma, from \eqref{4.48} we get
 \begin{equation}\label{4.49}
 \ds\frac{e^{i\mu/2}\,\psi_1(\zeta,x)}{\zeta\,E(x)}=-\ds\frac{K_1(x,x)\,e^{i\lambda x}}{i\lambda}+o\left(\frac{1}{\lambda}\right).
 \end{equation}
The large $\zeta$-asymptotics of $\psi_1(\zeta,x)$ is given in the first component 
of \eqref{2.36}, and we use it on the left-hand side of \eqref{4.49} and obtain
\begin{equation}\label{4.50}
e^{i\mu+i\lambda x}\left[\ds\frac{q(x)}{2i\lambda\,E(x)^2} +O\left(\ds\frac{1}{\lambda^2}\right)\right]=
-\ds\frac{K_1(x,x)\,e^{i\lambda x}}{i\lambda}+o\left(\frac{1}{\lambda}\right),
\end{equation}
By comparing the first-order terms on both sides of \eqref{4.50}, we get \eqref{4.43}.
We then establish \eqref{4.44}--\eqref{4.46} by proceeding in a similar manner, i.e.
by using integration by parts in \eqref{4.25}--\eqref{4.27},
obtain the large $\lambda$-asymptotics in the resulting expressions
with the help of the Riemann--Lebesgue lemma, then by using
the large $\zeta$-asymptotics from \eqref{2.36} and \eqref{2.38} in the resulting
equalities, and finally by comparing the first-order terms
in the corresponding asymptotic expressions. 
\end{proof}

In the next theorem we show how to recover the relevant quantities for \eqref{1.1}, including
the potentials and the Jost solutions, from the solution to the corresponding
Marchenko system \eqref{4.40}.

\begin{theorem}
\label{theorem4.4}
Let the potentials $q$ and $r$ in \eqref{1.1} belong to the Schwartz class. The relevant quantities
are recovered from the solution to the
Marchenko system \eqref{4.40} or equivalently from the uncoupled counterpart given in \eqref{4.41} and \eqref{4.42}
as follows:

\begin{enumerate}

\item[\text{\rm(a)}] The scalar quantity
$E(x)$ given in \eqref{2.18} is recovered from
the solution to the Marchenko system as
\begin{equation}\label{4.51}
E(x)=\exp\left(2\ds\int_{-\infty}^{x}dz\,Q(z)\right),
\end{equation}
where $Q(x)$ is the auxiliary scalar quantity constructed from $\bar K_1(x,y)$ and $K_2(x,y)$ as
\begin{equation}\label{4.52}
Q(x):=\bar K_1(x,x)-K_2(x,x).
\end{equation}

\item[\text{\rm(b)}] The complex-valued scalar constant $\mu$ given in \eqref{2.20} is
obtained from the solution to the Marchenko system as
\begin{equation}\label{4.53}
\mu=-4i\ds\int_{-\infty}^\infty dz\,Q(z).
\end{equation}

\item[\text{\rm(c)}] The potentials $q$ and $r$ are recovered from the solution to the Marchenko system as
\begin{equation}\label{4.54}
q(x)=-2K_1(x,x)\exp\left(-4\ds\int_x^\infty dz\,Q(z)\right),
\end{equation}
\begin{equation}\label{4.55}
r(x)=-2\bar K_2(x,x)\exp\left(4\ds\int_x^\infty dz\,Q(z)\right).
\end{equation}

\item[\text{\rm(d)}] The Jost solutions $\psi(\zeta,x)$ and $\bar\psi(\zeta,x)$ 
to \eqref{1.1} are recovered 
from the solution to the Marchenko system as
\begin{equation}\label{4.56}
\psi_1(\zeta,x)= \zeta\left(\ds\int_x^\infty dy\,K_1(x,y)\,e^{i\zeta^2y}\right)
\exp\left(-2\ds\int_x^\infty dz\,Q(z)
\right),
\end{equation}
\begin{equation}\label{4.57}
\psi_2(\zeta,x)=\left(e^{i\zeta^2x}+\ds\int_x^\infty dy\,K_2(x,y)\,e^{i\zeta^2 y}\right) \exp\left(2\ds\int_x^\infty dz\,Q(z)
\right),
\end{equation}
\begin{equation}\label{4.58}
\bar{\psi}_1(\zeta,x)=\left (e^{-i\zeta^2x}+\ds\int_x^\infty dy\,\bar K_1(x,y)\,e^{-i\zeta^2 y}\right)
\exp\left(-2\ds\int_x^\infty dz\,Q(z)
\right),
\end{equation}
\begin{equation}\label{4.59}
\bar{\psi}_2(\zeta,x)= \zeta\left(\ds\int_x^\infty dy\,\bar K_2(x,y)\,e^{-i\zeta^2y}\right)
\exp\left(2\ds\int_x^\infty dz\,Q(z)\right),
\end{equation}
where $\psi_1(\zeta,x),$ $\psi_2(\zeta,x),$ $\bar{\psi}_1(\zeta,x),$ and $\bar{\psi}_2(\zeta,x)$ 
are the components of the Jost solutions defined in \eqref{2.9}.

\end{enumerate}
\end{theorem}

\begin{proof}
From \eqref{4.44} and \eqref{4.45}, we see that the auxiliary scalar quantity $Q(x)$ defined in \eqref{4.52}
is related to the potentials $q$ and $r$ as
\begin{equation}
\label{4.60}
Q(x)=\ds\frac{i\,q(x)\,r(x)}{4}.
\end{equation}
Hence, from \eqref{2.18} and \eqref{4.60} we see that $E(x)$ is recovered as in \eqref{4.51}, which completes the proof of (a).
Similarly, from \eqref{2.20} and \eqref{4.60} we observe that $\mu$ is recovered as in \eqref{4.53}, 
and therefore the proof of (b) is also completed.
Let us now prove (c). Having obtained $E(x)$ and $\mu,$ we see that we can recover
$q(x)$ with the help of \eqref{4.43}. Thus, using \eqref{4.51} and \eqref{4.53}
in \eqref{4.43} we recover $q(x)$ as in \eqref{4.54}. Similarly, having 
$E(x)$ and $\mu$ already recovered, we see that we can obtain
$r(x)$ from \eqref{4.46}. Therefore, using \eqref{4.51} and \eqref{4.53}
in \eqref{4.46} we recover $r(x)$ as in  \eqref{4.55}. 
Let us now prove (d). Having $E(x)$ and $\mu$ at hand, we use \eqref{2.11}, \eqref{4.51}, and \eqref{4.53} in 
\eqref{4.24}--\eqref{4.27}, respectively, and get \eqref{4.56}--\eqref{4.59}.
Hence, the whole proof is complete.
\end{proof}

As in any inverse problem, the inverse problem for \eqref{1.1} has four aspects: the existence, uniqueness,
reconstruction, and characterization. The existence deals with the question whether
there exists at least one pair of potentials $q(x)$ and $r(x)$ in some class corresponding to a given set of
scattering data in a particular class. Once the existence problem is
solved, the uniqueness deals with the question whether there is only one pair of
potentials for that scattering data set or there are more such pairs. The reconstruction is concerned with the
recovery of the potentials from the scattering data set. Finally, the characterization
deals with the specification of the class of potentials and the class of scattering data sets
so that there is a one-to-one correspondence between the elements of
the class of potentials and the class of scattering data sets. It is clear that in this paper we only
deal with the reconstruction aspect of the inverse problem for \eqref{1.1}.
The remaining three aspects are challenging and need to be investigated. Since
the linear differential operator related to \eqref{1.1} is not selfadjoint, the analysis of the inverse problem
for \eqref{1.1} is naturally complicated. We anticipate that the development of the Marchenko method
in this paper will provide a motivation for the scientific community to analyze the other three
aspects of the corresponding inverse problem.

\section{Solution formulas with reflectionless scattering data}
\label{section5}

In this section we provide the solution to the Marchenko system \eqref{4.40} when the reflection coefficients
in the input scattering data
set are zero. Using the results of Section~\ref{section4}, we then obtain
the corresponding potentials and Jost solutions explicitly expressed in terms of the matrix
triplets $(A,B,C)$ and $(\bar A,\bar B,\bar C)$ 
with the triplet sizes $\mathcal N$ and $\bar{\mathcal N},$ respectively. 
We recall that $\mathcal N$ and $\bar{\mathcal N}$ are the integers appearing in
\eqref{3.13} and \eqref{3.15}, respectively.
Thus, with $R(\zeta)\equiv 0$ and $\bar R(\zeta)\equiv 0,$ from \eqref{4.37} 
and \eqref{4.38} we get
\begin{equation}\label{5.1}
\Omega(y)=C\,e^{iAy}\,B,\quad \bar\Omega(y)=\bar C\,e^{-i\bar A y}\,\bar B,
\end{equation}
\begin{equation}\label{5.2}
\Omega'(y)=i\,C A \,e^{iAy}\,B,\quad \bar\Omega'(y)=-i\bar C\bar A \,e^{-i\bar A y}\,\bar B.
\end{equation}
With the input from \eqref{5.1} and \eqref{5.2}, the Marchenko system \eqref{4.40} or the equivalent
uncoupled Marchenko system given in \eqref{4.41} and \eqref{4.42} is explicitly solvable by the
methods of linear algebra because the corresponding integral kernels are separable.
Consequently, we obtain the closed-form formulas for the potentials and Jost solutions for \eqref{1.1} 
corresponding to all reflectionless scattering data, where the formulas are explicitly expressed
in terms of the two matrix triplets.
We present the relevant formulas when the 
matrix triplet 
sizes $\mathcal N$ and $\bar{\mathcal N}$ 
are arbitrary.
We then prove that, if the potentials $q$ and $r$ in \eqref{1.1} belong to the
Schwartz class, in the reflectionless case we must have 
$\mathcal N=\bar{\mathcal N}.$

In the next theorem we present the solution to the Marchenko system with the input from \eqref{5.1} and \eqref{5.2},
which are uniquely determined by the matrix triplets 
$(A,B,C)$ and $(\bar A,\bar B,\bar C).$

\begin{theorem}
\label{theorem5.1}
When the scattering data set in \eqref{5.1} is used as input,
the Marchenko system \eqref{4.40} corresponding to \eqref{1.1} has the solution expressed in closed form given by
\begin{equation}\label{5.3}
K_1(x,y)
=-\bar C\,e^{-i\bar A x}\,\bar\Gamma(x)^{-1}\,e^{-i\bar A y}\,\bar B,
\end{equation}
\begin{equation}\label{5.4}
K_2(x,y)
=C\,e^{iAx}\,\Gamma(x)^{-1}\,e^{iAx}\,M\,
\bar A\,e^{-i\bar A (x+y)}\,\bar B,
\end{equation}
\begin{equation}\label{5.5}
\bar K_1(x,y)
=\bar C\,e^{-i\bar A x}\,\bar\Gamma(x)^{-1}\,e^{-i\bar A x}\,
\bar M\,A\,e^{iA(x+y)}\,B,
\end{equation}
\begin{equation}\label{5.6}
\bar K_2(x,y)
=-C\,e^{iAx}\,\Gamma(x)^{-1}\,e^{iAy}\,B,
\end{equation}
where $(A,B,C)$ and $(\bar A,\bar B,\bar C)$
are the two matrix triplets appearing in \eqref{5.1}, and
 $\Gamma(x),$ $\bar\Gamma(x),$ $M,$ and $\bar M$ are the matrices defined in terms of
the two matrix triplets as
\begin{equation}\label{5.7}
\Gamma(x):=I-e^{iAx}\,M\,\bar A\,e^{-2i\bar A x}\,\bar M\,e^{iAx},
\end{equation}
\begin{equation}\label{5.8}
\bar\Gamma(x):=I-e^{-i\bar A x}\,\bar M\,A\,e^{2iAx}\,M\,e^{-i\bar A x},
\end{equation}
\begin{equation}\label{5.9}
M:=\int_{0}^\infty dz\,e^{iAz}\,B\,\bar C\,e^{-i\bar A z},\quad 
\bar M:=\int_{0}^\infty dz\,e^{-i\bar A z}\,\bar B\,C\,e^{i A z},
\end{equation}
with $I$ denoting an identity matrix whose size is not necessarily the
same 
in different appearances.
\end{theorem}

\begin{proof}
Since the Marchenko system \eqref{4.40} is equivalent to the uncoupled
system given in \eqref{4.41} and \eqref{4.42}, we use \eqref{5.1} and \eqref{5.2} as input to
that uncoupled sysytem. The first line of \eqref{4.41} yields
\begin{equation*}
K_1(x,y)
+\bar C\,e^{-i\bar Ax-i\bar Ay}\,\bar B
+i\ds\int_x^\infty dz\int_x^\infty 
ds\,K_1(x,z)\,i\,C\,A\,e^{iAz+iAs}\,B\,\bar C\,e^{-i\bar As-i\bar A y}\,\bar B=0,
\end{equation*}
whose solution has the form
\begin{equation}\label{5.10}
K_1(x,y)=H_1(x)\,e^{-i\bar A y}\,\bar B,
\end{equation}
with $H_1(x)$ satisfying
\begin{equation}\label{5.11}
H_1(x)\left[I-\ds\int_x^\infty dz\int_x^\infty 
ds\,e^{-i\bar A z}\,\bar B\,C\,e^{iAz}\,A\,e^{iAs}\,B\bar C\,e^{-i\bar A s}\right]
=-\bar C\,e^{-i\bar A x}.
\end{equation}
The matrix in the brackets in \eqref{5.11} is equal to $\bar\Gamma(x)$ defined 
in \eqref{5.8}, and this can be seen by observing that
\begin{equation}\label{5.12}
\ds\int_x^\infty dz
\,e^{-i\bar A z}\,\bar B\,C\,e^{iAz}=e^{-i\bar A x}\,\bar M\, e^{iAx},
\end{equation}
\begin{equation}\label{5.13}
\ds\int_x^\infty ds
\,e^{i A s}\,B\,\bar C\,e^{-i \bar A s}=e^{iAx} \,M\, e^{-i\bar A x},
\end{equation}
where $M$ and $\bar M$ are the constant matrices defined
in \eqref{5.9}.
When the eigenvalues of $A$ are located in $\mathbb C^+$ and
the eigenvalues of $\bar A$ are in $\mathbb C^-,$ we see that
the two integrals in \eqref{5.9} are well defined. From \eqref{5.9} we also see that
the matrices $M$ and $\bar M$ can alternatively
be obtained from the matrix triplets 
$(A,B,C)$ and $(\bar A,\bar B,\bar C)$
by solving the respective linear systems given by
\begin{equation*}
i M\bar A-iA M=B \bar C, \quad i\bar A \bar M-i\bar M  A=\bar B C.
\end{equation*}
From \eqref{5.11} we have
\begin{equation}\label{5.14}
H_1(x)=-\bar C\,e^{-i\bar A x} \,\bar \Gamma(x)^{-1}.\end{equation}
Hence, using \eqref{5.14} in \eqref{5.10} we get \eqref{5.3}.
We obtain \eqref{5.4} in a similar manner, by using 
\eqref{5.1} and \eqref{5.2} as input in the second line of \eqref{4.41}.
We then have
\begin{equation*}
\bar K_2(x,y)
+C\,e^{iAx+iAy}\,B
-\ds\int_x^\infty dz\int_x^\infty 
ds\,\bar K_2(x,z)\,\bar C\,\bar A\,e^{-i\bar A z-i\bar A s}\,\bar B\,C\,e^{iAs+iAy}\,B=0,
\end{equation*}
whose solution has the form
\begin{equation}\label{5.15}
\bar K_2(x,y)=H_2(x)\,e^{iAy}\,B,
\end{equation}
with $H_2(x)$ satisfying
\begin{equation}\label{5.16}
H_2(x)\left[I-\ds\int_x^\infty dz\int_x^\infty 
ds\,e^{iAz}\,B\,\bar C\,e^{-i\bar A z}\,\bar A\,e^{-i\bar A s}\,\bar B\,C\,e^{iAs}\right]
=-C\,e^{iAx}.
\end{equation}
With the help of \eqref{5.12} and \eqref{5.13} we 
observe that the matrix in the brackets in \eqref{5.16} is equal to the matrix $\Gamma(x)$ defined in
\eqref{5.7}, and hence \eqref{5.16} yields
\begin{equation}\label{5.17}
H_2(x)=-C\,e^{iAx}\,\Gamma(x)^{-1}.
\end{equation}
Using \eqref{5.17} in \eqref{5.15} we obtain \eqref{5.4}.
Finally, using \eqref{5.3} and \eqref{5.4} as input
to \eqref{4.42}, with the help of
\eqref{5.9} we get \eqref{5.5} and \eqref{5.6}.
\end{proof}

In the next theorem we present the explicit expressions for the key quantity $E(x)$ in \eqref{2.18} and the potentials $q$ and $r$ in \eqref{1.1}
corresponding to the reflectionless scattering data set described by the 
pair of matrix triplets $(A,B,C)$ and $(\bar A,\bar B,\bar C).$

\begin{theorem}
\label{theorem5.2}
The scalar quantity $E(x)$ defined in \eqref{2.18} 
corresponding to the reflectionless scattering data in \eqref{5.1} is expressed explicitly 
in terms of the matrix triplets $(A,B,C)$ and $(\bar A,\bar B,\bar C)$ as
\begin{equation}\label{5.18}
E(x)=\exp\left(2\ds\int_{-\infty}^x dz\left[g_1(z)-g_2(z)\right]\right),
\end{equation}
and the potentials $q$ and $r$ in \eqref{1.1}
corresponding to the same reflectionless scattering data set are expressed explicitly 
in terms of the matrix triplets as
\begin{equation}\label{5.19}
q(x)=2\,\bar C\,e^{-i\bar A x}\,\bar\Gamma(x)^{-1}\,e^{-i\bar A x}\,\bar B\,e^{-4\,G(x)},
\end{equation}
\begin{equation}\label{5.20}
r(x)=2\,C\,e^{i A x}\,\Gamma(x)^{-1}\,e^{i A x}\,B\,e^{4\,G(x)},
\end{equation}
where we have defined
\begin{equation}
\label{5.21}
G(x):=\ds\int_x^\infty dz\left[g_1(z)-g_2(z)\right],
\end{equation}
\begin{equation}
\label{5.22}
g_1(z):=\bar C\,e^{-i\bar A z}\,\bar\Gamma(z)^{-1}\,e^{-i\bar A z}\,\bar M\,A\,e^{2iAz}\,B,
\end{equation}
\begin{equation}
\label{5.23}
g_2(z):=C\,e^{iAz}\,\Gamma(z)	^{-1}\,e^{iAz}\,M\,\bar A\,e^{-2i\bar A z}\,\bar B,
\end{equation}
with $\Gamma(x),$ $\bar\Gamma(x),$ $M,$ and $\bar M$  being the matrices
 appearing in \eqref{5.7}, \eqref{5.8}, and \eqref{5.9}.
\end{theorem}

\begin{proof} From \eqref{4.52}, \eqref{5.4}, and \eqref{5.5} we observe that
\begin{equation}
\label{5.24}
g_1(x)=\bar K_1(x,x), \quad g_2(x)=K_2(x,x), \quad g_1(x)-g_2(x)=Q(x).
\end{equation}
Then, we get \eqref{5.18} by using the last equality of \eqref{5.24} in \eqref{4.51}. We obtain \eqref{5.19} by 
using \eqref{5.3}, \eqref{5.4}, and \eqref{5.5} in 
\eqref{4.52} and \eqref{4.54} and simplifying the resulting expression.
Similarly, we get \eqref{5.20} by using \eqref{5.4}, \eqref{5.5}, and \eqref{5.6} in 
\eqref{4.52} and \eqref{4.55} and simplifying the resulting expression.
\end{proof}

In the following theorem, the Jost solutions $\psi(\zeta,x)$ and $\bar\psi(\zeta,x)$
corresponding to the reflectionless potentials are expressed explicitly in terms of the 
pair of matrix triplets $(A,B,C)$ and $(\bar A,\bar B,\bar C).$

\begin{theorem}
\label{theorem5.3}
The Jost solutions to \eqref{1.1} appearing in \eqref{2.1} and \eqref{2.2},
respectively, corresponding to the reflectionless
potentials $q$ and $r$ given in \eqref{5.19} and \eqref{5.20}, are
explicitly expressed in terms of the
pair of matrix triplets $(A,B,C)$ and $(\bar A,\bar B,\bar C)$ as
\begin{equation}\label{5.25}
\psi_1(\zeta,x)= \zeta e^{i \zeta^2 x}
\,g_3(\zeta,x)\, e^{-2\, G(x)},
\end{equation}
\begin{equation}\label{5.26}
\psi_2(\zeta,x)= e^{i \zeta^2 x}
\,g_4(\zeta,x)\, e^{2\, G(x)},
\end{equation}
\begin{equation}\label{5.27}
\bar\psi_1(\zeta,x)= e^{-i \zeta^2 x}
\,g_5(\zeta,x)\, e^{-2\, G(x)},
\end{equation}
\begin{equation}\label{5.28}
\bar\psi_2(\zeta,x)= \zeta e^{-i \zeta^2 x}
\,g_6(\zeta,x)\, e^{2\, G(x)},
\end{equation}
where $G(x)$ is the quantity defined via \eqref{5.21}--\eqref{5.23},
and the quantities $g_3(\zeta,x),$
$g_4(\zeta,x),$ $g_5(\zeta,x),$ $g_6(\zeta,x)$
are defined as
\begin{equation*}
g_3(\zeta,x):=
i\,\bar C\,e^{-i\bar A x}\,\bar\Gamma(x)^{-1}\,e^{-i\bar A x} (\bar A -\zeta^2 I)^{-1}\bar B,
\end{equation*}
\begin{equation}\label{5.29}
g_4(\zeta,x):=1-i\,C\,e^{iAx}\,\Gamma(x)^{-1}\,e^{iAx}\,M\,\bar A\,e^{-2i\bar A x}\,(\bar A -\zeta^2 I)^{-1}\,\bar B,
\end{equation}
\begin{equation}\label{5.30}
g_5(\zeta,x):=1+
i\,\bar C\,e^{-i\bar A x}\,\bar\Gamma(x)^{-1}\,e^{-i\bar A x}\,\bar M\,A\,e^{2iAx}\,(A-\zeta^2 I)^{-1}\,B,
\end{equation}
\begin{equation*}
g_6(\zeta,x):=
-i\,C\,e^{iAx}\,\Gamma(x)^{-1}\,e^{iAx}\,(A-\zeta^2 I)^{-1}\,B,
\end{equation*}
with $\Gamma(x),$ $\bar\Gamma(x),$ $M,$ and $\bar M$  being the matrices
 appearing in \eqref{5.7}, \eqref{5.8}, and \eqref{5.9}.
\end{theorem}

\begin{proof}
With the input \eqref{5.1} specified in terms of the matrix triplets
$(A,B,C)$ and $(\bar A,\bar B,\bar C),$ the corresponding
solution to the Marchenko system \eqref{4.40} is explicitly given 
in \eqref{5.3}--\eqref{5.6}. Using those expressions
in \eqref{4.56}--\eqref{4.59},
we obtain the corresponding Jost solutions 
$\psi(\zeta,x)$ and $\bar\psi(\zeta,x)$ specified in \eqref{5.25}--\eqref{5.28}. The details are as follows.
As seen from \eqref{4.52} and the last equality in \eqref{5.24},
the exponential factors on the right-hand sides of
\eqref{4.56}--\eqref{4.59}
 are equal to either $e^{-2 G(x)}$ or $e^{2 G(x)},$
where $G(x)$ is the quantity defined in \eqref{5.21}.
From \eqref{5.21}--\eqref{5.23}, we observe that each of
$e^{-2 G(x)}$ and $e^{2 G(x)}$ is explicitly expressed
in terms of the two matrix triplets
$(A,B,C)$ and $(\bar A,\bar B,\bar C).$
We then consider the integral terms related
to the Fourier transforms in 
\eqref{4.56}--\eqref{4.59} and show that
each of those integrals can be explicitly expressed in terms
of our two matrix triplets. In fact, from
\eqref{5.3} we get
\begin{equation}\label{5.31}
\ds\int_x^\infty dy\,e^{i\lambda y}K_1(x,y)=i\,\bar C\,e^{-i\bar A x}\,\bar\Gamma(x)
^{-1}\,(\bar A -\lambda I)^{-1}e^{-i\bar A x+i\lambda x}\bar B.
\end{equation}
With the help of \eqref{2.11} and \eqref{5.31}, we write \eqref{4.56} as \eqref{5.25}.
Similarly, from \eqref{5.4} we obtain
\begin{equation}\label{5.32}
\ds\int_x^\infty dy \,e^{i\lambda y}K_2(x,y)
=-i\,C\,e^{iAx}\,\Gamma
^{-1}\,e^{iAx}\,M\,\bar A \,e^{-i\bar A x}\,(\bar A -\lambda I)^{-1}\,e^{-i\bar A x+i\lambda x}\,\bar B,
\end{equation}
and with the help of \eqref{2.11} and \eqref{5.32}, we express \eqref{4.57} as \eqref{5.26}.
In a similar manner, from \eqref{5.5} we get
\begin{equation}\label{5.33}
\ds\int_x^\infty dy\,e^{-i\lambda y}\bar K_1(x,y)
=i\,\bar C\,e^{-i\bar A x}\,\bar\Gamma(x)
^{-1}\,e^{-i\bar A x}\,\bar M\,A\,e^{iAx}\,(A-\lambda I)^{-1}\,e^{iAx-i\lambda x}\,B.
\end{equation}
Then, using \eqref{2.11} and \eqref{5.33} we write \eqref{4.58} as \eqref{5.27}.
Finally, from \eqref{5.6} we obtain
\begin{equation}\label{5.34}
\ds\int_x^\infty dy\,e^{-i\lambda y}\bar K_2(x,y)=-i\,C\,e^{iAx}\,\Gamma(x)
^{-1}\,(A-\lambda I)^{-1}\,e^{iAx-i\lambda x}\,B,
\end{equation}
and with the help of \eqref{2.11} and \eqref{5.34}, we write \eqref{4.59} as \eqref{5.28}.
We remark that the right-hand sides in \eqref{5.25}--\eqref{5.28} are all
expressed explicitly in terms of the matrix
triplets $(A,B,C)$ and $(\bar A,\bar B,\bar C)$ because
$\Gamma(x),$ $\bar\Gamma(x),$ $M,$ and $\bar M$ and in turn
the quantities $g_1(x),$
$g_2(x),$ $g_3(\zeta,x),$
$g_4(\zeta,x),$ $g_5(\zeta,x),$ $g_6(\zeta,x),$ and $G(x)$ are all 
explicitly expressed in terms of those two matrix triplets.
\end{proof}

For the reflectionless scattering data set specified in \eqref{5.1}, in Theorem~\ref{theorem5.1} we
have determined the corresponding solution to the Marchenko system \eqref{4.40}, in Theorem~\ref{theorem5.2} we
have provided the corresponding potentials, and in Theorem~\ref{theorem5.3} we
have obtained the corresponding Jost solutions. In the next theorem, for that same data set
we express the corresponding value of the constant $\mu$ defined in \eqref{2.20} and the corresponding 
transmission coefficients
explicitly in terms of the matrix triplets $(A,B,C)$ and $(\bar A,\bar B,\bar C)$
appearing in \eqref{5.1}.

\begin{theorem}
\label{theorem5.4}
For the reflectionless scattering data set specified in \eqref{5.1} expressed explicitly in terms of the matrix
triplets $(A,B,C)$ and $(\bar A,\bar B,\bar C),$ we have the following:

\begin{enumerate}

\item[\text{\rm(a)}] The corresponding value of the constant $\mu$ defined in\eqref{2.20} is explicitly determined by
the matrix triplet pair $(A,B,C)$ and $(\bar A,\bar B,\bar C)$ as
\begin{equation}\label{5.35}
\mu=-4 i\ds\int_{-\infty}^\infty dz\left[g_1(z)-g_2(z)\right],
\end{equation}
and hence the value of $e^{i\mu/2}$ is determined by
the matrix triplet pair as
\begin{equation}\label{5.36}
e^{i\mu/2}=\exp\left(2\ds\int_{-\infty}^\infty dz\left[g_1(z)-g_2(z)\right]
\right),
\end{equation}
where, as seen from \eqref{5.22} and \eqref{5.23},
the quantities $g_1(z)$ and $g_2(z)$ are explicitly determined
by our matrix triplet pair with the help of \eqref{5.7}--\eqref{5.9}.

\item[\text{\rm(b)}] The transmission coefficients
$T(\zeta)$ and $\bar T(\zeta)$ corresponding to 
 the reflectionless scattering data set specified in \eqref{5.1} are explicitly determined by
the matrix triplet pair $(A,B,C)$ and $(\bar A,\bar B,\bar C)$ as
\begin{equation}\label{5.37}
T(\zeta)=\ds\frac{1}{g_4(\zeta,-\infty)}\,\exp\left(-2\ds\int_{-\infty}^\infty dz\left[g_1(z)-g_2(z)\right]
\right),
\end{equation}
\begin{equation}\label{5.38}
\bar T(\zeta)=\ds\frac{1}{g_5(\zeta,-\infty)}\,\exp\left(2\ds\int_{-\infty}^\infty dz\left[g_1(z)-g_2(z)\right]
\right),
\end{equation}
where, as seen from \eqref{5.29} and \eqref{5.30},
the quantities $g_4(\zeta,x)$ and $g_5(\zeta,x)$ are explicitly determined
by our pair of matrix triplets with the help of \eqref{5.7}--\eqref{5.9}.

\end{enumerate}
\end{theorem}

\begin{proof}
We obtain \eqref{5.35} directly from \eqref{4.53} and
the last equality in \eqref{5.24}.
Then, \eqref{5.36} is a direct consequence of \eqref{5.35}. Alternatively, as seen from the second equality in \eqref{2.19},
we get \eqref{5.36} from \eqref{5.18} by letting $x\to+\infty$ there. Hence, the proof of (a)
is complete.
Note that \eqref{5.37} follows from the second component of \eqref{2.5} with
the help of \eqref{5.21} and
\eqref{5.26}. Similarly, \eqref{5.38} is
obtained by using the first component of \eqref{2.6} with
the help of \eqref{5.21} and
\eqref{5.27}. 
\end{proof}

As indicated at the end of Section~\ref{section4}, in this paper we only deal with the reconstruction
aspect of the inverse problem for \eqref{1.1}. Hence, the results presented in this section should be interpreted
in the sense of the reconstruction. The potentials and the corresponding Jost solutions are reconstructed
explicitly in Theorems~\ref{theorem5.2} and \ref{theorem5.3}, respectively, from their reflectionless
scattering data expressed in terms of a pair of matrix triplets. When the potentials $q$ and $r$
belong to the Schwartz class, there are additional restrictions on the two matrix triplets used in Theorem~\ref{5.2}.
As seen from \eqref{5.19} and \eqref{5.20},
those restrictions amount to the following: The determinants of the matrices $\Gamma(x)$ and $\bar\Gamma(x)$ defined
in \eqref{5.7} and \eqref{5.8} should not vanish for any $x\in\mathbb R,$ and the
exponential terms in \eqref{5.19} and \eqref{5.20} should not cause an exponential
increase and in fact should not yield a nonzero asymptotic
value as $x\to\pm\infty.$ In the next section we will illustrate this issue with
some explicit examples.

When the potentials $q$ and $r$ in \eqref{1.1} belong to the Schwartz class, in the reflectionless case we present an
important restriction on the number of bound states for \eqref{1.1}, and we now elaborate on this issue. 
Recall that the nonnegative
integer
$\mathcal N$ defined in \eqref{3.13} corresponds to the number of bound states, including the multiplicities, associated with the
bound-state poles of the transmission coefficient $T(\zeta)$ in the first quadrant in
the complex $\zeta$-plane. Similarly, the nonnegative
integer
$\bar{\mathcal N}$ defined in \eqref{3.15} corresponds to the number of bound states, including the multiplicities, 
associated with the
bound-state poles of the transmission coefficient $\bar T(\zeta)$ in the second quadrant in
the complex $\zeta$-plane.
In general,  $\mathcal N$ and $\bar{\mathcal N}$
do not have to be equal to each other. 
However, in the reflectionless case, when $q$ and $r$ belong to the Schwartz class, 
we will prove that we must have $\mathcal N=\bar{\mathcal N}.$ In fact, we will prove that this is also true for the AKNS system 
\eqref{1.7}, i.e. when the potentials $u$ and $v$ belong to the Schwartz class, in the reflectionless case
the number of bound-state poles, including the multiplicities, of the transmission coefficient $T^{(u,v)}(\lambda)$ in $\mathbb C^+$
must be equal to the number of bound-state poles, including the multiplicities, of the transmission coefficient $\bar T^{(u,v)}(\lambda)$ in $\mathbb C^-.$
Thus, in the explicit solution formulas presented in Theorems~\ref{theorem5.1}--\ref{theorem5.4}, unless we choose
the sizes of the matrices $A$ and $\bar A$ equal to each other, the corresponding potentials $q$ and $r$ both
cannot belong to the Schwartz class. We will illustrate this in Example~\ref{example6.7} in the next section.

The next theorem indicates the restriction $\mathcal N=\bar{\mathcal N}$ in the reflectionless case
when the potentials in \eqref{1.1} belong to the Schwartz class.

\begin{theorem}
\label{theorem5.5}
Let the potentials $q$ and $r$ in \eqref{1.1} belong to the Schwartz class, and assume that
the corresponding reflection coefficients $R(\zeta)$ and $\bar R(\zeta)$ 
appearing in \eqref{2.7} and \eqref{2.8}, respectively, are zero.
Then, we have the following:

\begin{enumerate}

\item[\text{\rm(a)}] The corresponding nonnegative
integers
$\mathcal N$ and $\bar{\mathcal N}$ defined in \eqref{3.13} and \eqref{3.15}, respectively, must be equal to each other, i.e. we must have
\begin{equation}\label{5.39}
\mathcal N=\bar{\mathcal N}.
\end{equation}

\item[\text{\rm(b)}] The corresponding matrix triplets $(A,B,C)$ and $(\bar A,\bar B,\bar C)$ appearing in \eqref{5.1} must have the same sizes.

\end{enumerate}

\end{theorem}

\begin{proof} 
Recall that the spectral parameter $\zeta$ is related to the parameter $\lambda$ as in \eqref{2.11}.
Based on the bound-state information provided in \eqref{3.7}, we know that
the transmission coefficient $T(\zeta)$ appearing in \eqref{2.5} is a meromorphic function of $\lambda$ in $\mathbb C^+$ with the poles
at $\lambda=\lambda_j,$ each with multiplicity $m_j$ for $1\le j\le N.$ Using Theorem~\ref{theorem2.5} we conclude
that the quantity $1/T(\zeta)$ is analytic in $\lambda\in\mathbb C^+,$ is continuous in $\lambda\in\overline{\mathbb C^+},$ 
vanishes only
at $\lambda=\lambda_j$ for $1\le j\le N,$ and has the large $\lambda$-asymptotics described in \eqref{2.39}. 
Similarly, from Theorem~\ref{theorem2.5} we conclude that the transmission coefficient 
$\bar T(\zeta)$ appearing in \eqref{2.6} is analytic in $\lambda\in\mathbb C^-,$ 
is continuous in $\lambda\in\overline{\mathbb C^-},$ vanishes only
at $\lambda=\bar\lambda_k$ for $1\le k\le \bar N,$ and has the large $\lambda$-asymptotics described in \eqref{2.40}. 
Let us write $T(\zeta)$ and $\bar T(\zeta),$ respectively, as
\begin{equation}\label{5.40}
T(\zeta)=T_0(\zeta)\,\ds\prod_{j=1}^N\ds\left(\frac{\lambda-\lambda_j^*}{\lambda-\lambda_j}\right)^{m_j},
\end{equation}
\begin{equation}\label{5.41}
\bar T(\zeta)=\bar T_0(\zeta)\,\ds\prod_{k=1}^{\bar N}\ds\left(\frac{\lambda-\bar\lambda_k^*}{\lambda-\bar\lambda_k}\right)^{\bar m_k},
\end{equation}
where $1/T_0(\zeta)$ 
is analytic in $\lambda\in\mathbb C^+,$ 
is continuous in $\lambda\in\overline{\mathbb C^+},$ does not vanish in $\overline{\mathbb C^+},$
and has the large $\lambda$-asymptotics given by
\begin{equation}\label{5.42}
\ds\frac{1}{T_0(\zeta)}=\ds e^{i\mu/2}\left[1+O\left(\frac{1}{\lambda}\right)\right],\qquad \lambda\to\infty  \text{\rm{ in }} 
\lambda\in\overline{\mathbb C^+},
\end{equation}			
and $1/\bar T_0(\zeta)$ 
is analytic in $\lambda\in\mathbb C^-,$ 
is continuous in $\lambda\in\overline{\mathbb C^-},$ does not vanish in $\overline{\mathbb C^-},$
and has the large $\lambda$-asymptotics given by
\begin{equation}\label{5.43}
\ds\frac{1}{\bar T_0(\zeta)}=\ds e^{-i\mu/2}\left[1+O\left(\frac{1}{\lambda}\right)\right],\qquad \lambda\to\infty  \text{\rm{ in }} 
\lambda\in\overline{\mathbb C^-}.
\end{equation}	
Note that we use an asterisk to denote complex conjugation.
It is known \cite{AKNS1974,E2018} that the scattering coefficients for the AKNS system
\eqref{1.7} satisfy
\begin{equation}\label{5.44}
T^{(u,v)}(\lambda)\,\bar T^{(u,v)}(\lambda)+R^{(u,v)}(\lambda)\,\bar R^{(u,v)}(\lambda)=1,
\qquad \lambda\in\mathbb R.
\end{equation}
Using the first equalities of \eqref{2.43}--\eqref{2.46} in \eqref{5.44} we obtain
\begin{equation}\label{5.45}
T(\zeta)\,\bar T(\zeta)+R(\zeta)\,\bar R(\zeta)=1,
\end{equation}
and hence, in the reflectionless case,  from \eqref{5.45} we get
\begin{equation}\label{5.46}
T(\zeta)\,\bar T(\zeta)=1,\qquad \lambda\in\mathbb R,
\end{equation}
where we recall that $T(\zeta)$ and $\bar T(\zeta)$ each contain $\zeta$ as $\zeta^2$ and thus
\eqref{5.46} is valid for $\lambda\in\mathbb R.$
Using \eqref{5.40} and \eqref{5.41} in \eqref{5.46} we obtain
\begin{equation}\label{5.47}
T_0(\zeta)\,\bar T_0(\zeta)\ds\prod_{j=1}^N\ds\left(\frac{\lambda-\lambda_j^*}{\lambda-\lambda_j}\right)^{m_j}\,
\ds\prod_{k=1}^{\bar N}\ds\left(\frac{\lambda-\bar\lambda_k^*}{\lambda-\bar\lambda_k}\right)^{\bar m_k}
=1,\qquad \lambda\in\mathbb R.
\end{equation}
Let us rewrite \eqref{5.47} so that the left-hand side is analytic in $\lambda\in \mathbb C^+$ 
and the right-hand side is analytic in $\lambda\in \mathbb C^-.$
For $\lambda\in\mathbb R,$ we then get
\begin{equation}\label{5.48}
\ds e^{i\mu/2}\,T_0(\zeta)\,\ds\prod_{j=1}^N\ds\left(\lambda-\lambda_j^*\right)^{m_j}\,
\ds\prod_{k=1}^{\bar N}\ds\left(\frac{1}{\lambda-\bar\lambda_k}\right)^{\bar m_k}
=\ds\frac{1}{\ds e^{-i\mu/2}\,\bar T_0(\zeta)}\,\ds\prod_{j=1}^N\ds\left(\lambda-\lambda_j\right)^{m_j}\,
\ds\prod_{k=1}^{\bar N}\ds\left(\frac{1}{\lambda-\bar\lambda_k^*}\right)^{\bar m_k}.
\end{equation}
We must have either $\mathcal N\ge \bar{\mathcal N}$ or
$\mathcal N\le \bar{\mathcal N}.$ We will prove that either of those two inequalities can hold only in the
case of an equality. The proof for the former is as follows.
When $\mathcal N\ge \bar{\mathcal N},$ with the help of \eqref{5.42} we conclude that
the left-hand side of \eqref{5.48} has an extension from $\lambda\in\mathbb R$ to
$\mathbb C^+$ in such a way that that extension is analytic in $\lambda\in\mathbb C^+,$ continuous in
$\lambda\in\overline{\mathbb C^+},$ and is asymptotic to a monic polynomial $P(\lambda)$ of degree
$\mathcal N-\bar{\mathcal N}$ as $\lambda\to\infty$ in $\overline{\mathbb C^+}.$
Similarly, with the help of \eqref{5.43} we conclude that
the right-hand side of \eqref{5.48} is analytic in $\lambda\in\mathbb C^-,$ continuous in
$\lambda\in\overline{\mathbb C^-},$ and is asymptotic to $P(\lambda)$ as $\lambda\to\infty$ in $\overline{\mathbb C^-}.$ 
Thus, both sides of \eqref{5.48} must have an analytic extension to the entire complex $\lambda$-plane and 
be equal to a monic polynomial of degree $\mathcal N-\bar{\mathcal N},$ i.e. we
must have
\begin{equation}\label{5.49}
P(\lambda)=
\ds e^{i\mu/2}\,T_0(\zeta)\,\ds\prod_{j=1}^N\ds\left(\lambda-\lambda_j^*\right)^{m_j}\,
\ds\prod_{k=1}^{\bar N}\ds\left(\frac{1}{\lambda-\bar\lambda_k}\right)^{\bar m_k},\qquad \lambda\in\mathbb C,
\end{equation}
\begin{equation}\label{5.50}
P(\lambda)=
\ds\frac{1}{\ds e^{-i\mu/2}\,\bar T_0(\zeta)}\,\ds\prod_{j=1}^N\ds\left(\lambda-\lambda_j\right)^{m_j}\,
\ds\prod_{k=1}^{\bar N}\ds\left(\frac{1}{\lambda-\bar\lambda_k^*}\right)^{\bar m_k},\qquad \lambda\in\mathbb C.
\end{equation}
From \eqref{5.46} we see that neither $T(\zeta)$ nor $\bar T(\zeta)$ can have any zeros or any poles
when $\lambda\in\mathbb R.$ Hence, from \eqref{5.41} we can conclude that
$\bar T_0(\zeta)$ does not have any poles
when $\lambda\in\mathbb R.$
Consequently, from \eqref{5.50} we conclude that $P(\lambda)$ cannot have any zeros
when $\lambda\in\mathbb R.$
From the right-hand side of \eqref{5.50}, we also see that $P(\lambda)$ cannot have any zeros when
$\lambda\in\mathbb C^-$ and as a result any zero of $P(\lambda)$ can only occur when $\lambda\in\mathbb C^+.$
Consequently, any pole of $1/P(\lambda)$ can only occur when $\lambda\in\mathbb C^+.$
Let us write \eqref{5.49} as
\begin{equation}\label{5.51}
\ds\frac{1}{P(\lambda)}=
\ds \frac{1}{e^{i\mu/2}\,T_0(\zeta)}\,\ds\prod_{j=1}^N\ds\left(\frac{1}{\lambda-\lambda_j^*}\right)^{m_j}\,
\ds\prod_{k=1}^{\bar N}\ds\left(\lambda-\bar\lambda_k\right)^{\bar m_k},\qquad \lambda\in\mathbb C.
\end{equation}
From the right-hand side of \eqref{5.51} we see that $1/P(\lambda)$ cannot have any poles when $\lambda\in\mathbb C^+.$
Therefore, we conclude that the monic polynomial $P(\lambda)$ cannot have any zeros at all when $\lambda\in\mathbb C.$
Hence, we must have $P(\lambda)\equiv 1,$ which yields $\mathcal N=\bar{\mathcal N}.$ 
A similar argument shows that the case $\mathcal N\le \bar{\mathcal N}$ 
can occur only when $\mathcal N=\bar{\mathcal N}.$ Thus, the proof of (a) is complete.
The proof of (b) is a direct consequence of (a) because
the matrix triplet $(A,B,C)$ has size $\mathcal N$ and
the matrix triplet $(\bar A,\bar B,\bar C)$ has size $\bar{\mathcal N}.$ 
\end{proof}

Let us remark that the result presented in Theorem~\ref{theorem5.5} for \eqref{1.1} holds also for
the AKNS system given in \eqref{1.7}. Next, we present that result as a corollary because its proof follows by essentially
repeating the proof given for Theorem~\ref{theorem5.5}.

\begin{corollary}
\label{corollary5.6}
Let the potentials $u$ and $v$ in the AKNS system \eqref{1.7} belong to the Schwartz class. Let us also assume that
the corresponding reflection coefficients $R^{(u,v)}(\lambda)$ and $\bar R^{(u,v)}(\lambda)$ 
are zero. Then, 
the number of bound-state poles, including the multiplicities, of the transmission coefficient $T^{(u,v)}(\lambda)$ in $\mathbb C^+$
must be equal to the number of bound-state poles, including the multiplicities, of the transmission coefficient $\bar T^{(u,v)}(\lambda)$ in $\mathbb C^-.$
\end{corollary}

\section{Explicit examples}
\label{section6}

In this section we elaborate on the results from the previous sections with some illustrative and explicit examples.

As indicated in Section~\ref{section3}, for the linear system \eqref{1.1} one can construct
the norming constants $c_{jk}$ appearing in \eqref{3.7}
explicitly in terms of the set of residues $\{t_{jk}\}_{k=1}^{m_j}$ 
and the dependency constants $\{\gamma_{jk}\}_{k=0}^{m_j-1}.$
Similarly, one can construct
the norming constants
$\bar c_{jk}$ appearing in \eqref{3.7}
explicitly in terms of the set of residues $\{\bar t_{jk}\}_{k=1}^{\bar m_j}$ 
and the dependency constants $\{\bar\gamma_{jk}\}_{k=0}^{\bar m_j-1}.$
In the first two examples, we illustrate that construction and observe that, especially in the case
of bound states with multiplicities, it is
cumbersome to deal with the individual norming constants, and it is better to use
the bound-state information not in the form given in \eqref{3.7}
but rather in the form of matrix triplet pair
$(A,B,C)$ and $(\bar A,\bar B,\bar C).$
 
The first example considers the norming constants for simple bound states.

\begin{example}
\label{example6.1}
\normalfont
Consider the linear system \eqref{1.1} with the potentials $q$ and $r$ in the Schwartz class. We elaborate on step (d) 
appearing in
the beginning of Section~\ref{section3}.
If the bound state at $\lambda=\lambda_j$ is simple, then we have $m_j=1$ and hence there is only one norming 
constant $c_{j0}.$ By proceeding as
in \cite{AE2021} we obtain
\begin{equation}\label{6.1}
c_{j0}=-\ds\frac{i\,t_{j1}\,\gamma_{j0}}{\zeta_j},\end{equation}
where $\zeta_j$ is the complex number in the first quadrant
in $\mathbb C$ for which we have $\lambda_j=\zeta_j^2,$
the complex constant $t_{j1}$ corresponds to the residue in \eqref{3.1} in the expansion of
the transmission coefficient $T(\zeta),$ i.e.
\begin{equation*}
T(\zeta)=\ds\frac{t_{j1}}{\lambda-\lambda_j}+O(1),\qquad \lambda\to\lambda_j,
\end{equation*}
and $\gamma_{j0}$ is the dependency constant appearing in \eqref{3.3}, i.e.
\begin{equation*}
\phi(\zeta_j,x)=\gamma_{j0}\,\psi(\zeta_j,x),
\end{equation*}
with $\psi(\zeta,x)$ and $\phi(\zeta,x)$ being the Jost solutions appearing in
\eqref{2.1} and \eqref{2.3}, respectively. If the bound state at 
$\lambda=\bar\lambda_j$ is simple, we have $\bar m_j=1$ and hence there is only one norming 
constant $\bar c_{j0},$ which is expressed as
\begin{equation}\label{6.2}
\bar c_{j0}=\ds\frac{i\,\bar t_{j1}\,\bar\gamma_{j0}}{\bar\zeta_j},\end{equation}
where $\bar\zeta_j$ is the complex number in the fourth quadrant
in $\mathbb C$ for which we have $\bar\lambda_j=\bar\zeta_j^2,$
the complex constant $\bar t_{j1}$ corresponds to the residue in \eqref{3.2} in the expansion of
the transmission coefficient $\bar T(\zeta),$ i.e.
\begin{equation*}
\bar T(\zeta)=\ds\frac{\bar t_{j1}}{\lambda-\bar\lambda_j}+O(1),\qquad \lambda\to\bar\lambda_j,
\end{equation*}
and $\bar\gamma_{j0}$ is the dependency constant appearing in \eqref{3.5}, i.e.
\begin{equation*}
\bar\phi(\bar\zeta_j,x)=\bar\gamma_{j0}\,\bar\psi(\bar\zeta_j,x),
\end{equation*}
with $\bar\psi(\zeta,x)$ and $\bar\phi(\zeta,x)$ being the Jost solutions appearing in
\eqref{2.2} and \eqref{2.4}, respectively. 
As seen from \eqref{6.1} and \eqref{6.2}, the norming constants $c_{j0}$ and
$\bar c_{j0}$ are related 
to each other via the transformations
\begin{equation}\label{6.3}
\lambda_j\mapsto\bar\lambda_j,
\quad
t_{jk}\mapsto -\bar t_{jk},
\quad 
\gamma_{jk}\mapsto \bar\gamma_{jk},
\quad
c_{jk}\mapsto \bar c_{jk}.
\end{equation}

\end{example}

The next example considers the norming constants for bound states with multiplicities.

\begin{example}
\label{example6.2}
\normalfont
We consider the linear system \eqref{1.1} with the potentials $q$ and $r$ in the Schwartz class, and we elaborate on step (d) 
appearing in
the beginning of Section~\ref{section3}.
If the bound state at $\lambda=\lambda_j$ is double, we have $m_j=2$ and there are only two norming 
constant $c_{j0}$ and $c_{j1},$ which are expressed in terms of the
residues $t_{j1}$ and $t_{j2}$ and the dependency constants
$\gamma_{j0}$ and $\gamma_{j1}$ as
\begin{equation}\label{6.4}
\begin{cases}
c_{j1}=-\ds\frac{i\,t_{j2}\,\gamma_{j0}}{\zeta_j},\\
\noalign{\medskip}
c_{j0}=-\ds\frac{i\,t_{j1}\,\gamma_{j0}}{\zeta_j}-\ds\frac{i\,t_{j2}}{\zeta_j}\left(\gamma_{j1}-\frac{\gamma_{j0}}{2\lambda_j}\right),
\end{cases}
\end{equation}
where we recall that $\zeta_j$ is the complex constant  in the first quadrant
in $\mathbb C$ for which we have $\lambda_j=\zeta_j^2.$
If the bound state at $\lambda=\bar\lambda_j$ is double, we have $\bar m_j=2$ and there are only two norming 
constant $\bar c_{j0}$ and $\bar c_{j1},$ which are obtained 
from \eqref{6.4} by using the transformations given in \eqref{6.3}.
For a triple bound state at $\lambda=\lambda_j,$ we have  $m_j=3$ and the three norming 
constants are expressed 
in terms of the
residues $t_{j1},$ $t_{j2},$ $t_{j3}$ and the dependency constants
$\gamma_{j0},$ $\gamma_{j1},$ $\gamma_{j2}$
as
\begin{equation}
\label{6.5}
\begin{cases}
c_{j2}=-\ds\frac{i\,t_{j3}\,\gamma_{j0}}{\zeta_j},\\
\noalign{\medskip}
c_{j1}=-\ds\frac{i\,t_{j2}\,\gamma_{j0}}{\zeta_j}-\ds\frac{i\,t_{j3}}{\zeta_j}\left(\gamma_{j1}-\frac{\gamma_{j0}}{2\lambda_j}\right),
\\
\noalign{\medskip}
c_{j0}=-\ds\frac{i\,t_{j1}\,\gamma_{j0}}{\zeta_j}
-\ds\frac{i\,t_{j2}}{\zeta_j}\left(\gamma_{j1}
-\frac{\gamma_{j0}}{2\lambda_j}\right)
-\ds\frac{i\,t_{j3}}{2\zeta_j}\left(\gamma_{j2}-\frac{\gamma_{j1}}{\lambda_j}+\frac{3\gamma_{j0}}{4\lambda_j^2}\right).
\end{cases}
\end{equation}
For a bound state at $\lambda=\bar\lambda_j$ of multiplicity three, 
we can obtain the norming constants $\bar c_{j0},$ $\bar c_{j1},$ $\bar c_{j2}$
by using the transformations in \eqref{6.3} on \eqref{6.5}.
For bound states with higher multiplicities, the norming constants can be explicitly constructed
by using the corresponding residues and the dependency constants. However, as 
already mentioned, the use of the matrix triplet pair
$(A,B,C)$ and $(\bar A,\bar B,\bar C)$
is the simplest and most elegant way to represent the bound-state
information without having to deal with any cumbersome formulas involving the individual norming constants.

\end{example}

The formulas 
presented in Theorems~\ref{theorem5.1}, \ref{theorem5.2}, and \ref{theorem5.3}
express all the relevant quantities in a compact form with the help of
matrix exponentials.
We have prepared a Mathematica notebook
using the matrix triplets
$(A,B,C)$ and $(\bar A,\bar B,\bar C)$ as input and evaluating
all the relevant quantities by unpacking the matrix exponentials
and displaying all those relevant quantities in terms of 
elementary functions. 
In particular, our Mathematica notebook
provides in terms of elementary functions
the solution to the Marchenko system as indicated in Theorem~\ref{theorem5.1},
the potentials $q$ and $r$ given in Theorem~\ref{theorem5.2},
the Jost solutions given in Theorem~\ref{theorem5.3}, and
the corresponding auxiliary quantities $E(x)$ and $\mu$ given in \eqref{5.18} and \eqref{5.35}, respectively.
It also verifies 
that \eqref{1.1} is satisfied when those expressions for
the potentials and the Jost solutions are used in \eqref{1.1}.
As the matrix sizes in the triplets
get large, contrary to the compact expressions involving
the matrix exponentials, the equivalent
expressions presented in terms of elementary functions 
become lengthy. 

In the next example, we illustrate Theorems~\ref{theorem5.2} and \ref{theorem5.3} by using
a pair of matrix triplets corresponding to two simple bound states.

\begin{figure}[!h]
  \centering
\includegraphics[width=2.5in]{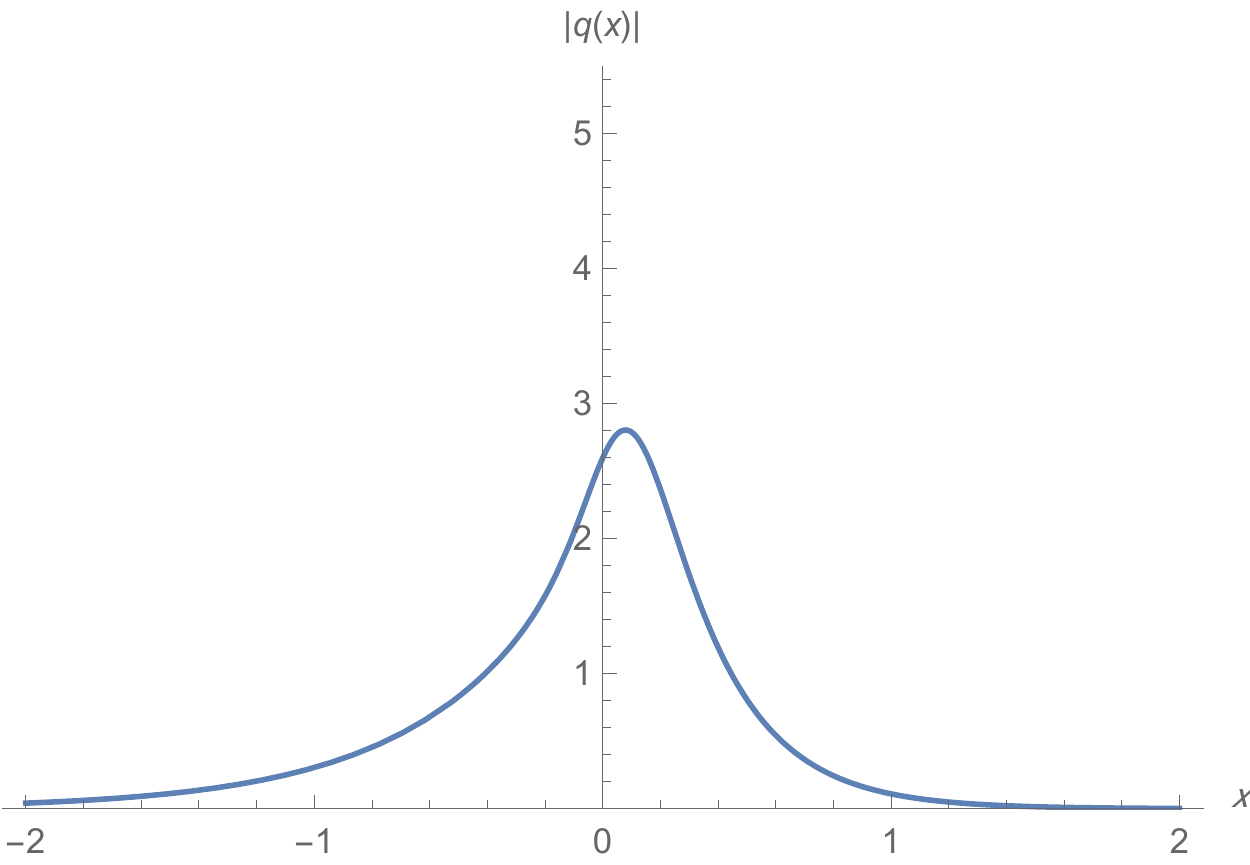}\hskip 0.5cm
\includegraphics[width=2.5in]{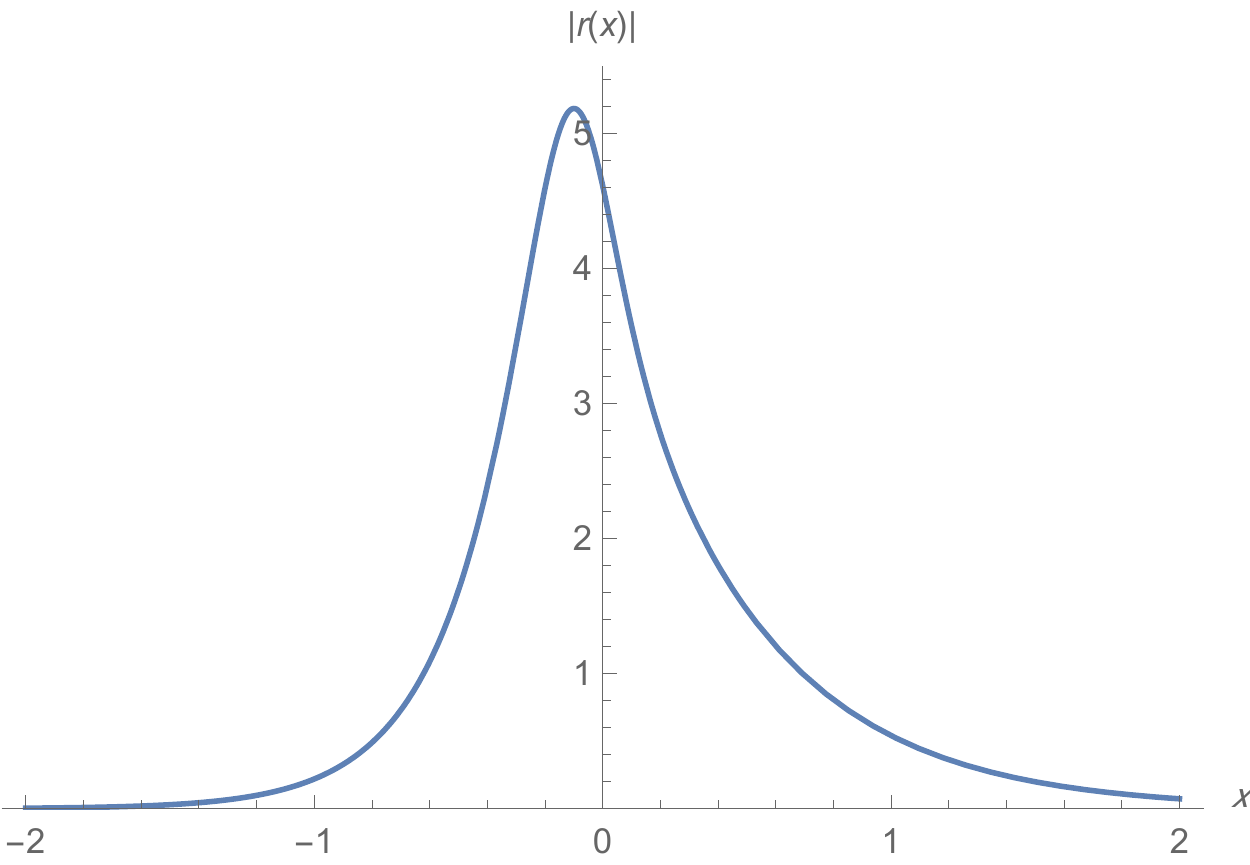}
  \caption{The absolute potentials $|q(x)|$ and $|r(x)|$ in Example~\ref{example6.3}}
\end{figure}

\begin{example}
\label{example6.3}
\normalfont
 Consider the reflectionless scattering data with two
simple bound states described by the matrix triplets $(A,B,C)$ and $(\bar A,\bar B,\bar C)$ given by
\begin{equation}
\label{6.6}
A=\begin{bmatrix}i\end{bmatrix}, \quad B=\begin{bmatrix}1\end{bmatrix}, \quad
C=\begin{bmatrix}2\end{bmatrix},
\quad
\bar A=\begin{bmatrix}-2i\end{bmatrix}, \quad \bar B=\begin{bmatrix}1\end{bmatrix}, \quad
\bar C=\begin{bmatrix}3\end{bmatrix}.
\end{equation}
Using \eqref{6.6} in \eqref{5.19} and \eqref{5.20}, we obtain the corresponding potentials $q$ and $r$ as
\begin{equation}
\label{6.7}
q(x)=\ds\frac{18}{3 e^{6x}-2i}\,\exp\bigg(2x-4\,\tanh^{-1}\big(1/3-i e^{6x}\big)\bigg),
\end{equation}
\begin{equation}
\label{6.8}
r(x)=\ds\frac{12}{3 e^{6x}+4i}\,\exp\bigg(4x+4\,\tanh^{-1}\big(1/3-i e^{6x}\big)\bigg),
\end{equation}
where we use the principal branch of the inverse hyperbolic tangent.
Since $q(x)$ and $r(x)$ are complex valued, in Figure~1 we present the plots of their absolute values.
From \eqref{6.7} and \eqref{6.8} we see that
$q$ and $r$ both belong to the Schwartz class.
The corresponding Jost solutions $\psi(\zeta,x)$ and $\bar\psi(\zeta,x)$ are obtained by using \eqref{6.6} in
\eqref{5.25}--\eqref{5.28}, and we get
\begin{equation*}
\psi_1(\zeta,x)=\ds\frac{-9 \zeta  \,\exp\bigg(i\lambda x+2x-2\,\tanh^{-1}\big(1/3-i e^{6x}\big)\bigg)}{(\lambda+2i)\left(2+3i e^{6x}\right)},
\end{equation*}
\begin{equation}\label{6.9}
\psi_2(\zeta,x)=\ds\frac{\omega_1 \,\exp\bigg(i\lambda x+2\,\tanh^{-1}\big(1/3-i e^{6x}\big)\bigg)}{(\lambda+2i)\left(-4+3i e^{6x}\right)},
\end{equation}
\begin{equation}\label{6.10}
\bar\psi_1(\zeta,x)=\ds\frac{\omega_2  \,\exp\bigg(-i\lambda x-2\,\tanh^{-1}\big(1/3-i e^{6x}\big)\bigg)}
{(\lambda-i)\left(2+3i e^{6x}\right)},
\end{equation}
\begin{equation*}
\bar\psi_2(\zeta,x)=\ds\frac{6 \zeta  \,\exp\bigg(-i\lambda x+4x+2\,\tanh^{-1}\big(1/3-i e^{6x}\big)\bigg)}
{(\lambda-i)\left(-4+3i e^{6x}\right)},
\end{equation*}
where we have defined
\begin{equation*}
\omega_1:=4(\lambda-i)-3i e^{6x}\,(\lambda+2i),
\quad
\omega_2:=-2(\lambda+2i)-3i e^{6x}\,(\lambda-i),
\end{equation*}
with $\lambda=\zeta^2,$ as indicated in \eqref{2.11}.
In this example, the constant $\mu$ appearing in \eqref{2.20} and the transmission coefficients
$T(\zeta)$ and $\bar T(\zeta)$ are given by
\begin{equation*}
\mu=2\pi-2i\,\ln 2,\quad T(\zeta)=-\ds\frac{1}{2}\,\left(\ds\frac{\lambda+2i}{\lambda-i}\right),\quad \bar T(\zeta)=
-2\,\left(\ds\frac{\lambda-i}{\lambda+2i}\right),
\end{equation*}
which can be verified by using the asymptotics of \eqref{6.9} and \eqref{6.10} as $x\to-\infty.$

\end{example}

\begin{figure}[!h]
  \centering
\includegraphics[width=2.5in]{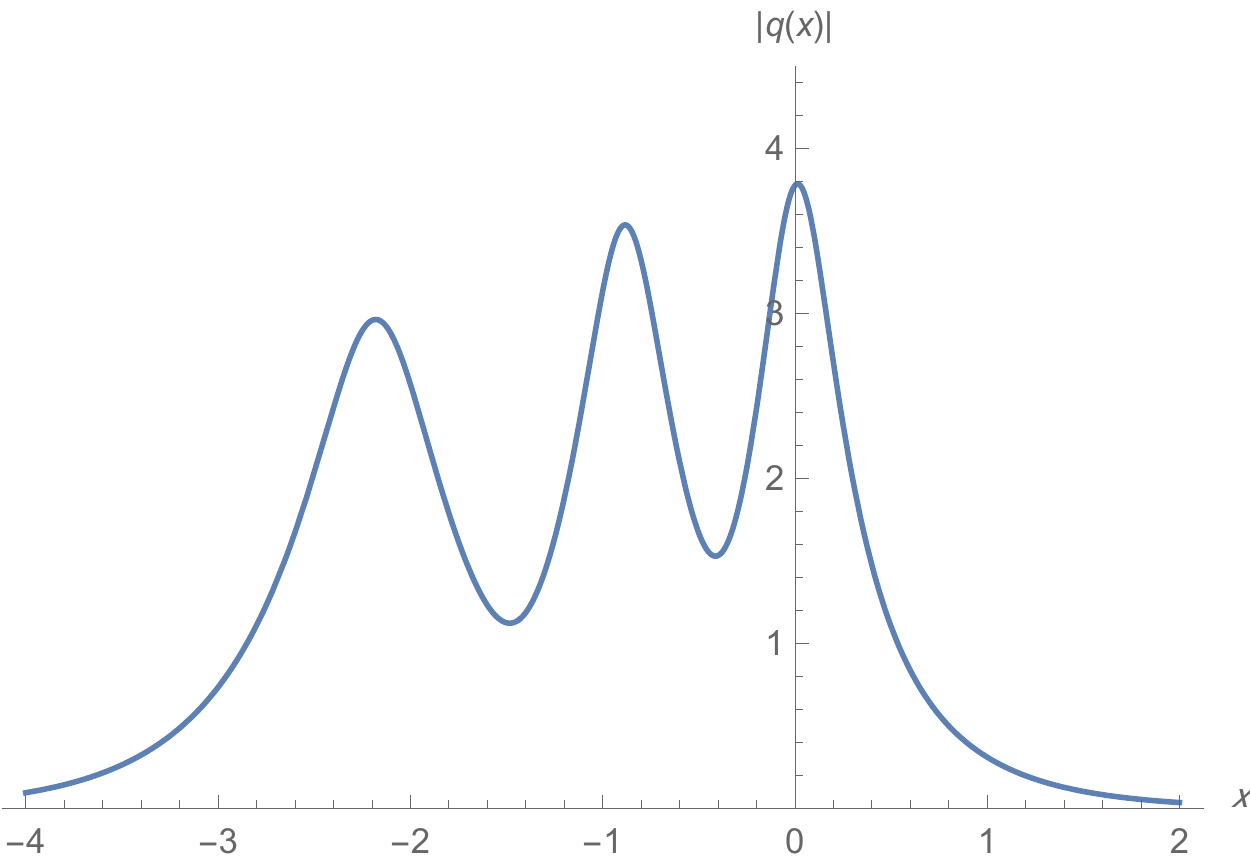}\hskip 0.5in
\includegraphics[width=2.5in]{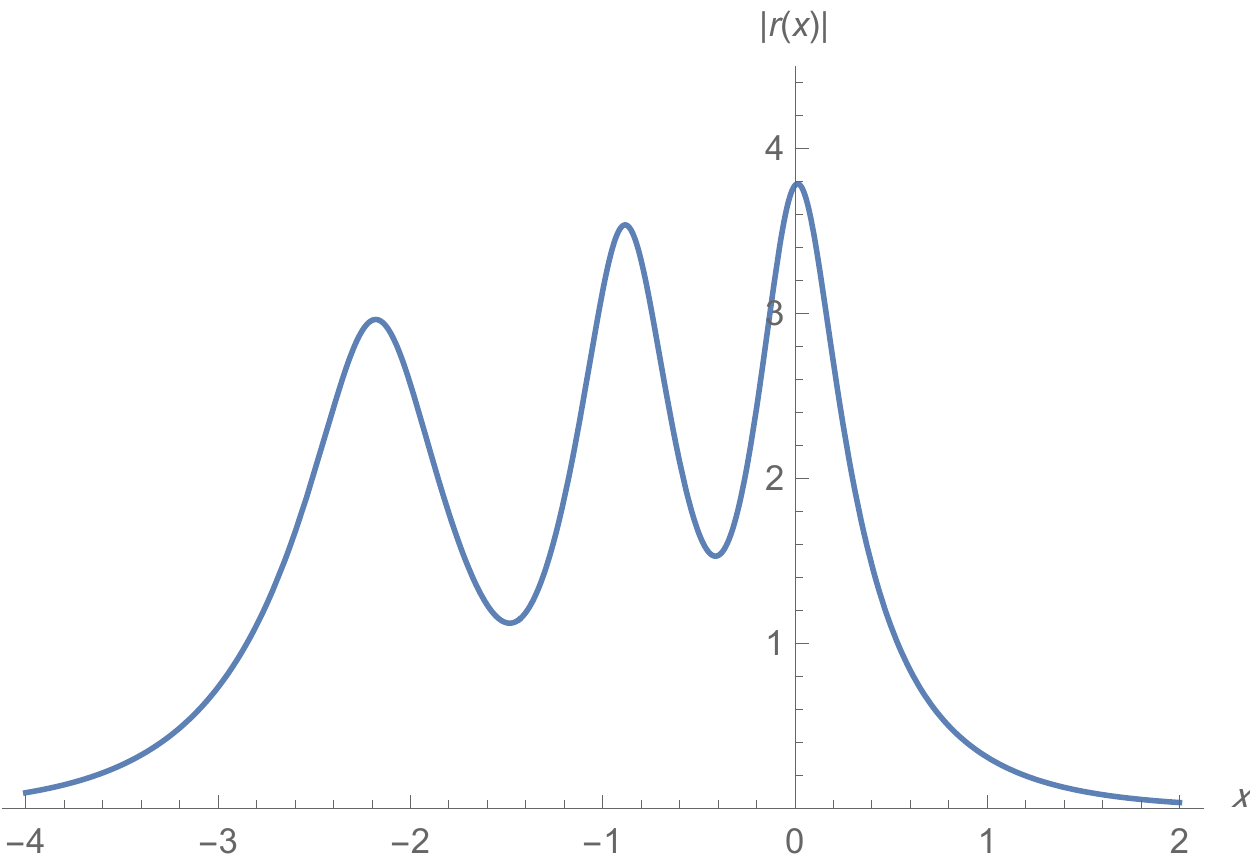}
  \caption{The absolute potentials $|q(x)|$ and $|r(x)|$ in Example~\ref{example6.4}.}
\end{figure}

In the next example we illustrate Theorem~\ref{theorem5.2} by using a pair of matrix triplets
corresponding to six simple
bound states.

\begin{example}
\label{example6.4}
\normalfont
Consider the reflectionless scattering data with six simple
bound states described by the matrix triplets $(A,B,C)$ and $(\bar A,\bar B,\bar C)$ given by
\begin{equation}\label{6.11}
A=\begin{bmatrix}i&0&0\\
0&2i&0\\
0&0&3i\end{bmatrix}, \quad B=\begin{bmatrix}1\\
1\\
1\end{bmatrix}, \quad
C=\begin{bmatrix}1&1&1\end{bmatrix},
\end{equation}
\begin{equation}\label{6.12}
\bar A=\begin{bmatrix}-i&0&0\\
0&-2i&0\\
0&0&-3i\end{bmatrix}, \quad \bar B=\begin{bmatrix}1\\ 1
\\
1\end{bmatrix}, \quad
\bar C=\begin{bmatrix}1&1&1\end{bmatrix}.
\end{equation}
Using \eqref{6.11} and \eqref{6.12} as input in \eqref{5.19} and \eqref{5.20}, we obtain the
corresponding potentials $q$ and $r$ as
\begin{equation}
\label{6.13}
q(x)=\ds\frac{48(\omega_5+\omega_6)}{\omega_7+\omega_8}\,\exp\bigg(2x+4i\,\tan^{-1} \left(\omega_3/\omega_4\right)\bigg),
\quad 
r(x)=q(x)^*,
\end{equation}
where we recall that we use an asterisk to denote complex conjugation and we have defined
\begin{equation*}
\omega_3:=24\, e^{4x} (-216-3600\,e^{2x}-18675\, e^{4x}-18000\, e^{6x}-5000\,e^{8x}+12960000\,e^{20x}),
\end{equation*}
\begin{equation*}
\omega_4:=-1+1000\, e^{12x}\left[25920+62208\,e^{2x}+116640\,e^{4x}\\
+103680\,e^{6x}+77760\,e^{8x}\right],
\end{equation*}
\begin{equation*}
\omega_5:=6+75\,e^{2x}+50\,e^{4x}+43200i\,e^{6x}+334800i\,e^{8x}+648000i\,e^{10x},
\end{equation*}
\begin{equation*}
\omega_6:=10000\,e^{12x}\left(99i+
36i\,e^{2x}-1296\,e^{4x}-1296\,e^{6x}-1296\,e^{8x}\right),
\end{equation*}
\begin{equation*}
\omega_7:=-i+5184\,e^{4x} +86400\,e^{6x}+448200\,e^{8x}+432000\,e^{10x}+
10000 (12+2592i)\,e^{12x},
\end{equation*}
\begin{equation*}
\omega_8:=1000\,e^{14x}\left( 62208i+116640i\,e^{2x}
+
103680i\,e^{4x}+77760i\,e^{6x}-311040\,e^{10x}\right).
\end{equation*}
In this example, the constant $\mu$ appearing in \eqref{2.20} and the transmission coefficients
$T(\zeta)$ and $\bar T(\zeta)$ are given by
\begin{equation*}
\mu=2\pi,\quad T(\zeta)=-\ds\frac{(\lambda+i)(\lambda+2i)(\lambda+3i)}{(\lambda-i)(\lambda-2i)(\lambda-3i)},\quad 
\bar T(\zeta)=-\ds\frac{(\lambda-i)(\lambda-2i)(\lambda-3i)}{(\lambda+i)(\lambda+2i)(\lambda+3i)}.
\end{equation*}
In Figure~2 we have the plots of the absolute values of the potentials $q$ and $r$
listed in \eqref{6.13}.

\end{example}

\begin{figure}[!h]
  \centering
\includegraphics[width=2.5in]{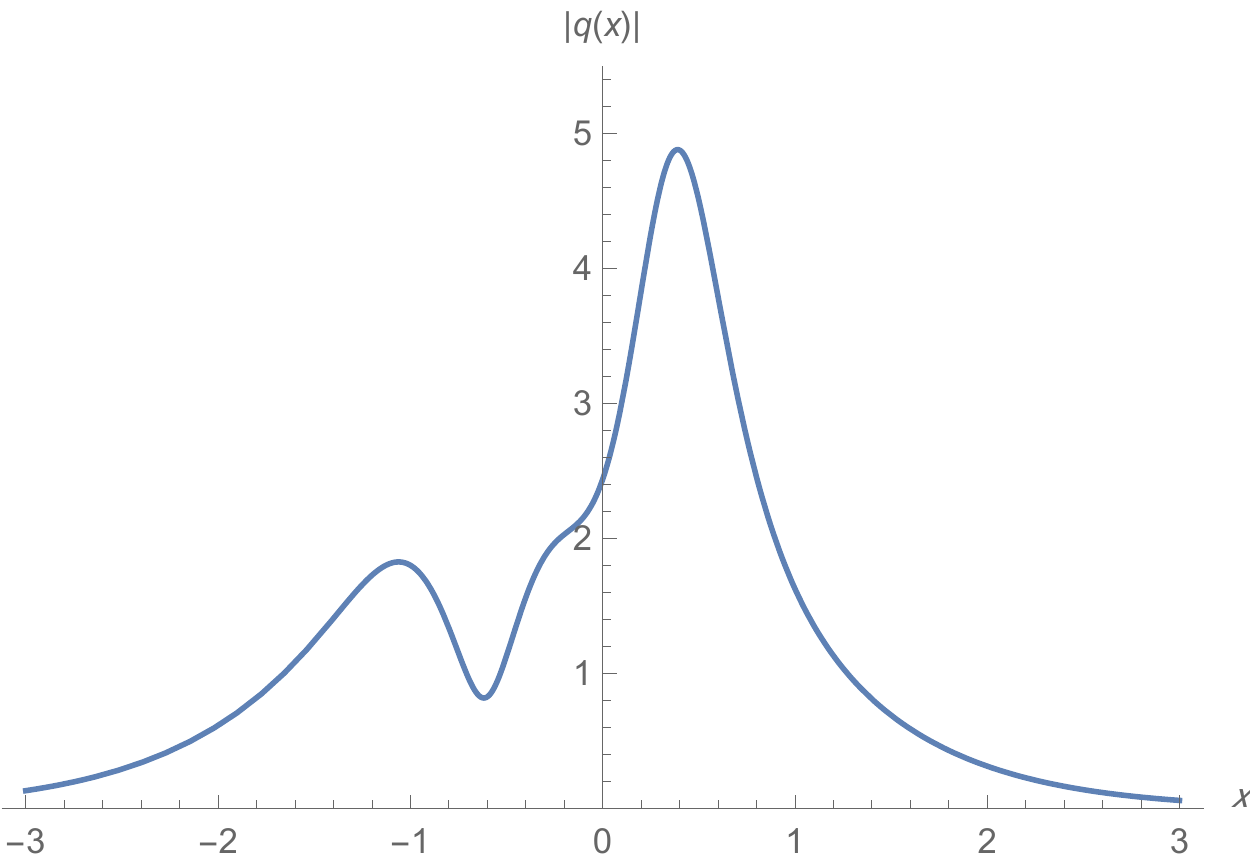}\hskip 0.5in
\includegraphics[width=2.5in]{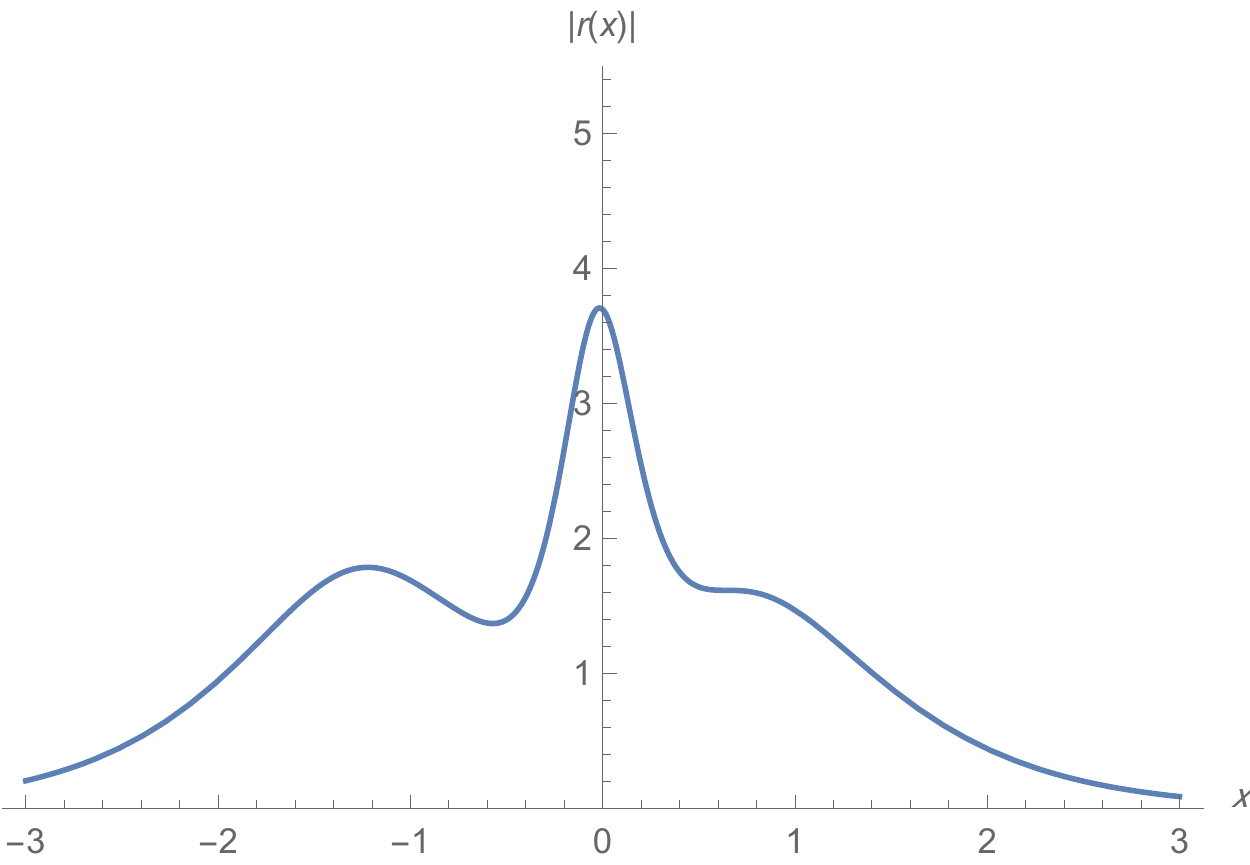}
  \caption{The absolute potentials $|q(x)|$ and $|r(x)|$ in Example~\ref{example6.5}.}
\end{figure}

In the next example, we illustrate Theorem~\ref{theorem5.2} by using
a pair of matrix triplets corresponding to two bound states each with multiplicity two.

\begin{example}
\label{example6.5}
\normalfont
Consider the reflectionless scattering data with two
double bound states described by the matrix triplets $(A,B,C)$ and $(\bar A,\bar B,\bar C)$ given by
\begin{equation}
\label{6.14}
A=\begin{bmatrix}i&1\\
0&i\end{bmatrix}, \quad B=\begin{bmatrix}0\\
1\end{bmatrix}, \quad
C=\begin{bmatrix}3&2\end{bmatrix},
\end{equation}
\begin{equation}
\label{6.15}
\bar A=\begin{bmatrix}-i&1\\
0&-i\end{bmatrix}, \quad \bar B=\begin{bmatrix}0\\ 1\end{bmatrix}, \quad
\bar C=\begin{bmatrix}2&3\end{bmatrix}.
\end{equation}
Using \eqref{6.14} and \eqref{6.15} in \eqref{5.19} and \eqref{5.20}, we get
\begin{equation}
\label{6.16}
q(x)=\ds\frac{32\,\omega_{11}\,\omega_{12}}{\omega_{13}+\omega_{14}}\,\exp\bigg(2x-2i\,\tan^{-1} \omega_9-2i\,\tan^{-1}\omega_{10}\bigg),
\end{equation}
\begin{equation}
\label{6.17}
r(x)=\ds\frac{8\,\omega_{15}\,\omega_{16}}{\omega_{17}+\omega_{18}}\,\exp\bigg(2x+2i\,\tan^{-1} \omega_9+2i\,\tan^{-1}\omega_{10}\bigg),
\end{equation}
where we have defined
\begin{equation*}
\omega_9:=\ds\frac{48\, e^{4x} (3+4x+8 x^2)}{-9+64\, e^{8x}+32\, e^{4x}(-2+5x)},
\quad
\omega_{10}:=\ds\frac{48\, e^{4x} (3+4 x+8 x^2)}{-9+64\, e^{8x}-16\, e^{4x}(9+10x)},
\end{equation*}
\begin{equation*}
\omega_{11}:=3-2i+6x+4 \,e^{4x} (3-4ix),
\end{equation*}
\begin{equation*}
\omega_{12}:=-9+64\, e^{8x}+16\, e^{4x}[-4+9i+(10+12i) x+24ix^2],
\end{equation*}
\begin{equation*}
\omega_{13}:=81+4096\, e^{16x}+288\, e^{4x} (9+10x)-2048 \,e^{12 x}(9+10x),
\end{equation*}
\begin{equation*}
\omega_{14}:=128\, e^{8x}(315+792\, x+1352\, x^2+1152\, x^3+1152\, x^4),
\end{equation*}
\begin{equation*}
\omega_{15}:=32\,e^{4x} (1+3ix)+9(2+3i)+36x),
\end{equation*}
\begin{equation*}
\omega_{16}:=-9+64\, e^{8x}-16(1+i)\, e^{4x}\,[9+(11+i)x+12(1+i) x^2],
\end{equation*}
\begin{equation*}
\omega_{17}:=81+4096\, e^{16x}+576\, e^{4x}(2-5x)+4096\, e^{12 x}(-2+5x),
\end{equation*}
\begin{equation*}
\omega_{18}:=128\, e^{8x}[185+272  x+1352  x^2+1152  x^3+1152  x^4].
\end{equation*}
In this example, we obtain the constant $\mu$ defined in \eqref{2.20} 
and the two transmission coefficients as
\begin{equation*}
\mu=0,\quad T(\zeta)=\left(\ds\frac{\lambda+i}{\lambda-i}\right)^2,\quad \bar T(\zeta)=\left(\ds\frac{\lambda-i}{\lambda+i}\right)^2,
\end{equation*}
where we recall that $\lambda=\zeta^2.$
In Figure~3 we present the plots of the absolute values of the potentials $q$ and $r$
given in \eqref{6.16} and \eqref{6.17}, respectively.

\end{example}

\begin{figure}[!h]
  \centering
\includegraphics[width=2.5in]{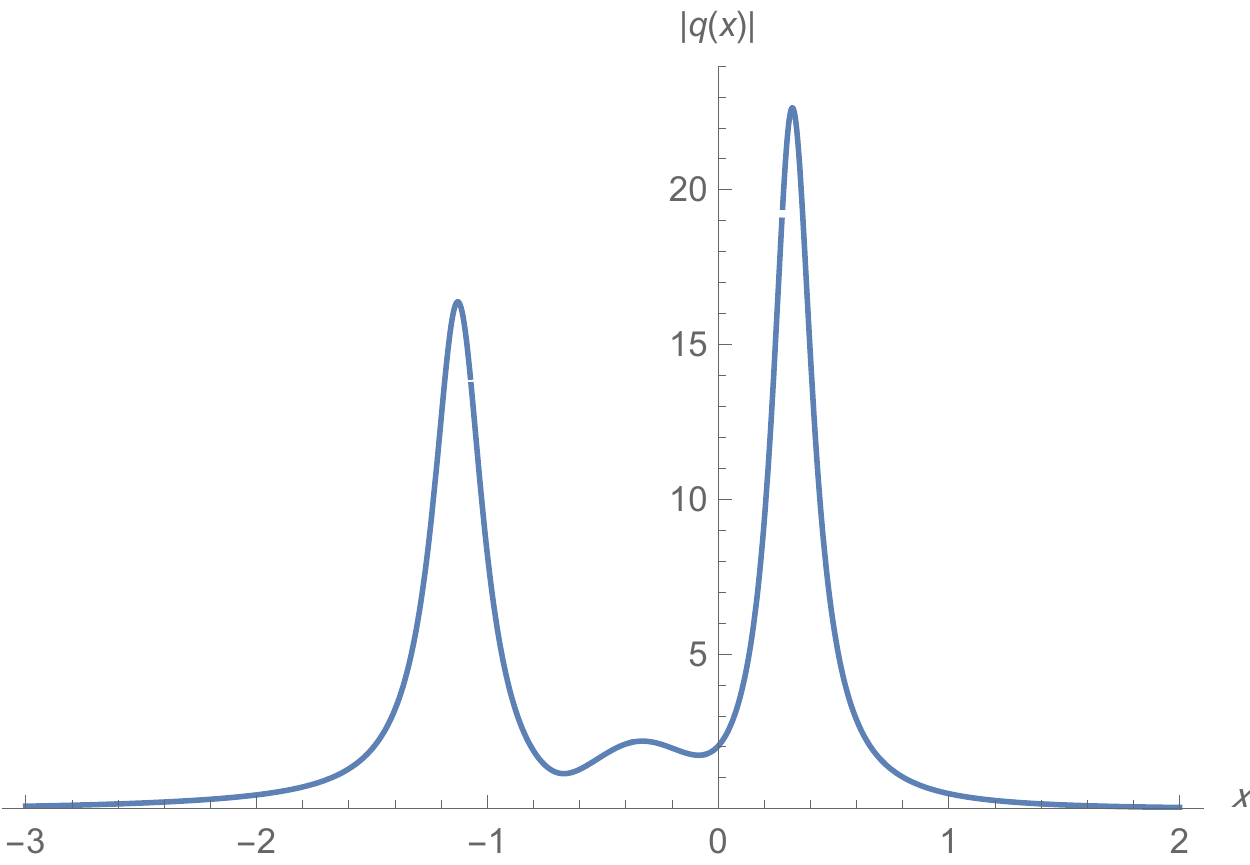}\hskip 0.5in
\includegraphics[width=2.5in]{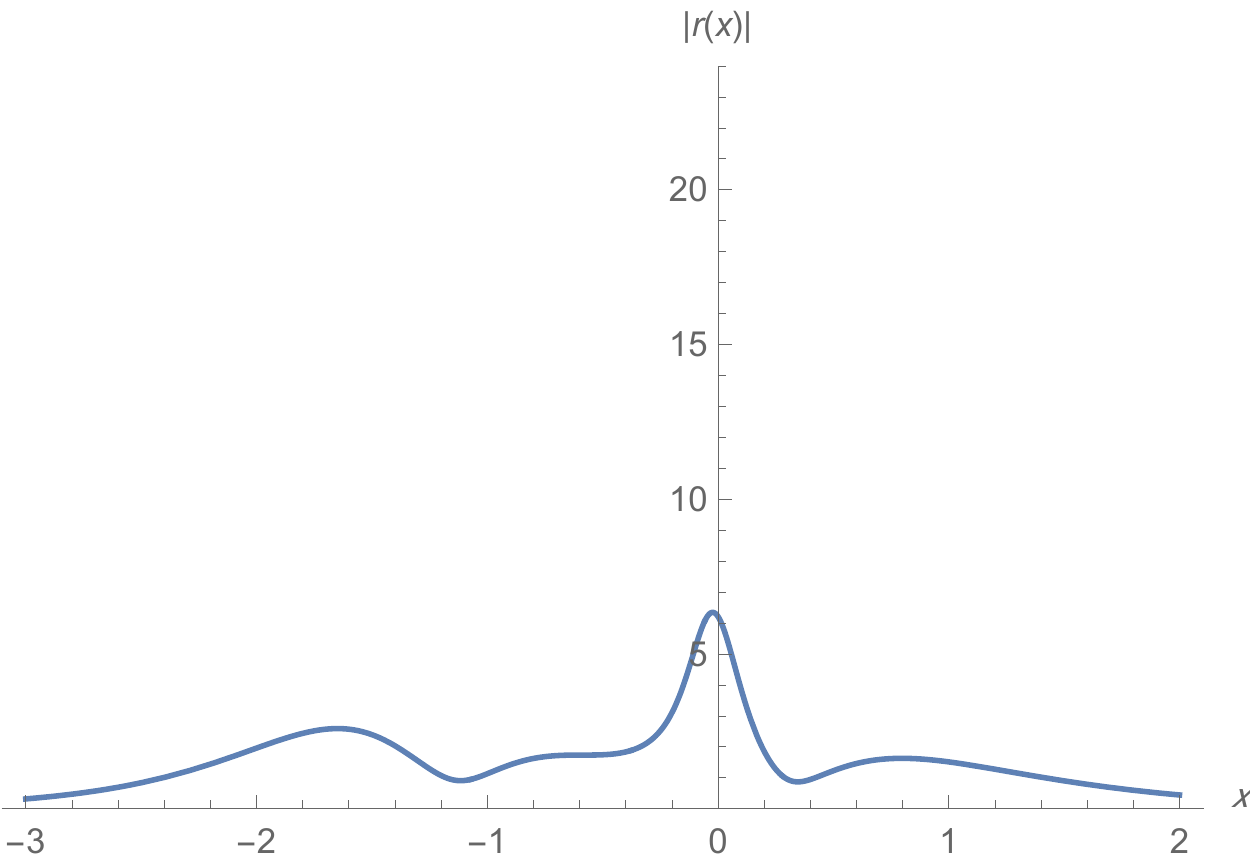}
  \caption{The absolute potentials $|q(x)|$ and $|r(x)|$ in Example~\ref{example6.6}.}
\end{figure}

In the next example, we illustrate Theorem~\ref{theorem5.2} by using
a pair of matrix triplets corresponding to one bound state of multiplicity two and
two simple bound states. 

\begin{example}
\label{example6.6}
\normalfont
Consider the reflectionless scattering data with two
double bound states described by the matrix triplets $(A,B,C)$ and $(\bar A,\bar B,\bar C)$ given by
\begin{equation}
\label{6.18}
A=\begin{bmatrix}i&1\\
0&i\end{bmatrix}, \quad B=\begin{bmatrix}0\\
1\end{bmatrix}, \quad
C=\begin{bmatrix}3&2\end{bmatrix},
\end{equation}
\begin{equation}
\label{6.19}
\bar A=\begin{bmatrix}-i&0\\
0&-2i\end{bmatrix}, \quad \bar B=\begin{bmatrix}1\\ 1\end{bmatrix}, \quad
\bar C=\begin{bmatrix}1&4\end{bmatrix}.
\end{equation}
Using \eqref{6.18} and \eqref{6.19} in \eqref{5.19} and \eqref{5.20}, we get
the corresponding potentials $q$ and $r$ as
\begin{equation}
\label{6.20}
q(x)=\ds\frac{8\,\omega_{21}\,\omega_{22}}{\omega_{23}+\omega_{24}}\,\exp\bigg(2x-2i\,\tan^{-1} \omega_{19}-2i\,\tan^{-1}\omega_{20}\bigg),
\end{equation}
\begin{equation}
\label{6.21}
r(x)=\ds\frac{8\,\omega_{25}\,\omega_{26}}{\omega_{27}+\omega_{28}}\,\exp\bigg(2x+2i\,\tan^{-1} \omega_{19}+2i\,\tan^{-1}\omega_{20}\bigg),
\end{equation}
where we have defined
\begin{equation*}
\omega_{19}:=\ds\frac{64\, e^{4x} +18\, e^{6x}}{-2+36\, e^{10x}-27\, e^{6x}(1+2x)-64\, e^{4x}(1+3x)},
\end{equation*}
\begin{equation*}
\omega_{20}:=\ds\frac{32\, e^{4x} +18\, e^{6x}}{-1+36\, e^{10x}+54\, x\,e^{6x}+16\, e^{4x}(-1+6x)},
\end{equation*}
\begin{equation*}
\omega_{21}:=2-2i+6x+36 \,e^{4x} +9\,e^{6x},
\end{equation*}
\begin{equation*}
\omega_{22}:=-1+36\, e^{10x}+18\, e^{6x}(i+3x)+16\, e^{4x}(-1+2i+6x),
\end{equation*}
\begin{equation*}
\omega_{23}:=1296\,e^{20 x}   +108\,e^{6 x} (1 + 2 x) - 1944\,^{16 x}(1 + 2 x) + 
 256\,e^{4 x} (1 + 3 x) - 4608 \,e^{14 x)}(1 + 3 x) ,
\end{equation*}
\begin{equation*}
\omega_{24}:=4+
 4096\,e^{8 x} (2 + 6 x + 9 x^2) + 81 \,e^{12 x} (13 + 36 x + 36 x^2) + 
 432\,e^{10 x} (13 + 40 x + 48 x^2),
\end{equation*}
\begin{equation*}
\omega_{25}:=128 i + 81 i \,e^{2 x} + 288 \,e^{6 x}(1 + 3 i x),
\end{equation*}
\begin{equation*}
\omega_{26}:=-2 + 36 \,e^{10 x} - 64 \,e^{4 x} (1 + i + 3 x) - 
 9 \,e^{6 x} (3 + 2 i+ 6 x),
\end{equation*}
\begin{equation*}
\omega_{27}:=2 + 2592 \,e^{20 x} - 216x\,e^{6 x}  + 7776 x\,e^{16 x}  +
 64 \,e^{4 x} (1 - 6 x) - 2304 \,e^{14 x} (1 -6 x),
\end{equation*}
\begin{equation*}
\omega_{28}:=648 \,e^{12 x} (1 + 9 x^2) + 512 \,e^{8 x} (5 - 12 x + 36 x^2) + 
 432 \,e^{10 x} (5 - 8 x + 48 x^2).
\end{equation*}
In this example, we obtain the constant $\mu$ defined in \eqref{2.20} 
and the transmission coefficients as
\begin{equation}
\label{6.22}
\mu=-2i\,\ln 2,\quad T(\zeta)=\ds\frac{(\lambda+i)(\lambda+2i)}{2(\lambda-i)^2},\quad \bar T(\zeta)=
\ds\frac{2(\lambda-i)^2}{(\lambda+i)(\lambda+2i)},
\end{equation}
where we recall that $\lambda=\zeta^2.$ As seen from \eqref{6.22}, $T(\zeta)$ has a double pole
at $\lambda=i$ and $\bar T(\zeta)$ has two simple poles at $\lambda=-i$ and $\lambda=-2i,$
respectively. Hence, we have one bound state of multiplicity two and two simple bound states.
In Figure~4 we present the plots of the absolute values of the potentials $q$ and $r$
given in \eqref{6.20} and \eqref{6.21}, respectively.

\end{example}

In the next example, we illustrate Theorem~\ref{theorem5.5} by using a pair of matrix triplets
with different sizes 
as input to the Marchenko system, demonstrating that
the corresponding potentials cannot both belong to the
Schwartz class.

\begin{example}
\label{example6.7}
\normalfont
Using the matrix triplet $(A,B,C)$ and $(\bar A,\bar B,\bar C)$ given by
\begin{equation*}
A=\begin{bmatrix}i&1&0\\
0&i&1\\
0&0&i\end{bmatrix}, \quad B=\begin{bmatrix}0\\
0\\
1\end{bmatrix}, \quad
C=\begin{bmatrix}1&1&1\end{bmatrix},
\end{equation*}
\begin{equation*}
\bar A=\begin{bmatrix}-i&1\\
0&-i\end{bmatrix}, \quad \bar B=\begin{bmatrix}0\\ 1\end{bmatrix}, \quad
\bar C=\begin{bmatrix}1&1\end{bmatrix},
\end{equation*}
as input in \eqref{5.19} and \eqref{5.20}, we obtain the corresponding potentials $q$ and $r$ as
\begin{equation}
\label{6.22}
q(x)=\ds\frac{32\,\omega_{31}\,\omega_{32}}{\omega_{33}
+\omega_{34}+\omega_{35}}\,\exp\bigg(2x-2i\,\tan^{-1} \omega_{29}-2i\,\tan^{-1}\omega_{30}\bigg),
\end{equation}
\begin{equation}
\label{6.23}
r(x)=\ds\frac{\omega_{36}\,\omega_{37}}{4(\omega_{38}+\omega_{39}
+\omega_{40})}\,\exp\bigg(-2x+2i\,\tan^{-1} \omega_{29}+2i\,\tan^{-1}\omega_{30}\bigg),
\end{equation}
where we have defined
\begin{equation*}
\omega_{29}:=\ds\frac{2+4x+32\, e^{4x} (7+8x+8 x^2)}{512\,e^{8x}+(1+2x)^2+64\,e^{4x}(-2+x+4x^2+8x^3)},
\end{equation*}
\begin{equation*}
\omega_{30}:=\ds\frac{-4(1+x)+32\, e^{4x} (3+8 x^2)}{3+8x+4x^2+512\, e^{8x}-32\, e^{4x}(7+10x+16 x^2+16x^3)},
\end{equation*}
\begin{equation*}
\omega_{31}:=1+2i-4(1-i)x-4x^2 +32 \,e^{4x} (i+2x),
\end{equation*}
\begin{equation*}
\omega_{32}:=2-i+4(1-i)x-4ix^2-512ie^{8x}+32e^{4x}[
7+4i+(8-2i)x+8(1-i)x^2-16i x^3],
\end{equation*}
\begin{equation*}
\omega_{33}:=(5+8x+4x^2)^2+262144\,e^{16x}-32768\,e^{12x}(7+10x+16x^2+16x^3),
\end{equation*}
\begin{equation*}
\omega_{34}:=
-64\,e^{4x}(33+98x+188x^2+248x^3+192x^4+64x^5),
\end{equation*}
\begin{equation*}
\omega_{35}:=
1024\,e^{8x}(61+148x+376x^2+544x^3+640x^4+512x^5+256x^6),
\end{equation*}
\begin{equation*}
\omega_{36}:=(2+i+2x)^2+512 e^{8x}-32(1-i) e^{4x}[2+5i+(4+12i)x^2+(1+i)(5x+8x^3)],
\end{equation*}
\begin{equation*}
\omega_{37}:=-1+4096e^{8x}(1+2ix-2x^2)+64e^{4x}[-6+5i+(8i-22)x+(8i-24)x^2-16x^3],
\end{equation*}
\begin{equation*}
\omega_{38}:=(1+2x)^2 (5+4x+4x^2)+262144\,e^{16x}+65536\,e^{12x}(-2+x+4x^2+8x^3),
\end{equation*}
\begin{equation*}
\omega_{39}:=128\,e^{4x}(5+15x+24x^2+44x^3+48x^4+32x^5),
\end{equation*}
\begin{equation*}
\omega_{40}:=2048\,e^{8x}(33+50x+60x^2+16x^3+96x^4+128x^5+128x^6).
\end{equation*}
In Figure~5 we present the plots of the absolute values of the potentials $q$ and $r$
given in \eqref{6.22} and \eqref{6.23}, respectively. From \eqref{6.22}
we observe that $q$ belongs to the Schwartz class. On the other hand, from the graph in Figure~5 
it is clear that $r$ cannot belong to the Schwartz class because $|r(x)|$ becomes
unbounded as $x\to-\infty.$
In this example, as $x\to-\infty$ we have 
\begin{equation*}
\omega_{36}=4x^2[1+o(1)],\quad \omega_{37}=-1+o(1),\quad  
\omega_{38}+\omega_{39}+\omega_{40}=16x^4[1+o(1)],
\end{equation*} 
and hence the term responsible for the blow up of $|r(x)|$ as $x\to-\infty$
is the term $e^{-2x}$ appearing in \eqref{6.23}.

\end{example}

\begin{figure}[!h]
  \centering
\includegraphics[width=2.5in]{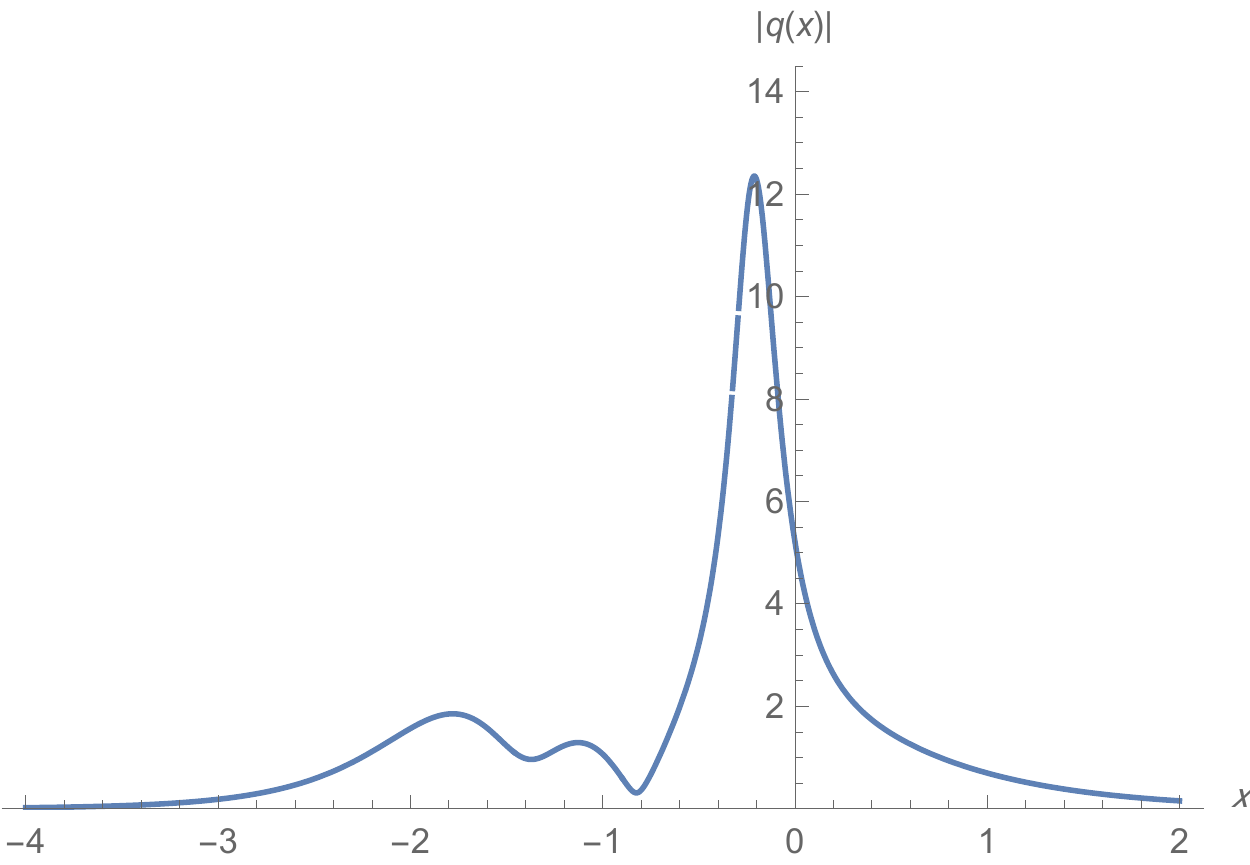}\hskip 0.5in
\includegraphics[width=2.5in]{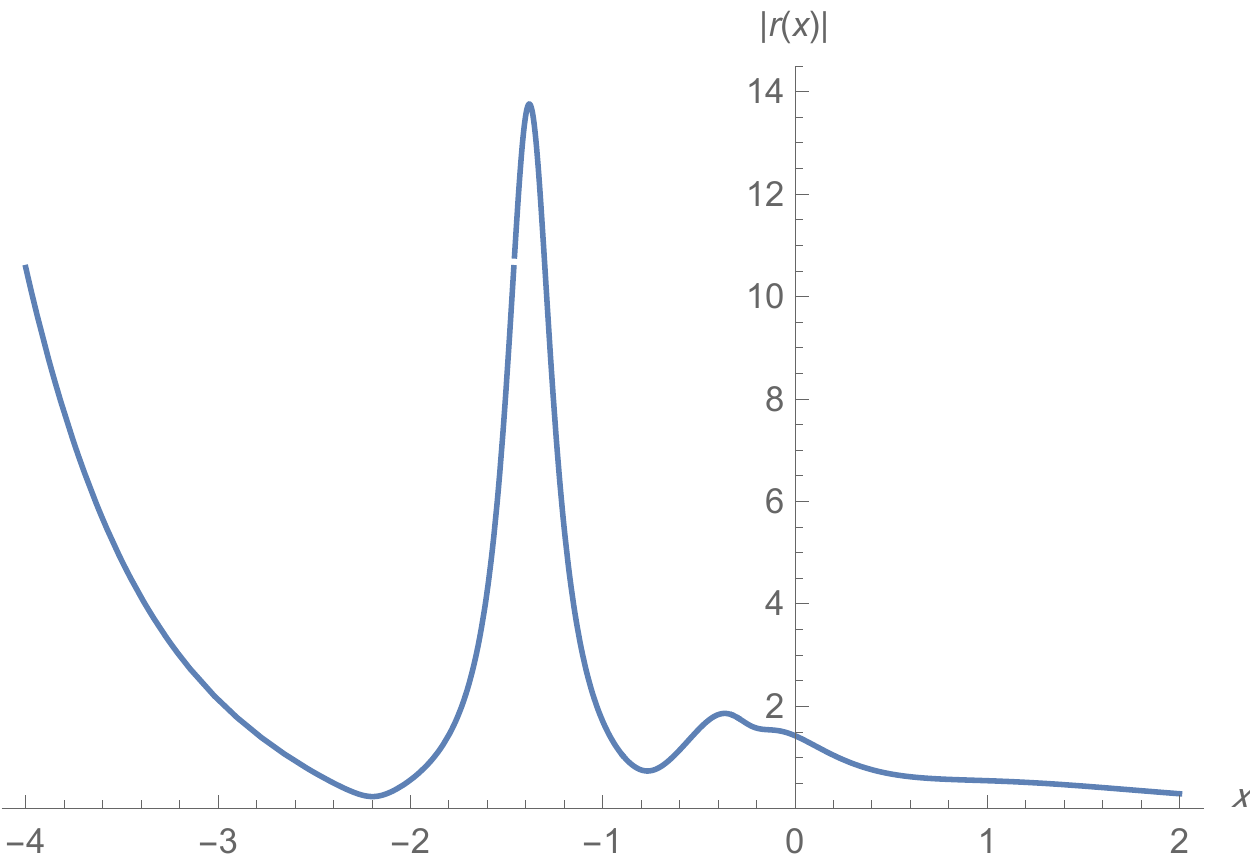}
  \caption{The absolute potentials $|q(x)|$ and $|r(x)|$ in Example~\ref{example6.7}.}
\end{figure}

\end{document}